\documentclass[letterpaper,twocolumn,10pt]{article}
\PassOptionsToPackage{hyphens}{url}\usepackage{hyperref}
\usepackage{usenix}
\usepackage[available,functional,reproduced]{usenixbadges}

\def\BibTeX{{\rm B\kern-.05em{\sc i\kern-.025em b}\kern-.08emT\kern-.1667em\lower.7ex\hbox{E}\kern-.125emX}}

\usepackage{algorithm}
\usepackage[noend]{algpseudocode}
\usepackage{boxedminipage}
\usepackage{framed}
\usepackage{placeins}
\usepackage{xpatch}
\xpatchcmd\algorithmic{\leftmargin\labelwidth}{\leftmargin+\labelsep}{}{}

\algrenewcommand\algorithmicrequire{\textbf{Input:}}
\algrenewcommand\algorithmicensure{\textbf{Output:}}

\usepackage{balance}
\usepackage{soul}
\usepackage{listings}
\usepackage{bbm}

\usepackage{makecell}

\usepackage{graphicx}
\usepackage{tabularx}
\usepackage{booktabs}
\usepackage[nointegrals]{wasysym}
\usepackage{url}
\usepackage{etoolbox}
\usepackage[numbers]{natbib}
\usepackage{xfrac}

\usepackage{amsmath}

\usepackage{enumitem,amssymb}
\usepackage{xspace}
\usepackage{glossaries}
\usepackage{graphicx}
\usepackage{amsmath}
\newlist{todolist}{itemize}{2}
\setlist[todolist]{label=$\square$}
\newlist{hybrid}{enumerate}{1}
\setlist[hybrid]{align=left, itemsep=2pt, topsep=8pt, leftmargin=12pt, label={$\text{Hyb}_{\arabic*}$}, ref={\arabic*}}
\newlist{algos}{itemize}{2}
\setlist[algos]{align=left,itemsep=2pt,left=0pt,label=•}
\newlist{protocol_steps}{enumerate}{2}
\setlist[protocol_steps]{align=left,itemsep=0pt,topsep=0pt,left=0pt,label={\arabic*.},ref={\arabic*}}
\newlist{requirements}{enumerate}{1}
\setlist[requirements]{label={\textbf{{RQ}$\mathbf{\arabic*}$}},leftmargin=0pt,align=left,labelsep=5pt,labelwidth=!,ref={\arabic*}}

\usepackage[makeroom]{cancel}

\usepackage{graphicx}
\usepackage{subcaption}
\usepackage{comment}
\usepackage{xspace}

\usepackage{multirow}
\usepackage{soul}

\usepackage{tikz}
\usetikzlibrary{positioning,arrows,shapes,fit}
\usepackage{subcaption}
\usepackage{fontawesome5}

\tikzset{
    state/.style={
        circle,
        draw=black,
        minimum width=1cm
    },
    client/.style={
        circle,
        draw=black,
        minimum width=0.5cm
    },
    party/.style={
        circle,
        draw=black,
        minimum width=0.5cm
    },
    lock/.style={
        fill=white,
        pos=0.65
    },
    whitebg/.style={
        midway,
        fill=white
    },
    protocol/.style={
        fill=white,
        draw=black
    },
    mpc_party/.style={
        circle,
        draw=black,
        minimum width=0.1cm
    },
}

\usepackage{tablefootnote}

\usepackage{amsthm}
\usepackage{amsfonts}
\usepackage{mathtools}
\usepackage[utf8]{inputenc}
\newcommand{\defeq}{\vcentcolon=}
\newcommand{\eqdef}{=\vcentcolon}

\DeclareMathAlphabet{\mathmybb}{U}{bbold}{m}{n}

\usepackage[keys]{cryptocode}
\usepackage{dashbox}

\createprocedureblock{procb}{center,boxed}{}{}{}

\theoremstyle{definition}
\newtheorem{definition}{Definition}[section]
\newtheorem{protocol}[definition]{Protocol}
\newtheorem{theorem}[definition]{Theorem}
\newtheorem{lemma}[definition]{Lemma}

\usepackage{color}
\usepackage{xcolor}

\usepackage{enumitem}
\setlist[itemize]{topsep=0pt, noitemsep}
\setlist[enumerate]{noitemsep, topsep=0pt}

\usepackage{colortbl}
\definecolor{Gray}{gray}{0.65}
\definecolor{LightGray}{gray}{0.9}

\definecolor{myred}{RGB}{178,34,34}
\definecolor{mygreen}{RGB}{107,142,35}
\definecolor{wifigray}{RGB}{200,200,200}
\definecolor{wifigreen}{RGB}{0,100,0}

\usepackage{array}
\usepackage{arydshln}
\usepackage{glossaries}
\glsdisablehyper

\newacronym{mpc}{MPC}{Multi-Party Computation}
\newacronym{fhe}{FHE}{Fully Homomorphic Encryption}
\newacronym{smpc}{SMPC}{Secure Multi-Party Computation}
\newacronym{2pc}{2PC}{2-party computation}
\newacronym{3pc}{3PC}{3-party computation}
\newacronym{4pc}{4PC}{4-party computation}
\newacronym{5pc}{5PC}{5-party computation}
\newacronym{sgd}{SGD}{stochastic gradient descent}
\newacronym{pca}{PCA}{principle component analysis}
\newacronym{god}{GOD}{guaranteed output delivery}
\newacronym{knn}{kNN}{k-Nearest-Neighbors}
\newacronym{mad}{MAD}{Median-Absolute-Deviation}
\newacronym{ml}{ML}{Machine Learning}
\newacronym{cl}{CL}{collaborative learning}
\newacronym{abb}{ABB}{arithmetic black-box}
\newacronym{soc}{SOC}{secure outsourced computation}
\newacronym{zkpok}{ZKPoK}{zero-knowledge proof-of-knowledge}
\newacronym{zkpot}{PoT}{proof-of-training}
\newacronym{zkpoi}{PoI}{proof-of-inference}
\newacronym{zkpoc}{PoC}{proof-of-consistency}
\newacronym{poc}{PoC}{Proof-of-Consistency}
\newacronym{pol}{PoL}{Proof-of-Learning}
\newacronym{mac}{MAC}{message authentication code}
\newacronym{lsss}{LSSS}{linear secret sharing scheme}
\newacronym{idabort}{ID-Abort}{Identifiable Abort}
\newacronym{edabit}{edaBit}{extended doubly authenticated bit}
\newacronym{msm}{MSM}{multi-scalar multiplication}

\newacronym{lime}{LIME}{Local Interpretable Model-Agnostic Explanations}
\newacronym{shap}{SHAP}{Shapley Additive Explanations}

\newacronym{dos}{DoS}{Denial-of-Service}
\newacronym{pki}{PKI}{public-key infrastructure}
\newacronym{sis}{SIS}{short integer solution}

\newacronym{ppml}{PPML}{Privacy-Preserving Machine Learning}
\newacronym{rtt}{RTT}{round-trip time}
\newacronym{lan}{LAN}{local area network}
\newacronym{wan}{WAN}{wide area network}
\newacronym{pvc}{PVC}{Pedersen Vector Commitment}
\newacronym{vss}{VSS}{Verifiable Secret Sharing}

\newglossaryentry{training}{name=training,description={}}
\newglossaryentry{deployment}{name=deployment,description={}}
\newglossaryentry{auditing}{name=auditing,description={}}

\newcommand{\role}[1]{#1}
\newglossaryentry{r:datasource}{name={\role{data source}},description={},plural={\role{data sources}}}
\newglossaryentry{r:inputparty}{name={\role{data owner}},description={},plural={\role{data owners}}}
\newglossaryentry{r:tcomputer}{name={\role{training computer}},description={},plural={\role{training computers}}}
\newglossaryentry{r:modelowner}{name={\role{model owner}},description={},plural={\role{model owners}}}
\newglossaryentry{r:icomputer}{name={\role{inference computer}},description={},plural={\role{inference computers}}}
\newglossaryentry{r:acomputer}{name={\role{audit computer}},description={},plural={\role{audit computers}}}
\newglossaryentry{r:auditrequester}{name={\role{client}},description={}}
\newglossaryentry{r:client}{name={\role{client}},description={}}

\newglossaryentry{r:participants}{name={party},description={},plural={parties}}
\newglossaryentry{audittrigger}{name={audit trigger},description={}}

\newglossaryentry{worldpublic}{name={\textit{super public}},description={}}
\newglossaryentry{public}{name={\textit{public}},description={}}
\newglossaryentry{private-entity}{name={\textit{private-entity}},description={}}
\newglossaryentry{secret}{name={\textit{secret}},description={}}

\newglossaryentry{req:privacy}{name={secrecy},description={}}
\newglossaryentry{req:authenticity}{name={input integrity},description={}}
\newglossaryentry{req:resilience}{name={resilience},description={}}

\newglossaryentry{consistencycheck}{name={consistency check},description={}}

\newglossaryentry{sc:adult}{name=\underline{(W1:Adult)},description={}}
\newglossaryentry{sc:mnist}{name=\underline{(W2:MNIST)},description={}}
\newglossaryentry{sc:cifar}{name=\underline{(W3:CIFAR-10)},description={}}
\newglossaryentry{sc:qnli}{name=\underline{(W4:QNLI)},description={}}

\newglossaryentry{f:knnshapley}{name={KNN-Shapley},description={}}
\newglossaryentry{f:robustness}{name={Robustness},description={}}
\newglossaryentry{f:fairness}{name={Fairness},description={}}
\newglossaryentry{f:shap}{name={Kernel-SHAP},description={}}
\newglossaryentry{f:camel}{name={Camel},description={}}

\newcommand{\ssGeneric}[1]{\ensuremath{\left[\,#1\,\right]}}
\newcommand{\ssIdeal}[1]{\ensuremath{[\![\,#1\,]\!]}}

\newcommand{\ssA}[1]{\ensuremath{\ssGeneric{#1}_{\sring_{\smod}}}}
\newcommand{\ssB}[1]{\ensuremath{\ssGeneric{#1}_{\sring_{\smod^\prime}}}}
\newcommand{\ssBit}[1]{\ensuremath{\ssGeneric{#1}_{\sring_{2}}}}
\newcommand{\ssField}[1]{\ensuremath{\ssGeneric{#1}_{\sfield}}}

\usepackage{stmaryrd}

\newcommand{\todoCameraReady}[1]{\textcolor{olive}{Camera ready submission: #1}}

\renewcommand{\todoCameraReady}[1]{}

\newcommand{\evalnum}[1]{#1}

\newcommand{\blue}[1]{{\color{blue}{#1}}}
\newcommand{\green}[1]{{\color{olive}{#1}}}

\newcommand{\shorten}[1]{}

\newcommand{\lsec}[1]{\label{sec:#1}}
\newcommand{\lfig}[1]{\label{fig:#1}}
\newcommand{\ltab}[1]{\label{tab:#1}}

\newcommand{\rsec}[1]{\S\ref{sec:#1}}
\newcommand{\rfig}[1]{Fig.~\ref{fig:#1}}
\newcommand{\rFig}[1]{Figure~\ref{fig:#1}}

\newcommand{\rtab}[1]{Table~\ref{tab:#1}}

\newcommand{\ldef}[1]{\label{def:#1}}
\newcommand{\rdef}[1]{Definition~\ref{def:#1}}
\newcommand{\llem}[1]{\label{lem:#1}}
\newcommand{\rlem}[1]{Lemma~\ref{lem:#1}}
\newcommand{\cfref}[1]{cf. #1}
\newcommand{\lstep}[1]{\label{step:#1}}
\newcommand{\rstep}[1]{Step~\ref{step:#1}}

\newcommand{\smodel}{\ensuremath{M}\xspace}

\newcommand{\scommitment}{\ensuremath{c}\xspace}

\newcommand{\scommitmentModel}{\ensuremath{\scommitment_{\smodel}}\xspace}
\newcommand{\scomtrainsets}{\ensuremath{\scommitment_{\strainset_1},\ldots,\scommitment_{\strainset_{\sNumParties}}}\xspace}
\newcommand{\scomtrainsetsconcat}{\ensuremath{\scommitment_{\strainset_1}\concat\ldots\concat\scommitment_{\strainset_{\sNumParties}}}\xspace}

\newcommand{\scomtrainseti}{\ensuremath{\scommitment_{\strainset_i}}\xspace}

\newcommand{\scommitmentX}{\ensuremath{\scommitment_{\sX}}\xspace}
\newcommand{\scommitmentY}{\ensuremath{\scommitment_{y}}\xspace}
\newcommand{\strainrand}{\ensuremath{J}\xspace}
\newcommand{\strainrandSpace}{\ensuremath{\mathcal{J}}\xspace}
\newcommand{\scomtrainrand}{\ensuremath{\scommitment_{\strainrand}}}
\newcommand{\ssig}{\ensuremath{\sigma}}
\newcommand{\sproofTrain}{\ensuremath{\pi_{\text{T}}}\xspace}
\newcommand{\sproofInference}{\ensuremath{\pi_{\text{I}}}\xspace}
\newcommand{\ssigTrain}{\ensuremath{\ssig_{\text{T}}}\xspace}
\newcommand{\ssigTrainingComputer}{\ensuremath{\ssig_{\sPartyServer}}\xspace}
\newcommand{\ssigInferenceComputer}{\ensuremath{\ssig_{\sPartyInferenceComputer}}\xspace}

\newcommand{\ssigReceipt}{\ssig_{\text{I}}\xspace}

\newcommand{\scomMessageSpace}{\ensuremath{\mathcal{M}}\xspace}
\newcommand{\scomRandomnessSpace}{\ensuremath{\mathcal{R}}\xspace}
\newcommand{\scomCommitmentSpace}{\ensuremath{\mathcal{C}}\xspace}
\newcommand{\scomRandomnessModel}{\ensuremath{r_{\smodel}}\xspace}

\newcommand{\scomRandomnessData}{\ensuremath{r_{\strainset}}\xspace}
\newcommand{\scomRandomnessDataI}{\ensuremath{r_{\strainset_i}}\xspace}
\newcommand{\scomRandomnessDatas}{\ensuremath{r_{\strainset_1},\ldots,r_{\strainset_{\sNumParties}}}\xspace}
\newcommand{\scomRandomnessTrainrand}{\ensuremath{r_{\strainrand}}\xspace}

\newcommand{\scomRandomnessX}{\ensuremath{r_{\sX}}\xspace}
\newcommand{\scomRandomnessY}{\ensuremath{r_{y}}\xspace}

\newcommand{\sComSetup}{\textsf{COM.Setup}\xspace}
\newcommand{\sComCommit}{\textsf{COM.Commit}\xspace}
\newcommand{\sComVerify}{\textsf{COM.Verify}\xspace}

\newcommand{\sPot}{\textsf{POT}\xspace}
\newcommand{\sPotSetup}{\textsf{POT.Setup}\xspace}

\newcommand{\sPotProve}{\textsf{POT.Prove}\xspace}
\newcommand{\sPotVerify}{\textsf{POT.Verify}\xspace}
\newcommand{\sPoI}{\textsf{POI}\xspace}
\newcommand{\sPoiSetup}{\textsf{POI.Setup}\xspace}

\newcommand{\sPoiProve}{\textsf{POI.Prove}\xspace}
\newcommand{\sPoiVerify}{\textsf{POI.Verify}\xspace}

\newcommand{\sSigSetup}{\textsf{SIG.Setup}\xspace}
\newcommand{\sSigSign}{\textsf{SIG.Sign}\xspace}
\newcommand{\sSigDistSign}{\textsf{SIG.DistSign}\xspace}
\newcommand{\sSigVerify}{\textsf{SIG.Verify}\xspace}

\newcommand{\sPoCLabel}{\textsf{PoC}\xspace}
\newcommand{\sPoCSetup}{\textsf{\sPoCLabel.Setup}\xspace}
\newcommand{\sPoCCommit}{\textsf{\sPoCLabel.Commit}\xspace}

\newcommand{\sPoCEvalNoLabel}{\textsf{Check}\xspace}
\newcommand{\sPoCEval}{\textsf{\sPoCLabel.\sPoCEvalNoLabel}\xspace}

\newcommand{\sPoCEvalID}{\ensuremath{\textsf{\sPoCLabel.Check}_\textsf{[ID]}}}

\newcommand{\sCCLabel}{\textsf{CC}\xspace}
\newcommand{\sCCSetup}{\textsf{\sCCLabel.Setup}\xspace}
\newcommand{\sCCCommit}{\textsf{\sCCLabel.Commit}\xspace}
\newcommand{\sCCEval}{\textsf{\sCCLabel.Check}\xspace}

\newcommand{\PCScheme}{\textsf{PC}}
\newcommand{\PCSetup}{\textsf{\PCScheme.Setup}}
\newcommand{\PCCommit}{\textsf{\PCScheme.Commit}}
\newcommand{\PCProve}{\textsf{\PCScheme.Prove}}
\newcommand{\PCCheck}{\textsf{\PCScheme.Check}}

\newcommand{\sInputs}{\ensuremath{\mathbf{x}}\xspace}

\newcommand{\sRandomness}{\ensuremath{r}\xspace}
\newcommand{\sProtocolRandomness}{\ensuremath{\omega}\xspace}

\newcommand{\PPT}{\textsf{PPT}\xspace}
\newcommand{\pp}{\ensuremath{\textsf{pp}}\xspace}

\newcommand{\pppot}{\ensuremath{\textsf{pp}_\textsf{pot}}\xspace}
\newcommand{\pppoi}{\ensuremath{\textsf{pp}_\textsf{poi}}\xspace}
\newcommand{\pppoc}{\ensuremath{\textsf{pp}_\textsf{poc}}\xspace}

\newcommand{\pkIPI}{\ensuremath{\pk_{\sPartyClient_i}}\xspace}
\newcommand{\skIPI}{\ensuremath{\sk_{\sPartyClient_i}}\xspace}
\newcommand{\pkTC}{\ensuremath{\pk_\sPartyServer}\xspace}

\newcommand{\skTCJ}{\ensuremath{\sk_{\sPartyServer_j}}\xspace}
\newcommand{\pkMH}{\ensuremath{\pk_{\sPartyModelHolder}}\xspace}

\newcommand{\pkMHK}{\ensuremath{\pk_{\sPartyModelHolder_k}}\xspace}
\newcommand{\skMHK}{\ensuremath{\sk_{\sPartyModelHolder_k}}\xspace}
\newcommand{\pkIC}{\ensuremath{\pk_{\sPartyInferenceComputer}}\xspace}

\newcommand{\skICJ}{\ensuremath{\sk_{\sPartyInferenceComputer_j}}\xspace}

\newcommand{\sX}{\ensuremath{x}\xspace}
\newcommand{\sXFeatureSize}{\ensuremath{l}\xspace}
\newcommand{\sY}{\ensuremath{y}\xspace}
\newcommand{\spredictionX}{\ensuremath{\tilde{x}}\xspace}
\newcommand{\spredictionY}{\ensuremath{\tilde{y}}\xspace}

\newcommand{\sdomainX}{\mathcal{X}}
\newcommand{\sdomainY}{\mathcal{Y}}

\newcommand{\sNumPartiesGenericParty}{\ensuremath{{N}}\xspace}
\newcommand{\sNumParties}{\ensuremath{{N_{\sPartyClient}}}\xspace}

\newcommand{\sNumServers}{\ensuremath{{N_{\sPartyServer}}}\xspace}
\newcommand{\sNumInferenceComputers}{\ensuremath{{N_{\sPartyInferenceComputer}}}\xspace}
\newcommand{\sNumAuditComputers}{\ensuremath{{N_{\sPartyAuditComputer}}}\xspace}
\newcommand{\sNumAuditors}{\ensuremath{{N_{\sPartyAuditor}}}\xspace}
\newcommand{\sNumModelHolders}{\ensuremath{{N_{\sPartyModelHolder}}}\xspace}
\newcommand{\sModelSize}{\ensuremath{m}\xspace}
\newcommand{\strainset}{D\xspace}
\newcommand{\strainsetInputs}{\mathbf{D}\xspace}

\newcommand{\strainsets}{\ensuremath{\strainset_1, \ldots, \strainset_{\sNumParties}}}

\newcommand{\sauditfunctionoutput}{\ensuremath{o}\xspace}
\newcommand{\spredictionInputs}{\ensuremath{\mathbf{x}}\xspace}

\newcommand{\sauditInputs}{\ensuremath{\mathbf{a}}\xspace}

\newcommand{\real}{\textnormal{\texttt{Real}}}
\newcommand{\ideal}{\textnormal{\texttt{Ideal}}}
\newcommand{\advAux}{\ensuremath{\textsf{aux}_{\sadversary}}}
\newcommand{\aux}{\ensuremath{\textsf{aux}}}

\newcommand{\Cset}{\mathcal{C}}

\newcommand{\sNumSamples}{d\xspace}
\newcommand{\sfeature}{\phi}
\newcommand{\sAlgorithm}{\ensuremath{\mathcal{T}\xspace}}

\newcommand{\sprotocolCC}{\ensuremath{\Pi_\text{cc}}\xspace}
\newcommand{\sprotocolCCID}{\ensuremath{\Pi_\text{cc [ID]}}\xspace}

\newcommand{\sfairDistance}{d_\sdomainX}

\newcommand{\sexplainmodel}{g\xspace}
\newcommand{\sproximity}{\ensuremath{\vec{z}}\xspace}

\newcommand{\sabb}{\ensuremath{\mathcal{F}_{\text{ABB}}}\xspace}
\newcommand{\sabbid}{\ensuremath{\mathcal{F}_{\text{ABB [ID]}}}\xspace}

\newcommand{\sidealrand}{\ensuremath{\mathcal{F}_\text{RAND}}\xspace}

\newcommand{\sabbec}{\ensuremath{\mathcal{F}_{\text{ABB}}^{\text{[EC]}}}\xspace}

\newcommand{\sabbATrain}{\ensuremath{\sabb\textsf{.Train}_\sAlgorithm}\xspace}
\newcommand{\sabbAInf}{\ensuremath{\sabb\textsf{.Predict}}\xspace}
\newcommand{\sabbAAudit}{\ensuremath{\sabb\textsf{.Audit}}\xspace}
\newcommand{\sabbAAuditID}{\ensuremath{\sabbid\textsf{.Audit}}\xspace}

\newcommand{\sabbARand}{\ensuremath{\sabb\textsf{.RAND}}\xspace}

\newcommand{\sideal}{\ensuremath{\mathcal{F}}\xspace}
\newcommand{\sidealfun}{\ensuremath{\mathcal{F}_{\text{Arc}}}\xspace}

\newcommand{\sidealinputdata}{\texttt{InputData}\xspace}

\newcommand{\sidealoutput}{\texttt{Output}\xspace}
\newcommand{\sidealaudittrigger}{\texttt{Audit}\xspace}
\newcommand{\sidealdeliver}{\texttt{Deliver}\xspace}
\newcommand{\sidealabort}{\texttt{Abort}\xspace}
\newcommand{\sidealinference}{\texttt{Predict}\xspace}

\newcommand{\sidealpki}{\ensuremath{\mathcal{F}_{\text{PKI}}}}
\newcommand{\sidealcrs}{\ensuremath{\mathcal{F}_{\text{CRS}}}}
\newcommand{\sideallist}{\ensuremath{L_\text{P}}\xspace}
\newcommand{\sidealmodelstore}{\ensuremath{L_\text{M}}\xspace}
\newcommand{\sidealmalicious}{\texttt{Malicious}\xspace}
\newcommand{\sidealbroadcast}{\ensuremath{\mathcal{F}_{\text{BC}}}\xspace}
\newcommand{\sidealmodelid}{\ensuremath{\textsf{id}_\smodel}\xspace}

\newcommand{\sadversary}{\ensuremath{\mathcal{A}}\xspace}

\newcommand{\sresponse}{\ensuremath{a}}
\newcommand{\ssim}{\ensuremath{\mathcal{S}}\xspace}

\newcommand{\sauditfunction}{\ensuremath{f_\texttt{audit}}}
\newcommand{\sauditfunctionSpace}{\ensuremath{F_\texttt{audit}}}
\newcommand{\spartyinput}{D}
\newcommand{\sample}{\ensuremath{\overset{{\scriptscriptstyle\$}}{\leftarrow}}}
\newcommand{\sfield}{\ensuremath{\mathbb{F}}\xspace}
\newcommand{\sgroup}{\ensuremath{\mathbb{G}}\xspace}
\newcommand{\sring}{\ensuremath{\mathbb{Z}}\xspace}
\newcommand{\smod}{\ensuremath{M}\xspace}
\newcommand{\sprover}{\ensuremath{\mathbb{P}}\xspace}
\newcommand{\sverifier}{\ensuremath{\mathbb{V}}\xspace}

\newcommand{\sinputSamples}{\ensuremath{\mathbf{x}}}
\newcommand{\sinputSample}{\ensuremath{x}}

\newcommand{\concat}{\ensuremath{\mathbin\Vert}}
\newcommand{\sdecomp}{\texttt{bitdec}\xspace}
\newcommand{\scomp}{\texttt{bitcom}\xspace}
\newcommand{\spolynomial}{\ensuremath{g}\xspace}
\newcommand{\sgenerator}{\ensuremath{h}\xspace}

\newcommand{\sPartySet}{\ensuremath{\mathcal{P}}\xspace}
\newcommand{\sPartyAll}{\ensuremath{\texttt{P}}\xspace}
\newcommand{\sPartyClient}{\ensuremath{\texttt{DH}}\xspace}
\newcommand{\sPartyServer}{\ensuremath{\texttt{TC}}\xspace}
\newcommand{\sPartyInferenceComputer}{\ensuremath{\texttt{IC}}\xspace}
\newcommand{\sPartyModelHolder}{\ensuremath{\texttt{M}}\xspace}
\newcommand{\sPartyAuditor}{\ensuremath{\texttt{C}}\xspace}
\newcommand{\sPartyAuditComputer}{\ensuremath{\texttt{AC}}\xspace}
\newcommand{\sPartyCorrupted}{\ensuremath{M_\sPartyAll}\xspace}

\newcommand{\sprotocol}{\ensuremath{\Pi}\xspace}
\newcommand{\sprotocolcamel}{\ensuremath{\sprotocol_\text{Arc}}\xspace}

\newcommand{\cPlainInf}[1]{\blue{#1}}
\newcommand{\cPlainTrain}[1]{\green{#1}}

\renewcommand{\secparam}{\ensuremath{\smash{1^\lambda}}} %

\newcommand{\citeextended}{\ifdefined\isnotextended~\cite{arc_full}\fi\xspace}

\newglossaryentry{a:ped}{name={$\sPoCLabel_\textsf{PED}$},description={}}
\newglossaryentry{a:sha3}{name={$\sPoCLabel_\textsf{SHA3}$},description={}}

\newcommand{\figcaptionvspace}{\vspace{-0.5em}} %
\renewcommand{\figcaptionvspace}{}

\newcommand{\subsecspacingtop}{}
\newcommand{\subsecspacingbot}{\vspace{0pt}}
\newcommand{\fakeparspacingtop}{\vspace{0.35em}}
\renewcommand{\fakeparspacingtop}{}
\newcommand{\figcaptionvspacebot}{}

\usepackage{flushend}

\usepackage{upquote}

\definecolor{darkslategray}{rgb}{0.18, 0.31, 0.31}

\newcommand{\fakeparagraph}[1]{\vskip 0pt\noindent\textbf{#1 }}
\newcommand{\indentparagraph}[1]{\textbf{#1 }}

\newcommand{\oursystem}{Arc\xspace}
\usepackage[nolist]{acronym}
\usepackage{ulem}

\renewcommand{\emph}[1]{\textit{#1}}

\usepackage{environ}

\ifdefined\isnotextended

\NewEnviron{myhideenv}{%
  \setbox0\vbox{%
   \let\xxwrite\write
   \protected\def\write{\immediate\xxwrite}%
   {\tiny XX\BODY XX}}
}

\else
    \newenvironment{myhideenv}{}{}
\fi

\def\isnotanon{1}

\newcommand{\wifi}[1][4]{%
    \colorlet{col1}{wifigray}%
    \colorlet{col2}{wifigray}%
    \colorlet{col3}{wifigray}%
    \colorlet{col4}{wifigray}%
    \ifnum0<#1
    \colorlet{col1}{wifigreen}%
    \fi%
    \ifnum1<#1
    \colorlet{col2}{wifigreen}%
    \fi%
    \ifnum2<#1
    \colorlet{col3}{wifigreen}%
    \fi%
    \ifnum3<#1
    \colorlet{col4}{wifigreen}%
    \fi%
    \begin{tikzpicture}[rounded corners=0.05ex]
        \fill[col1] (0ex,0ex) rectangle ++(0.5ex,0.9ex);
        \fill[col2] (0.7ex,0ex) rectangle ++(0.5ex,1.05ex);
        \fill[col3] (1.4ex,0ex) rectangle ++(0.5ex,1.20ex);
        \fill[col4] (2.1ex,0ex) rectangle ++(0.5ex,1.35ex);
    \end{tikzpicture}%
}

\begin{document}

\title{
	\Large \bf Holding Secrets Accountable: Auditing Privacy-Preserving Machine Learning}

\patchcmd{\maketitle}
{\@maketitle}
{\@maketitle}%

{}
{}

\ifdefined\isnotanon
\author{{\rm Hidde Lycklama\textsuperscript{1}, Alexander Viand\textsuperscript{2}, Nicolas K\"uchler\textsuperscript{1}, Christian Knabenhans\textsuperscript{3}, Anwar Hithnawi\textsuperscript{1}}  \\
\\
{\textsuperscript{1}\textit{ETH Zurich} \ \textsuperscript{2}\textit{Intel Labs} \ \textsuperscript{3}\textit{EPFL}}}
\else
\author{Paper \#234, 14 pages + references + appendix}
\fi

\date{}

\maketitle

\begin{abstract}
Recent advancements in privacy-preserving machine learning are paving the way to extend the benefits of ML to highly sensitive data that, until now, has been hard to utilize due to privacy concerns and regulatory constraints.
Simultaneously, there is a growing emphasis on enhancing the transparency and accountability of ML,
including the ability to audit deployments for aspects such as fairness, accuracy and compliance.
Although ML auditing and privacy-preserving machine learning have been extensively researched, they have largely been studied in isolation. 
However, the integration of these two areas is becoming increasingly important.
In this work, we introduce \oursystem, an MPC framework designed for auditing privacy-preserving machine learning.
\oursystem cryptographically ties together the training, inference, and auditing phases to allow robust and private auditing.
At the core of our framework is a new protocol for efficiently verifying inputs against succinct commitments.
We evaluate the performance of our framework when instantiated with our consistency protocol and compare it to hashing-based and homomorphic-commitment-based approaches, demonstrating that it is up to \evalnum{$10^4\times$} faster and up to \evalnum{$10^6\times$} more concise.

\end{abstract}

\subsecspacingtop
\section{Introduction}
\subsecspacingbot
\lsec{intro}

Mounting concerns regarding security and privacy in \gls{ml} have spurred interest in \gls{ppml}~\cite{NSTC2016AI,Biden2023EO}. 
These developments aim to 
address concerns related to user data, whether during inference or training, as well as securing ML models, as organizations seek to maintain a competitive advantage by keeping them confidential. 
Consequently, secure inference and secure training frameworks have emerged to address various security and privacy concerns inherent in using and training machine learning models~\cite{Kim2023-fhe_imagenet, Juvekar2018-bb,Koti2021-Swift,Dalskov2021-rc,Mohassel2018-gy,Wagh2021-Falcon}.
The majority of these frameworks rely on secure computation techniques~\cite{Keller2020-mpspdz,Keller2022-quantizedtraining,Barak2020-quantizedinference}, which offer security guarantees by hiding the data and/or the model during computation.
While these techniques are effective in achieving the intended security goals, they also introduce new challenges due to their inherent opacity.
To achieve secrecy, these technologies conceal the processes of training and inference, making it difficult to audit \gls{ppml} pipelines for fairness, transparency, accountability, and other (often legally mandated) objectives.
However, a variety of scenarios require both the privacy guarantees of \gls{ppml} and the ability to audit the \gls{ml} pipeline.
For example, banks interested in collaborating on training better risk assessment algorithms (e.g., for creditworthiness) will need to rely on \gls{ppml} techniques to protect both customer privacy and commercial information.
At the same time, they must also fulfill legal requirements for auditability such as those of the recently enacted EU AI Act~\cite{EU2021_AIACT}, which specifically calls out creditworthiness evaluation~\cite{EUAIAbriefing}.
While auditability may seem to be in direct conflict with the privacy requirements of \gls{ppml},
secure computation can, in principle, offer a way forward for verifying these properties while preserving privacy.
However, realizing this in an efficient and robust manner is challenging.

\paragraph{Private \gls{ml} Auditing.}
\gls{ml} auditing involves the examination and verification of machine learning models, algorithms, and data to ensure accountability and desired properties such as fairness, transparency, and accuracy during deployment.
Approaches to \gls{ml}  auditing
can be divided into a priori and post-hoc auditing mechanisms~\cite{Birhane2024-iy}.
The former focus on pre-deployment verification techniques that act as predefined sets of verifications on the model, data, or training process~\cite{Jia2021-PoL,Choi2023-zb}, such as 
model and data validation tests~\cite{Chang2023-mk,Lycklama2023-cx} or robustness~\cite{Drews2019-cj} and fairness~\cite{Bastani2018-hr} verification.
While covering important use cases of auditing, these remain limited to known prior issues, which, in the case of ML, are hard to exhaustively address given the black box nature of ML.
Post-hoc audits, which are triggered in response to detecting undesirable behavior or other triggers, are therefore essential to ensure accountability in real-world deployments of ML~\cite{Shan2021-ad,Hendrycks2021-hx,Shamsabadi2022-us}. 
For example, individuals may seek an explanation of a decision to mitigate potential harm or to investigate its fairness~\cite{Lundberg2017-unifiedmodelpredictions,Ribeiro2016-lime}.
Recent efforts have examined the realization of a priori-auditing techniques in secure settings using a variety of ad-hoc techniques.
This includes work for verifying robustness, verifiable fairness, and model and data validation techniques~\cite{Jovanovic2022-tz,Kilbertus2018-zl,Segal2020-tu,Chang2023-mk}.
Post-hoc audits, however, have received scant attention in the secure setting.
Due to their on-demand nature, they present a unique set of challenges that a priori audits do not face.
In this paper, we, therefore, focus primarily on achieving secure post-hoc audits for \gls{ppml}.

\fakeparspacingtop
\fakeparagraph{Secure Post-Hoc Audits.}
In current practice, auditing of \gls{ppml} systems generally requires assuming a trusted third party (which can be granted access to training data, models, and predictions) that applies traditional auditing solutions. 
However, in addition to undermining the privacy-preserving nature of \gls{ppml}, even a trusted auditor is not sufficient to achieve robust audits.
Specifically, parties might inadvertently or maliciously alter their inputs to the auditing phase so that they no longer match their original inputs to the \gls{ppml} system, distorting the results of the auditing phase.
Instead, the auditor would need visibility into the entire training and inference process to ensure the consistency of the audit.
One might consider realizing such a trusted auditor cryptographically, by relying on (maliciously secure) \gls{mpc}  for the entire pipeline. %
In practice, however, it is generally not feasible to continuously run large MPC deployments with many parties (e.g., different clients receiving inferences and/or different auditing parties).
This is because MPC, in general, scales extremely poorly in the number of involved parties, %
and due to the complexities of maintaining (and periodically refreshing) a large amount of secret state over extended periods~\cite{Herzberg1995-pss,Schultz2010-mpss,Maram2019-ao}.
Note that \gls{ppml} systems usually sidestep these issues, as training and inference can be realized as distinct phases.
As a result, the (usually significant) resources utilized for training do not need to be maintained in order to perform inferences. 
A practical approach to cryptographic auditing for \gls{ppml}, therefore, needs to maintain this decoupling while nevertheless ensuring consistent audits.

\vspace{1em}
\fakeparagraph{Contribution.}
This paper presents \oursystem, a new framework for privacy-preserving auditing of PPML systems.
\oursystem is highly modular and supports a wide range of efficient PPML approaches and auditing functions, and is the first framework to efficiently implement post-hoc auditing for PPML.
Our framework ties together training, inference, and auditing while maintaining consistency via the use of concise cryptographic \emph{receipts}.
\oursystem supports a wide range of PPML approaches, including mixed \mbox{secure/plaintext} settings which are common in practical deployments but pose significant challenges for auditing.
The overhead of our framework is primarily determined by the efficiency of the underlying consistency mechanism.
We first describe (and prove secure) our auditing protocol using a black-box definition of the consistency layer.
We then present a highly efficient instantiation that makes \oursystem practical for a wide range of PPML deployment scenarios.
Finally, we evaluate the performance of our framework when instantiated with our consistency protocol and compare it to hashing-based and homomorphic-commitment-based approaches, demonstrating that it is up to \evalnum{$10^4\times$} faster and up to \evalnum{$10^6\times$} more concise.

In the following, we discuss background and related work in~\rsec{background}.
We present the requirements of \gls{ppml} auditing systems and the design of our \gls{ppml} auditing framework in~\rsec{overview}.
In~\rsec{design:inputconsistency}, we formalize the \gls{poc} and present our consistency check protocol.
In~\rsec{functions}, we discuss how to realize common auditing functions under \gls{mpc}.
Finally, in~\rsec{eval}, we evaluate our framework and compare against related work.

\section{Background \& Related Work}
\subsecspacingbot
\lsec{background}

We briefly introduce relevant background for \gls{ppml} and \gls{ml} auditing, and then discuss related work.

\paragraph{\acrlong{ppml}.}
\gls{ppml} enables parties to securely train and deploy sensitive ML models in environments that involve untrusted or potentially compromised entities. 
There has been significant progress in \gls{ppml} in recent years, leveraging advanced cryptographic techniques to ensure data privacy and model integrity~\cite{Kim2023-fhe_imagenet, Juvekar2018-bb,Koti2021-Swift,Byali2020-Flash,Koti2022-Tetrad,Dalskov2021-rc,Mohassel2018-gy,Wagh2021-Falcon,Barak2020-quantizedinference}.
Approaches that rely on \gls{mpc} typically offer the best performance by distributing trust among $n$ parties. These parties collaboratively execute training or 
inference computations, all while preserving the privacy of each party's inputs. 
Protocols are categorized based on the number of parties ($t$) an adversary can corrupt without breaching security, with distinctions made between a majority of honest 
parties ($t < \frac{n}{2}$) and a dishonest majority ($t < n$).
Moreover, protocols are designed to withstand different adversarial behaviors, ranging from passive corruption,
where compromised parties may collude to learn information while following the protocol honestly, to active corruption, allowing adversaries to deviate from the protocol 
arbitrarily.
As of today, the most efficient \gls{mpc} protocols for \gls{ppml} rely on homomorphic secret sharing over a field $\sfield_p$ or ring $\sring_{2^k}$~\cite{Barak2020-quantizedinference}.
This allows them to perform integer arithmetic by adding and scaling shares using the homomorphism of the scheme. 
Communication among parties is only required during the multiplication of shares.
\gls{ppml} frameworks frequently also offer higher-order primitives essential for machine learning, such as dot products, comparisons, bit extraction, exponentiation, and truncation~\cite{Koti2021-Swift,Dalskov2021-rc,Mohassel2018-gy,Wagh2021-Falcon,Barak2020-quantizedinference,Keller2022-quantizedtraining}.
Different functionalities might be implemented most efficiently in different fields or rings, in which case we can use share conversion to switch between them, e.g., ring and field-based MPC. 
\ifdefined\isnotextended
In Appendix~\ref{sec:share_conversion} in the extended version of the paper~\cite{arc_full}, we discuss this technique in more detail.
\else
In Appendix~\ref{sec:share_conversion}, we discuss this technique in more detail.
\fi

\paragraph{ML Auditing.}

Auditing of \gls{ml} systems is an emerging field focused on enhancing the accountability of \gls{ml} algorithms.
Auditing involves verifying the compliance of organizations\textquotesingle\xspace \gls{ml} models with safety and legal standards, e.g., ensuring they do not infringe on copyright laws. Here, we refer to auditing techniques that analyze an algorithm to 
offer further insights or assurances regarding the model and its predictions. 
This includes efforts to enhance transparency by explaining predictions, ensuring fairness, 
or providing accountability for the contributions of different parties.
Depending on the technique, these algorithms may require access to the training data, the model, the prediction, or a combination thereof. 
Techniques that involve only the training data and the model can often be conducted a priori as part of an internal quality assurance process. 
However, such a priori techniques are inherently limited:
important classes of auditing techniques fundamentally require, or only become practical with, access to the prediction sample.
For example, audits for explainability/accountability or fairness, respectively.
In addition, due to the nature of ML, we can frequently only identify a prediction as unwanted after the fact, potentially even only after significant time has passed.
As a result, we require the ability to perform \emph{post-hoc} auditing in real-world ML deployments.

A wide range of approaches for post-hoc auditing have been proposed in the literature, many of which build upon similar techniques.
For example, many algorithms in this space rely on perturbing input data or prediction features to assess the impact of such changes on the model's behavior, 
effectively treating the model as a black box. 
These methods find application in a variety of contexts, such as providing explanations for predictions~\cite{Lundberg2017-unifiedmodelpredictions,Ribeiro2016-lime}, 
investigating the model's training data for biased or poisoned samples~\cite{Wang2019-neuralcleanse,Lycklama2022-CamelMLSafety,Shan2021-ad}, 
or ensuring fairness by analyzing model predictions under hypothetical scenarios where specific input features are altered~\cite{Mukherjee2020-indivfairnessdata}. 
As evaluating these methods can be resource-intensive, alternative techniques employ propagation-based methods, which are more computationally efficient by assuming knowledge of the model's internal structure. 
These methods attribute importance to model neurons, input features, or training samples based on gradients or activations~\cite{Sundararajan2017-integratedgradients,Bach2015-lrp}.
These techniques share foundational computational operations with training and inference processes, such as forward passes through the neural network and 
backpropagation. 
This similarity in computational models implies that the protocols developed for training and inference can be repurposed, to some extent, for auditing purposes.
In~\rsec{functions}, we provide a detailed description of the algorithmic aspects of the auditing functions supported in our framework.

\begin{table}
    \centering
    \setlength{\tabcolsep}{3pt}
    \begin{tabular}{lcccccccc}
        \toprule
        & Mal. Sec. & \textsf{T} & \textsf{M} & \textsf{I} & \textsf{Co} & \textsf{Ba} & \textsf{St} \\
        \midrule
        Phoenix~\cite{Jovanovic2022-tz}  & $\times$ & \Circle & \CIRCLE & \CIRCLE & \wifi[2] & \wifi[3] & -- \\
        Agrawal et al.~\cite{Agrawal2021-rg} & $\times$ & \Circle & \CIRCLE & \Circle & \wifi[2] & \wifi[2] & \wifi[4] \\
        Kilbertus et al.~\cite{Kilbertus2018-zl} & $\checkmark$ & \Circle & \CIRCLE & \CIRCLE & \wifi[1] & \wifi[1] & \wifi[4] \\
        Segal et al.~\cite{Segal2020-tu}  & $\checkmark$ & \Circle & \CIRCLE & \Circle & \wifi[1] & \wifi[1] & \wifi[4]  \\
        Holmes~\cite{Chang2023-mk}  & $\checkmark$ & \CIRCLE & \Circle & \Circle & \wifi[4] & \wifi[4] & -- \\
        Cerebro~\cite{Zheng2021-tb} & $\checkmark$ & \CIRCLE & \Circle & \Circle & \wifi[1] & \wifi[3] & \wifi[1] \\
        Ours~(\rsec{overview}) & $\checkmark$ & \CIRCLE & \CIRCLE & \CIRCLE & \wifi[4] & \wifi[4] & \wifi[4] \\
        \bottomrule
    \end{tabular}
    \vspace{0.25em}
    \caption{Related work covers different subsets of the \gls{ppml} pipeline by allowing to audit combinations of the training data (\textsf{T}), the model (\textsf{M}) and the inference (\textsf{I}), and have different overheads for compute (\textsf{Co}), bandwidth (\textsf{Ba}) and storage (\textsf{St}). 
    }
    \ltab{related_work}
\end{table}

\paragraph{Related Work.}
While this is, to the best of our knowledge, the first framework for (post-hoc) PPML auditing, our work is closely related to efforts aimed at enhancing the reliability of PPML systems. 
Thus, we briefly discuss the most relevant related work here.
Prior research primarily focuses on narrow aspects, enhancing isolated components and instantiations of the \gls{ppml} pipeline as shown in~\rtab{related_work}.
Phoenix integrates randomized smoothing techniques into \gls{fhe}-based \gls{ml} inference to guarantee robust and fair model predictions~\cite{Jovanovic2022-tz}.
Holmes improves the quality of \gls{mpc} training to conduct distribution tests on training data via efficient interactive zero-knowledge proofs before training starts~\cite{Chang2023-mk}.
These works apply and optimize reliability techniques to \gls{ppml} inference and training but do not allow for retroactive auditing of predictions or training data.

Another line of work enables retroactive verification of certain properties, but only for specific components of the PPML pipeline in isolation.
Fairness certification allows clients to verify that their private predictions were generated by a certified model~\cite{Kilbertus2018-zl,Segal2020-tu,Agrawal2019-bu}.
This is usually achieved by having a regulator sign a hash-based commitment of the model.
Cerebro~\cite{Zheng2021-tb} extends \gls{mpc} training by enabling an auditor to conduct post-hoc computation on parties' inputs through a consistency check involving cryptographic commitments. 
However, their system only allows auditing of parties' datasets individually, which significantly limits the scope of auditing. %
Additionally, the commitment techniques they employ to ensure the integrity of the training data do not scale to a complete \gls{ppml} system handling large amounts of training data and potentially many clients.

\section{\oursystem Design}
\subsecspacingbot
\lsec{overview}

\begin{figure*}[t]
	\centering
	\includegraphics[width=1.0\textwidth]{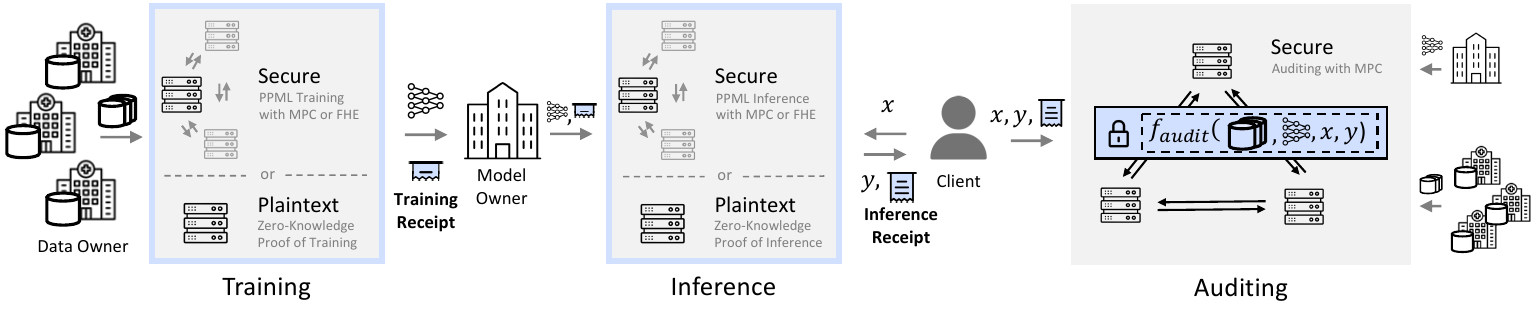}
\caption{Overview of \oursystem, which augments existing \acrshort{ppml} pipelines with an \acrshort{mpc} auditing phase to execute auditing functions.
\Glspl{r:auditrequester} receive a receipt that can later be used to verify the consistency of the training data, model and prediction under audit.}
    \lfig{scenario}
\end{figure*}

\oursystem enables private and secure auditing for existing \gls{ppml} systems by providing end-to-end consistency of data, model, and predictions while preserving the benefits of distinct training, inference, and auditing phases.
Our framework is highly modular and supports a wide range of inference and training approaches, including deployments that mix secure computation and plaintext computation.
For example, many scenarios permit the release of a differentially private model after a secure training phase, or consider secure inference for a centrally trained model.
Our framework allows the execution of arbitrary auditing functions over the original training data, model, and predictions using secure multi-party computation.
Notably, these audits can be conducted post-hoc, i.e., long after the training or inference phases.
In~\rfig{scenario}, we illustrate how \oursystem extends a typical PPML pipeline. 
Specifically, \oursystem augments the training phase with a consistency layer that generates a concise cryptographic receipt linking the model received by the \glspl{r:modelowner} and the training data provided by the \glspl{r:inputparty}. 
During inference, \oursystem{}’s consistency layer extends this receipt to include the \gls{r:auditrequester}'s prediction sample and result. 
Using this receipt, the \gls{r:auditrequester} can later request audit functions to be executed  
and verify that the \glspl{r:inputparty} and \glspl{r:modelowner} provided their original inputs.
This design allows our framework to scale independently of the number of \glspl{r:auditrequester} and predictions, as the only additional state necessary is the (concise) prediction held by the \gls{r:auditrequester}.

\fakeparagraph{Threat Model.}
We consider an actively malicious adversary that can (statically) compromise parties across the training, inference, and auditing phases.
The adversary can observe and modify all inputs, states and network traffic of the parties it controls.
We assume that at least one party that provides inputs (i.e., training data, model, or prediction sample) or receives outputs is honest. %
Note that because not all parties are involved in each phase, it is possible that all parties interacting in a phase are malicious.
If a phase involves secure computation, we assume at least one of the computational parties is honest.
Note that, certain instantiations of secure computation might impose additional constraints on the adversary.
For example, \oursystem can be used with MPC protocols that assume an honest majority of computing parties,
which are frequently significantly more efficient than their dishonest majority counterparts.

\begin{figure*}[!t]
    \begin{boxedminipage}{\textwidth}
    \begin{center}
        \textbf{Functionality $\sidealfun$}
    \end{center}
    The functionality is parameterized by a learning algorithm $\sAlgorithm$, a set of allowed auditing functions $\sauditfunctionSpace$, \sNumParties \glspl{r:inputparty} $\sPartyClient  = \{\sPartyClient_1,\ldots,\sPartyClient_{\sNumParties}\}$,
    \sNumModelHolders \glspl{r:modelowner} $\sPartyModelHolder = \{\sPartyModelHolder_1,\ldots,\sPartyModelHolder_{\sNumModelHolders}\}$,
    \sNumAuditors \glspl{r:auditrequester} $\sPartyAuditor = \{\sPartyAuditor_1,\ldots,\sPartyAuditor_{\sNumAuditors} \}$, 
    \sNumServers \glspl{r:tcomputer} $\sPartyServer = \{\sPartyServer_1,\ldots,\sPartyServer_{\sNumServers}\}$,
    \sNumInferenceComputers \glspl{r:icomputer} $\sPartyInferenceComputer = \{\sPartyInferenceComputer_1,\ldots,\sPartyInferenceComputer_{\sNumInferenceComputers}\}$,
    \sNumAuditComputers \glspl{r:acomputer} $\sPartyAuditComputer = \{\sPartyAuditComputer_1,\ldots,\sPartyAuditComputer_{\sNumAuditComputers}\}$.
    We denote the set of all parties as $\sPartySet = \sPartyClient \cup \sPartyModelHolder \cup \sPartyAuditor \cup \sPartyServer \cup \sPartyInferenceComputer \cup \sPartyAuditComputer$,
    The functionality is reactive and its state consists of a set \sidealmodelstore of models and corresponding datasets, and a set \sideallist of inference samples and corresponding predictions, %
    Operations only relevant to plaintext training are marked in \cPlainTrain{olive} and those only relevant to plaintext inference are marked in \cPlainInf{blue}.
    \\

    \begin{itemize}[itemsep=1em,label={},leftmargin=0em]
        \item \input{figures/new_functionality/camel_ideal_functionality_train}
        \item \input{figures/new_functionality/camel_ideal_functionality_inference}
        \item \input{figures/new_functionality/camel_ideal_functionality_audit}
    \end{itemize}

    \end{boxedminipage}
    \caption{\oursystem's Ideal Functionality.}
    \lfig{camel:idealfunctionality}
\end{figure*}

\subsection{Modeling Cryptographic Auditing}
A privacy-preserving auditing system must achieve secrecy, correctness, and soundness.
For \emph{secrecy}, the system must preserve the privacy guarantees of the underlying \gls{ppml} systems, except for what can be inferred from the output of auditing\footnote{Special care should be taken when choosing auditing functions to ensure their output presents an acceptable privacy-utility trade-off.}.
However, to prevent unexpected leakage from malicious audit requests, the system must also be restricted to serving only valid auditing requests, i.e., those corresponding to actual predictions made by the system.
Moreover, in order to allow us to rely on the results of auditing, the system must be \emph{correct \& sound}, i.e.,  the audit must be correctly computed even in the presence of malicious parties.
In particular, the system must ensure the audit is performed on the original training data and model corresponding to the prediction that is being audited.
The system must also have the ability to detect malicious disruptions of the audit process.
Specifically, we want to prevent malicious parties from surreptitiously aborting the audit computation and therefore require \gls{idabort} security for the auditing phase.
Otherwise, malicious actors could prevent auditing without fear of repercussion, fundamentally undermining the concept of auditability.
While \emph{publicly identifiable} abort (i.e., everyone, including input parties, learns the identity) would potentially be desirable in the \gls{soc} setting, protocols to achieve this introduce prohibitive overhead~\cite{Cunningham2017-spdzidabort,Ozdemir2022-aw} and we therefore only require traditional \gls{idabort}.
For training and inference, in comparison, we require only security with abort, as is common in practical protocols for \gls{ppml} training and inference~\cite{Mohassel2018-gy,Wagh2021-Falcon,Chida2018-ra,Goyal2021-ei,Patra2020-Blaze,Chandran2019-xg,Zheng2019-qi}.

As we prove our protocol secure in the real-ideal paradigm~\cite{Canetti2000-plainmodel}, we begin by modeling the ideal functionality that a privacy-preserving auditing protocol for \gls{ppml} should achieve based on the requirements set out above.
Despite the separation of the training, inference and auditing phases, we model the intended behavior as a single reactive ideal functionality \sidealfun (\cfref{\rfig{camel:idealfunctionality}}) as this directly implies input consistency.
We annotate parts only relevant in the plaintext training and/or inference settings in which the adversary can generate local models and predictions (in \cPlainTrain{olive} and \cPlainInf{blue}, respectively).

\fakeparagraph{Training \& Inference.}
The functionality allows training models on the input data (assuming the adversary does not choose to abort the computation), adding the resulting model to a list \sidealmodelstore.
We assume that there is an out-of-band communication channel for learning which models exist and assume the adversary learns about any models that have been trained, even if all parties involved in the training were honest.
In the functionality, we model this by leaking the model identifier to the adversary after training.
In the case of plaintext training, the functionality additionally leaks the input data and model to the adversary if it controls any \gls{r:tcomputer}.
Inference proceeds similarly, storing the inference result in the list \sideallist.
However, we must take care to also model the special case that occurs if an adversary has control over all \glspl{r:inputparty}, all training computers and at least one \gls{r:modelowner}.
In this case, the adversary could locally generate a valid combination of data and model $((\spartyinput_1, \cdots, \spartyinput_N), \smodel)$.
In the ideal world, we model this by extending the functionality to accept such locally trained models and, if consistent with the training data, append them to \sidealmodelstore before proceeding with the inference.

\fakeparagraph{Auditing.}
Any \gls{r:auditrequester} that has received a prediction, can request audits of that prediction using any auditing function from a set of allowed auditing functions.
During auditing, the lists \sidealmodelstore and \sideallist are used to verify that an auditing request is consistent.
Similarly to the special case we had to consider during inference, 
we must model the ability to locally generate valid inferences if the adversary controls all \glspl{r:icomputer} and a \gls{r:modelowner}. 
In this case, the functionality appends such locally generated predictions to \sideallist (if internally consistent) before verifying the auditing request.
During auditing, the adversary can still choose to abort the computation, but as we require auditing to achieve identifiable abort, the adversary must reveal the identity of at least one malicious party to the functionality.
In the case of inconsistencies in the audit inputs, the functionality also aborts and identifies the (uniquely determined) party at fault to the \glspl{r:acomputer}.

\begin{figure*}
    \captionlistentry{}
    \lfig{protocol:camel}
    \vspace{-0.5em}
    \procb{\text{Figure \ref*{fig:protocol:camel}}: Protocol $\sprotocolcamel$}{
        \parbox{\textwidth}{
            $\sprotocolcamel$ is a protocol between 
            \sNumParties \glspl{r:inputparty} $\sPartyClient  = \{\sPartyClient_1,\ldots,\sPartyClient_{\sNumParties}\}$,
            \sNumModelHolders \glspl{r:modelowner} $\sPartyModelHolder = \{\sPartyModelHolder_1,\ldots,\sPartyModelHolder_{\sNumModelHolders}\}$,
            \sNumAuditors \glspl{r:auditrequester} $\sPartyAuditor = \{\sPartyAuditor_1,\ldots,\sPartyAuditor_{\sNumAuditors} \}$, 
            \sNumServers \glspl{r:tcomputer} $\sPartyServer = \{\sPartyServer_1,\ldots,\sPartyServer_{\sNumServers}\}$,
            \sNumInferenceComputers \glspl{r:icomputer} $\sPartyInferenceComputer = \{\sPartyInferenceComputer_1,\ldots,\sPartyInferenceComputer_{\sNumInferenceComputers}\}$, and
            \sNumAuditComputers \glspl{r:acomputer} $\sPartyAuditComputer = \{\sPartyAuditComputer_1,\ldots,\sPartyAuditComputer_{\sNumAuditComputers}\}$.
            $\sprotocolcamel$ is parameterized by
            a learning algorithm $\sAlgorithm$, 
            a set of allowed auditing functions $\sauditfunctionSpace$,
            a \acrlong{zkpoc} $\sPoCLabel$ as in~\rdef{poc}\citeextended,
            a signature scheme $\textsf{SIG}$ as in~\rdef{signatures}\citeextended.
            Protocol parts only relevant to plaintext training are marked in \cPlainTrain{olive}, and those only relevant to plaintext inference are marked in \cPlainInf{blue}.
            \cPlainTrain{In the case of plaintext training, \oursystem is also parameterized by a proof of training $\sPot$ as in~\rdef{zkp:training}\citeextended} and,
            \cPlainInf{in the case of plaintext inference, a  proof of inference $\sPoI$ as in~\rdef{zkp:inference}\citeextended}.
            $\sprotocolcamel$ assumes access to an MPC protocol represented by instances of $\sabb$ and, in the case of \cPlainTrain{plaintext training} or \cPlainInf{plaintext inference}, a distributed randomness source $\sidealrand$.
            $\sprotocolcamel$ also assumes access to a broadcast channel $\sidealbroadcast$ and an MPC protocol $\sabbid$ with \gls{idabort}, also used by \sPoCLabel internally.
            \vspace{0.25em}
            \fakeparagraph{Input:} Each $\sPartyClient_i$ holds their training dataset $\strainset_i \in \sfield^{\sXFeatureSize \times \sNumSamples_i}$ consisting of a vector of $\sNumSamples_i$ input feature vectors of size \sXFeatureSize.
            Each \gls{r:auditrequester} $\sPartyAuditor_j$ holds a list of prediction samples $[\sX]$ where $\sX \in \sfield_p^\sXFeatureSize$ and a set of audit inputs which is a subset of $[\sX]$.
            \vspace{0.25em}
            \fakeparagraph{Initialize:}
            All parties except the \glspl{r:auditrequester} receive signing keys from $\sidealpki$.
            All parties receive all corresponding verification keys from $\sidealpki$.
            The parties also receive public setup parameters for \gls{zkpoc} $\pppoc \leftarrow \sPoCSetup(\secparam, \sNumSamples)$ 
            where $\sNumSamples$ is the maximum of all $\sNumSamples_i$ and \sModelSize,
            \cPlainTrain{and (in the case of plaintext training) $\pppot \leftarrow \sPotSetup(\secparam)$}
            \cPlainInf{and  (in the case of plaintext inference) $\pppoi \leftarrow \sPoiSetup(\secparam)$}
            .
            \vspace{0.25em}
            \fakeparagraph{Training:} The protocol proceeds as follows with \glspl{r:tcomputer} \sPartyServer, \glspl{r:inputparty} $\sPartyClient$ and \glspl{r:modelowner} $\sPartyModelHolder$, using a new instance of $\sabb$:
            \begin{protocol_steps}[label=\textsf{T}.\arabic*]
                \item \lstep{train:inputparty} Each \gls{r:inputparty} $\sPartyClient_i$ samples a random decommitment value $\scomRandomnessDataI \sample \scomRandomnessSpace$ and:
                \begin{itemize}[wide=0pt]
                    \item Inputs $\strainset_i$ to $\sabb$ \cPlainTrain{or sends $(\strainset_i, \scomRandomnessDataI)$ to all $\sPartyServer$}.
                    \item Computes a commitment to the training dataset $\scomtrainseti = \sPoCCommit(\pppoc, \strainset_i, \scomRandomnessDataI)$ and
                    sends $(\scomtrainseti)$ to all $\sPartyServer$.
                    \item \lstep{train:dist:inputparty:check} Executes 
                    \gls{poc} $\sPoCEval(\pppoc, \scomtrainseti, \ssIdeal{\strainset_i}; \strainset_i, \scomRandomnessDataI)$ with all \sPartyServer \cPlainTrain{ or each $\sPartyServer$ verifies that $\scomtrainseti = \sPoCCommit(\pppoc, \strainset_i, \scomRandomnessDataI)$ for all $\sPartyClient_i$}.
                \end{itemize}
                \item \lstep{train:tcomputer} Each \gls{r:tcomputer} $\sPartyServer_j$:
                \begin{itemize}[wide=0pt]
                    \item Samples $\ssIdeal{\scomRandomnessModel}, \ssIdeal{\scomRandomnessTrainrand}, \ssIdeal{\strainrand}$ using \sabbARand \cPlainTrain{ or all \sPartyServer and \sPartyModelHolder receive \scomRandomnessModel, \scomRandomnessTrainrand, \strainrand from \sidealrand}.
                    \item Invoke $\sabbATrain(\ssIdeal{\strainset_1}, \ldots, \ssIdeal{\strainset_\sNumParties}, \ssIdeal{\strainrand})$
                    to compute the model $\ssIdeal{\smodel}$ 
                    \cPlainTrain{ or compute $\smodel \leftarrow \sAlgorithm(\strainset_1, \ldots, \strainset_\sNumParties, \strainrand)$}.
                    \item Using \sabb, commit to the model $\ssIdeal{\scommitmentModel} = \sPoCCommit(\pppoc, \ssIdeal{\smodel}, \ssIdeal{\scomRandomnessModel})$, randomness $\ssIdeal{\scomtrainrand} = \sPoCCommit(\pppoc, \ssIdeal{\strainrand}, \ssIdeal{\scomRandomnessTrainrand})$ and open $\scommitmentModel, \scomtrainrand$ to all \sPartyServer, \sPartyClient and \sPartyModelHolder
                    \cPlainTrain{ or compute $\scommitmentModel \leftarrow \sPoCCommit(\pppoc, \smodel, \scomRandomnessModel)$ and $\scomtrainrand \leftarrow \sPoCCommit(\pppoc, \strainrand, \scomRandomnessTrainrand)$}.
                    \item Compute $\ssIdeal{\ssigTrainingComputer} \leftarrow \sSigDistSign(\skTCJ, \scomtrainsetsconcat\concat\scommitmentModel\concat \scomtrainrand)$ and open \ssigTrainingComputer to \sPartyModelHolder, \sPartyClient using \sabb \\ \cPlainTrain{ or $\sproofTrain \leftarrow \sPotProve(\pppot, (\scomtrainsets, \scommitmentModel, \scomtrainrand); \strainsets, \smodel, \strainrand, \scomRandomnessDatas, \scomRandomnessModel, \scomRandomnessTrainrand)$}. 
                    \item Send $(\scomtrainsets)$ and open $\ssIdeal{\smodel}, \ssIdeal{\scomRandomnessModel}$  to \sPartyModelHolder using \sabb \cPlainTrain{ or send $(\scomtrainsets, \scommitmentModel, \scomtrainrand, \sproofTrain, \smodel, \scomRandomnessModel, \scomRandomnessTrainrand, \strainrand)$ to all $\sPartyModelHolder$}.
                    \item Send $(\scomtrainsets, \cPlainTrain{\scommitmentModel, \scomtrainrand, \sproofTrain})$ to all \glspl{r:inputparty} $\sPartyClient$.
                \end{itemize}
                \item \lstep{train:ip:sign} Each $\sPartyClient_i$ checks that it received the same $(\scomtrainsets, \cPlainTrain{\scommitmentModel, \scomtrainrand, \sproofTrain})$ from all \sPartyServer, $\sSigVerify(\pkTC, \scomtrainsetsconcat \concat \scommitmentModel \concat \scomtrainrand, \ssigTrainingComputer)$ \cPlainTrain{ or \sPotVerify(\pppot, \scomtrainsets, \scommitmentModel, \scomtrainrand, \sproofTrain)}, and its \scomtrainseti is contained in \scomtrainsets and aborts otherwise.
                Then, each computes $\ssigTrain^i \leftarrow \sSigSign(\skIPI, \scomtrainsetsconcat \concat \scommitmentModel \concat \scomtrainrand)$ and sends $\ssigTrain^i$ to all $\sPartyModelHolder$.
                \item \lstep{train:modelowner} Each \gls{r:modelowner} $\sPartyModelHolder_k$ checks each of the following and aborts if any fail:
                \begin{itemize}[wide=0pt]
                    \item Verify that the $(\scomtrainsets, \ssigTrainingComputer)$ \cPlainTrain{(and $\scommitmentModel, \scomtrainrand, \smodel, \scomRandomnessModel, \scomRandomnessTrainrand, \strainrand, \sproofTrain$)} received from each $\sPartyServer$ are consistent with each other. 
                    \item \cPlainTrain{$\scommitmentModel = \sPoCCommit(\pppoc, \smodel, \scomRandomnessModel)$ and $\scomtrainrand = \sPoCCommit(\pppoc, \strainrand, \scomRandomnessTrainrand)$.}
                    \item The list of signatures $\sSigVerify(\pkIPI, \scomtrainsetsconcat \concat \scommitmentModel \concat \scomtrainrand, \ssigTrain^i)$ for each $\sPartyClient_i$.
                    \item $\sSigVerify(\pkTC, \scomtrainsetsconcat \concat \scommitmentModel \concat \scomtrainrand, \ssigTrainingComputer)$ \cPlainTrain{ or \sPotVerify(\pppot, \scomtrainsets, \scommitmentModel, \scomtrainrand, \sproofTrain)}.
                \end{itemize}
            \end{protocol_steps}
            \vspace{0.25em}
            \fakeparagraph{Inference:} The protocol proceeds as follows between \glspl{r:icomputer} \sPartyInferenceComputer, \gls{r:auditrequester} $\sPartyAuditor_i$ and \gls{r:modelowner} $\sPartyModelHolder_k$, using a new instance of $\sabb$:
            \begin{protocol_steps}[label=\textsf{I}.\arabic*]
                \item \lstep{camel:inf:input:client} $\sPartyAuditor_i$ sends $\scommitmentModel^\prime$ (identifying the requested model) to all $\sPartyInferenceComputer$, and inputs a prediction sample $\sX$ to \sabb \cPlainInf{ or sends \sX to all $\sPartyInferenceComputer$}, then:
                \begin{itemize}[wide=0pt]
                    \item \lstep{camel:inf:input:modelowner} All  $\sPartyInferenceComputer$ ask $\sPartyModelHolder_k$ to send $(\scommitment, \ssigTrain, \ssigTrainingComputer\cPlainTrain{\text{ or }\sproofTrain})$ where $\scommitment = (\scomtrainsets, \scommitmentModel, \scomtrainrand)$ to $\sPartyInferenceComputer$, and input the model $\smodel$ to \sabb \cPlainInf{ or send $\smodel,\scomRandomnessModel$ to $\sPartyInferenceComputer$}.
                    \item \lstep{camel:inf:input:checkcomm} The \glspl{r:icomputer} abort if $\scommitmentModel \neq \scommitmentModel^\prime$.
                    \item \lstep{camel:inf:input:cc} \lstep{inf:poc} $\sPartyModelHolder_k$ executes \sPoCEval(\pppoc, \scommitmentModel, \ssIdeal{\smodel}; \smodel, \scomRandomnessModel) with all \sPartyInferenceComputer \cPlainInf{ or each $\sPartyInferenceComputer$ checks $\scommitmentModel = \sPoCCommit(\pppoc, \smodel, \scomRandomnessModel)$} and aborts if it fails.
                \end{itemize}
                \item \lstep{camel:inf:icomputer} Each \gls{r:icomputer} $\sPartyInferenceComputer_j$:
                \begin{itemize}[wide=0pt]
                    \item Computes $\ssIdeal{y}$ by invoking $\sabbAInf(\ssIdeal{\smodel}, \ssIdeal{\sX})$ \cPlainInf{ or computes $y \leftarrow \smodel(\sX)$}.
                    \item Samples $\ssIdeal{\scomRandomnessX}, \ssIdeal{\scomRandomnessY}$ using \sabbARand \cPlainInf{ or all \sPartyInferenceComputer and $\sPartyClient_i$ receive \scomRandomnessX, \scomRandomnessY \sidealrand}.
                    \item Computes $\ssIdeal{\scommitmentX} = \sPoCCommit(\pppoc, \ssIdeal{\sX}, \ssIdeal{\scomRandomnessX})$ and $\ssIdeal{\scommitmentY} = \sPoCCommit(\pppoc, \ssIdeal{y}, \ssIdeal{\scomRandomnessY})$ and opens \scommitmentX and \scommitmentY to \sPartyInferenceComputer, $\sPartyAuditor_i$ and $\sPartyModelHolder_k$ using \sabb 
                    \cPlainInf{ or computes $\scommitmentX = \sPoCCommit(\pppoc, \sX, \scomRandomnessX)$ and $\scommitmentY = \sPoCCommit(\pppoc, \sY, \scomRandomnessY)$ and send $(\scommitmentX, \scommitmentY)$ to $\sPartyAuditor_i$ and $\sPartyModelHolder_k$ }.
                    \item Computes $\ssIdeal{\ssigInferenceComputer} {\leftarrow} \sSigDistSign(\skICJ, \scommitment \concat \scommitmentX \concat \scommitmentY)$ \& opens \ssigInferenceComputer to $\sPartyAuditor_i,\sPartyModelHolder_k$ using \sabb \cPlainInf{ or $\sproofInference \leftarrow \sPoiProve(\pppoi, \scommitmentModel, \scommitmentX, \scommitmentY; \smodel, \sX, \sY, \scomRandomnessModel, \scomRandomnessX, \scomRandomnessY)$}.
                    \item Sends $(\scommitment, \ssigTrain, \ssigTrainingComputer\cPlainTrain{\text{ or }\sproofTrain})$ to $\sPartyAuditor_i$ and open $(\ssIdeal{y}, \ssIdeal{\scomRandomnessX}, \ssIdeal{\scomRandomnessY})$ with \sabb to $\sPartyAuditor_i$, \cPlainInf{ or send $(\scommitment, \ssigTrain, \ssigTrainingComputer$}\cPlainTrain{$\text{ or }\sproofTrain,$}\cPlainInf{$ \sproofInference, \sY, \scomRandomnessX, \scomRandomnessY)$ to $\sPartyAuditor_i$}.
                    Sends $(\scommitment, \cPlainInf{\sproofInference})$ to $\sPartyModelHolder_k$.
                \end{itemize}
                \item \lstep{camel:inf:modelownerverify} The \gls{r:modelowner} checks that it receives the same $(\scommitment, \scommitmentX, \scommitmentY)$, that $\sSigVerify(\pkIC, \scommitment \concat \scommitmentX \concat \scommitmentY, \ssigInferenceComputer)$ \cPlainTrain{ or \sPoiVerify(\pppoi, \scommitmentModel, \scommitmentX, \scommitmentY, \sproofInference)}, aborting otherwise,
                and computes $\ssigReceipt \leftarrow \sSigSign(\skMHK, \scommitment \concat \scommitmentX \concat \scommitmentY \concat \ssigTrain \concat \ssigTrainingComputer\cPlainTrain{\text{ or } \sproofTrain} \concat  \ssigInferenceComputer\cPlainInf{\text{ or }\sproofInference} )$ and sends $\ssigReceipt$ to $\sPartyAuditor_i$.
                \item \lstep{camel:inf:clientverify} The \gls{r:auditrequester} \sPartyAuditor checks each of the following and aborts if any fails:
                \begin{itemize}[wide=0pt]
                    \item Verify that the $(\scommitment, \cPlainInf{\scommitmentX, \scommitmentY,} \ssigTrain, \ssigTrainingComputer\cPlainTrain{\text{ or }\sproofTrain}, \ssigInferenceComputer \cPlainInf{\text{ or } \sproofInference, \sY, \scomRandomnessY, \scomRandomnessX})$                   
                    received from each $\sPartyServer$ are consistent with each other. 
                    \item $\sSigVerify(\pkMHK, \scommitment \concat \scommitmentX \concat \scommitmentY \concat \ssigTrain \concat \ssigTrainingComputer\cPlainTrain{\text{ or } \sproofTrain} \concat  \ssigInferenceComputer\cPlainInf{\text{ or }\sproofInference} , \ssigReceipt)$ is a valid signature by $\pkMH$.
                    \item \cPlainInf{Verify that $\scommitmentX = \sPoCCommit(\pppoc, \sX, \scomRandomnessX)$ and $\scommitmentY = \sPoCCommit(\pppoc, y, \scomRandomnessY)$.}
                    \item The list of signatures $\sSigVerify(\pkIPI, \scommitment, \ssigTrain^i)$ for each $\sPartyClient_i$, and $\sSigVerify(\pkIC, \scommitment \concat \scommitmentX \concat \scommitmentY, \ssigInferenceComputer)$ \cPlainTrain{ or \sPoiVerify(\pppoi, \scommitmentModel, \scommitmentX, \scommitmentY, \sproofInference)}.
                \end{itemize}
            \end{protocol_steps}
        }
    }
\end{figure*}

\begin{figure*}
    \procb{Figure~\ref*{fig:protocol:camel}: Protocol $\sprotocolcamel$ (cont.)}{
        \parbox{\textwidth}{
            \fakeparagraph{Auditing:} The protocol proceeds as follows on a new instance of $\sabbid$ between computing parties \sPartyAuditComputer, a \gls{r:auditrequester} $\sPartyAuditor_j$ and the \gls{r:modelowner} $\sPartyModelHolder_k$:
            \begin{protocol_steps}[label=\textsf{A}.\arabic*]
                \item \lstep{camel:audit:clientinput} The \gls{r:auditrequester} $\sPartyAuditor_j$ inputs $({\sX}, {y})$ to $\sabbid$
                and broadcasts $(\scommitment, \scommitmentX, \scommitmentY, \ssigReceipt, \ssigTrain, \ssigTrainingComputer \cPlainTrain{\text{ or }\sproofTrain}, \ssigInferenceComputer \cPlainInf{\text{ or } \sproofInference}, \pkMHK, \sauditfunction, \aux)$ to all parties using $\sidealbroadcast$.
                \item \lstep{camel:audit:localchecks} All parties check that \pkMHK is a valid identity from $\sidealpki$, verify the \gls{r:modelowner} signature with $\sSigVerify(\pkMHK, \scommitment \concat \scommitmentX \concat \scommitmentY \concat \ssigTrain \concat \ssigTrainingComputer\cPlainTrain{\text{ or } \sproofTrain} \concat  \ssigInferenceComputer\cPlainInf{\text{ or }\sproofInference}, \ssigReceipt)$
                and check that $\sauditfunction\in \sauditfunctionSpace$. 
                Otherwise, each party aborts marking $\sPartyAuditor_j$ as malicious.
                \item \lstep{camel:audit:verifyauditor} \fakeparagraph{Verify Audit Requester:}
                    The \gls{r:auditrequester} $\sPartyAuditor_j$ runs $\sPoCEvalID(\pppoc, \scommitmentX, \ssIdeal{\sX}; \sX, \scomRandomnessX),
                    \sPoCEvalID(\pppoc, \scommitmentY, \ssIdeal{y}; y, \scomRandomnessY)$ with the \glspl{r:acomputer}
                    to prove to the \sPartyAuditComputer that its inputs $\sX$ and $y$ are consistent with $\scommitmentX$ and $\scommitmentY$.
                   If any of the checks fail, $\sPartyAuditComputer_i$ aborts marking $\sPartyAuditor_j$ as malicious.
                \item \fakeparagraph{Verify Inference:}\lstep{camel:audit:verifyinference}
                The \gls{r:modelowner} inputs the model ${\smodel}$ to $\sabbid$:
                \begin{itemize}[wide=0pt]
                    \item The \gls{r:modelowner} $\sPartyModelHolder_k$ runs $\sPoCEvalID(\pppoc, \scommitmentModel, \ssIdeal{\smodel}; \smodel, \scomRandomnessModel)$ with the \glspl{r:acomputer}
                    acting as the verifiers to proof that its model input is consistent with \scommitmentModel
                    from $\sPartyAuditor_j$. Each $\sPartyAuditComputer_i$ aborts marking $\sPartyModelHolder_k$ as malicious if verification fails.
                    \item Each \gls{r:acomputer} computes $\sSigVerify(\pkIC, \scommitment \concat \scommitmentX \concat \scommitmentY, \ssigInferenceComputer)$ \cPlainInf{ or \sPoiVerify(\pppoi, \scommitmentModel, \scommitmentX, \scommitmentY, \sproofInference)}.
                \end{itemize}
                \item \fakeparagraph{Verify Training:} \lstep{camel:audit:verifytraining}
                Each \gls{r:inputparty} $\sPartyClient_i$ inputs their dataset ${\strainset_i}$
                to \sabbid.
                \begin{itemize}[wide=0pt]
                    \item Each \gls{r:inputparty} $\sPartyClient_i$ performs $\sPoCEvalID(\pppoc, \scomtrainseti, \ssIdeal{\strainset_i}; \strainset_i, \scomRandomnessDataI)$
                with the \glspl{r:acomputer} acting as verifiers to proof that its input $\ssIdeal{\strainset_i}$ is consistent with $\scomtrainseti$.
                    $\sPartyAuditComputer_j$ also checks $\sSigVerify(\pkIPI, \scommitment, \ssigTrain)$ for each $\sPartyClient_i$.
                    If verification fails, $\sPartyAuditComputer_j$ aborts marking $\sPartyModelHolder_k$ as malicious.
                    \item Each \gls{r:acomputer} computes $\sSigVerify(\pkTC, \scommitment, \ssigTrainingComputer)$ \cPlainTrain{ or \sPotVerify(\pppot, \scommitment, \sproofTrain)}.
                \end{itemize}
                \item \lstep{camel:audit:function} The \glspl{r:acomputer} compute $\ssIdeal{o} \leftarrow \sabbAAuditID(\sauditfunction, \ssIdeal{\strainset_1}, \ldots, \ssIdeal{\strainset_\sNumParties}, \ssIdeal{\smodel}, \ssIdeal{\sX}, \ssIdeal{\sY}, \aux)$
                and use \sabbid to open $o$ at $\sPartyAuditor_j$.
            \end{protocol_steps}
        }
     }
\end{figure*}

\subsection{\oursystem Protocol}
\lsec{auditing_framework}
Our protocol lifts existing \gls{ppml} systems to the cryptographic auditing setting by augmenting training and inference to produce receipts that can later be used to verify the consistency of the training data, model, and prediction under audit.
Supporting real-world deployments with many potential inference \glspl{r:auditrequester} requires scaling independently of the number of inferences and \glspl{r:auditrequester} in the system.
At the same time, we want to minimize the state that \glspl{r:auditrequester} need to store beyond the received predictions.
In \oursystem, we achieve this through concise cryptographic receipts for training and inference, which allows \glspl{r:auditrequester} to efficiently store all material necessary to later verify the consistency of training data, model, inference sample, and prediction during auditing.
At the same time, receipts are cryptographically bound to specific inferences, i.e., \glspl{r:auditrequester} cannot generate audit requests for predictions they did not receive.
In the following, we give an overview of our protocol, \sprotocolcamel (c.f. \rfig{protocol:camel}) and its building blocks.
\ifdefined\isnotextended
We prove the security of our protocol in the real/ideal world paradigm~\cite{Canetti2000-plainmodel} in Appendix~\ref{sec:apx:proofs:composability} in the extended version of the paper~\cite{arc_full}.
\else
We prove the security of our protocol in the real/ideal world paradigm~\cite{Canetti2000-plainmodel} in Appendix~\ref{sec:apx:proofs:composability}.\fi

\fakeparagraph{Building Blocks.}
We construct our protocol from several cryptographic primitives, including a reactive \gls{abb} interface to abstract \gls{ppml} protocols.
In addition, we require secure point-to-point channels between parties that participate in the same phase and, in the auditing phase, a secure broadcast channel in order to achieve identifiable abort.
We assume that all parties (except for \glspl{r:auditrequester}) have a cryptographic identity which is set up through a \gls{pki}, and that \glspl{r:auditrequester} can access the (public) identities of the other parties through the \gls{pki}.
Our protocol makes use of standard signatures (\cfref{\rdef{signatures}} in Appendix~\ref{sec:apx:proofs:definitions}\citeextended) and a \acrfull{poc} (\cfref{\rdef{poc}}) which acts like a commitment, but admits significantly more efficient instantiations in the secure computation setting, as we discuss in the next section.
In the following, we will use commitment and \gls{poc} interchangeably. 

\fakeparagraph{Proof of Training/Inference.}
When training and inference are realized via secure computation, computational integrity follows directly from the guarantees of the underlying \gls{mpc} protocol.
However, supporting real-world deployment scenarios where either training or inference are computed centrally on plaintexts, requires explicit proofs of computational integrity.
There has been significant work on increasingly efficient proofs for both training~\cite{Garg2023-bn,Sun2023-zkdl} and inference~\cite{Kang2022-zkml,Chen2024-ng}, making use of advances in SNARKs and related proof techniques.
We model these as \cPlainTrain{\glsreset{zkpot}\gls{zkpot}} or \cPlainInf{\glsreset{zkpoi}\gls{zkpoi}} (\cfref{Definitions~\ref{def:zkp:training} and~\ref{def:zkp:inference} in Appendix~\ref{sec:apx:proofs:definitions}}\citeextended).
Recently, \gls{pol}~\cite{Jia2021-PoL} has emerged, aiming to provide an easier-to-generate alternative to \gls{zkpot} by relying only on heuristic assumptions rather than strong cryptographic assumptions~\cite{Fang2023-zq}.
However, \gls{pol} \emph{verification} requires access to the original training data and is computationally expensive, requiring many epochs of (re-)training.
Therefore, \gls{pol} is less attractive for our setting, where verification must be computed under MPC.

\fakeparagraph{Training \& Inference.}
The training and inference phase of \sprotocolcamel both provide a layer of consistency around calls to the underlying \gls{ppml} implementation (i.e., \sabbATrain or \sabbAInf).
While the inference computation is fully deterministic, training also requires randomness which we securely sample (i.e., \sidealrand~\cite{Blum1983-ek}) and commit to in order to prevent reordering attacks~\cite{Shumailov2021-dataorderingattacks}.
Otherwise, the two phases proceed  nearly identically:
\begin{itemize}[leftmargin=0pt,align=left,labelsep=5pt,labelwidth=!]
    \item \emph{Input Commitments} (\ref{step:train:inputparty}/\ref{step:camel:inf:input:client}).
    In addition to their inputs, parties must provide a commitment to their inputs and the protocol verifies the consistency of these commitments before proceeding.
    In the secure computation setting, this uses \sPoCEval which involves an (efficient) multi-party computation.
    In the plaintext setting, the computing parties can simply locally recompute the commitments.

    \item \emph{Computation Integrity} (\ref{step:train:tcomputer}/\ref{step:camel:inf:icomputer}).
    After computing the underlying \gls{ml} training or inference, the computing parties commit (in the secure setting, collaboratively under MPC) to the result and provide the result, commitment, and associated decommitment randomness to the output-receiving parties.
    As we later need to show that these outputs were the result of a valid computation, the computing parties attest to the integrity of the computation.
    In the secure setting, this can be achieved via a distributed signature, as at least one of the computing parties must be honest.
    In the plaintext setting, this requires a \gls{zkpot} or \gls{zkpoi}, as we cannot rely on a split-trust assumption for integrity.

    \item \emph{Input Integrity}  (\ref{step:train:ip:sign}/\ref{step:camel:inf:modelownerverify}).
    While the signature or proof tie the result to a valid computation, they do not provide sufficient guarantees about the inputs.
    Therefore, the input parties verify their inputs were used and provide signatures to attest to this.

    \item \emph{Output Consistency} (\ref{step:train:modelowner}/\ref{step:camel:inf:clientverify}).
    The output from the computation is provided to the output-receiving parties (in the plain) along with the associated receipt data.
    These parties then verify that the signatures and proofs are valid, and that the outputs (and decommitment randomness) they received matches the commitments in the receipt.

\end{itemize}

Note that providing a plaintext inference result to the client is essential for a useful inference service.
However, one could consider a variant of the protocol in which the \glspl{r:modelowner} only receive shares of the model, rather than the plaintext model.
While our protocol could be trivially extended to support this,
this would both unnecessarily complicate the notation and would require the long-term storage of secret shares and potentially complicated operations such as secret-share maintenance and re-sharing to new sets of entities~\cite{Schultz2010-mpss,Herzberg1995-pss,Maram2019-ao}, a complexity which we aim to avoid in our design.

The  receipt received by the \glspl{r:modelowner} after training comprises
commitments to the training data, the training randomness and the resulting model; the \glspl{r:inputparty}'s signatures; and either the signature from the computing parties or a proof of training.
During inference, the \glspl{r:modelowner} provide this training receipt instead of merely the model commitment.
As a result, the inference receipt is essentially an extension of the training receipt and includes the equivalent commitments and signatures (or proofs, where applicable) for both training and inference.
Therefore, the conciseness of the underlying commitments (i.e., \gls{poc}) is crucial to ensuring that the overhead imposed upon the client due to the need to store this receipt is minimized.

\fakeparagraph{Auditing}
During auditing, the client provides the receipt and inputs the prediction sample and result into the MPC computation.
Meanwhile, the \glspl{r:modelowner} and \glspl{r:inputparty} need to provide their respective inputs.
The protocol first confirms that the signatures (or proofs, where applicable) in the receipt are valid, in reverse order, i.e., beginning with the last signature generated at the end of inference.
Then, it uses \sPoCEval to verify the consistency of the provided inputs with the commitments in the receipt.
Finally, after all checks have passed, the protocol computes the audit function (i.e., \sabbAAudit).
Should any of the checks fail, the protocol aborts and identifies the party at fault.
During the auditing phase, we can rely on the fact that an honest client will only ever accept valid receipts, as the proof of computational integrity (or, in the secure computation setting, signatures) allow them to verify that the receipt was generated correctly.
At the same time, these checks prevent the adversary from constructing a valid malicious receipt that would incriminate an honest party.
As a result, any inconsistency between the receipt and the input data provided to the audit phase can be uniquely attributed to the party providing the corresponding auditing input.
For more details, we refer to the proof in Appendix~\ref{sec:apx:proofs:composability}\citeextended.

\section{Proof of Consistency}
\lsec{design:inputconsistency}
In our auditing protocol \sprotocolcamel, we assume access to a \Acrlong{poc} protocol \sPoCLabel that allows a party to commit to their (secret) inputs and later allows the parties to collaboratively check that a given (set of) secret shared\footnote{More precisely, in the representation used by \sabb.} values is consistent with the provided commitment.
Note that \gls{poc} might seem related to \gls{vss}, which guarantees that parties receive a valid sharing of a given value.
However, this is orthogonal to our requirement of ensuring consistency of \emph{inputs} across different phases and, therefore, different sharings of that value.
In theory, the consistency guarantees required for \oursystem could also be achieved straightforwardly by using standard commitments.
However, as we show in our evaluation, such approaches incur significant performance overheads, especially during verification, making practical deployment infeasible.

\subsection{Defining Proof-of-Consistency}
\lsec{design:consistency:approaches}
In the following, we provide the formal definition of \sPoCLabel and the properties it needs to achieve before discussing several approaches based on existing literature and highlighting their inherent limitations with regard to efficiency and succinctness.

\begin{definition}[\Acrlong{poc} Protocol]%
    \ldef{poc}
    A valid Proof-of-Consistency is an interaction between a Prover \sprover and a set of $\sNumPartiesGenericParty - 1$ Verifiers \sverifier.
    This protocol allows the verifiers to check that a vector $\ssIdeal{\sInputs} = (\ssIdeal{\sInputs_1}, \ldots, \ssIdeal{\sInputs_{\sNumSamples}})$
    stored in an ideal functionality $\sabb$ is consistent
    with a commitment $\scommitment$ %
    to $\sInputs = (\sInputs_1, \ldots, \sInputs_{\sNumSamples}) \in \sfield_p^{\sNumSamples}$.
    A Proof-of-Consistency is defined as a set of protocols (\sPoCSetup, \sPoCCommit, \sPoCEval) where:
    \begin{algos}
        \item $\sPoCSetup(1^\lambda, \sNumSamples) \rightarrow \pppoc$: prepares public parameters $\mathsf{pp}$ supporting inputs of size $\sNumSamples$.
        \item $\sPoCCommit(\pppoc, \sInputs, \sRandomness) \rightarrow \scommitment$:
        An algorithm in which the prover generates a commitment to (a vector of) inputs $\sInputs$ with randomness \sRandomness.
        \item $\sPoCEval(\pppoc, \scommitment, \ssIdeal{\sInputs}; \sInputs, \sRandomness) \rightarrow \{ 0, 1\}$:
        A protocol where the prover convinces the verifiers that the commitment $\scommitment$ is consistent with $\ssIdeal{\sInputs}$.
        Only the prover knows \sInputs and \sRandomness.

    \end{algos}
\end{definition}

\noindent
A valid Proof-of-Consistency should satisfy correctness, soundness and zero-knowledge, which we formally define in Appendix~\ref{sec:apx:proofs:consistency}.
In addition to the formal requirements,
a \sPoCLabel instantiation should yield succinct commitments and its protocols should be efficiently computable (specifically, require minimal MPC operations  in \sPoCEval) in order to be practical for \gls{ml}.
Succinctness is crucial for ensuring efficient communication and storage, especially when dealing with large input sizes and resource-constrained clients.
In the following, we discuss existing approaches and how they fall short in our setting. 
\ifdefined\isnotextended
We refer to Appendix~\ref{sec:alternativeapproaches} in the extended version of the paper~\cite{arc_full} for formal definitions of the corresponding protocols.
\else
We refer to Appendix~\ref{sec:alternativeapproaches} for formal definitions of the corresponding protocols.
\fi

\fakeparspacingtop
\fakeparagraph{Direct Commitments~\cite{Agrawal2021-rg,Segal2020-tu,Kilbertus2018-zl}.}
A straightforward approach to \sPoCLabel is to use a cryptographic commitment scheme
to instantiate \sPoCSetup and \sPoCCommit with \sComSetup and \sComCommit, respectively.
In \sPoCEval, the commitment is verified with respect to the secret shared inputs \ssIdeal{\sInputs} and decommitment \ssIdeal{\sRandomness} by computing \sComVerify using \sabb.
This typically requires recomputing the commitment under \gls{mpc}, because the usual implementation of \sComVerify is to re-compute the commitment $\scommitment^\prime \leftarrow \sComCommit(\ssIdeal{\sInputs}, \ssIdeal{\sRandomness})$ and checking that $\scommitment^\prime = \scommitment$.
Related work has suggested to use this protocol with commitments based on a collision-resistant hash function,
such as \mbox{SHA-2}~\cite{Kilbertus2018-zl}, \mbox{SHA-3}~\cite{Segal2020-tu} and
MPC-friendly constructions such as LowMCHash-256~\cite{Agrawal2021-rg}.
The advantage of this approach lies in its succinct commitment size which is typically constant.
However, despite its efficient storage needs, the hash-based approach incurs significant computational costs during verification.
This is primarily due to hashes relying on non-linear operations, which are expensive to compute in \gls{mpc}.

\fakeparspacingtop
\fakeparagraph{Homomorphic Commitments~\cite{Zheng2021-tb}.}
To mitigate the \gls{mpc} cost of \sPoCEval, one can rely on homomorphic commitments such as Pedersen commitments instantiated using an elliptic curve group~\cite{Zheng2021-tb, Pedersen1992-Commitments}.
Instead of calling \sComVerify for the full input vector \ssIdeal{\sInputs}, parties use the homomorphism to compute a linear combination of commitments to individual elements $\sinputSample_i$, trading off \gls{mpc} overhead with local computation.
As a result, parties only compute a single commitment $\tilde{\scommitment}^\prime = \sComCommit(\sum_i \gamma^i \cdot \ssIdeal{\sinputSample_i} )$ with \sabb in \sPoCEval and compare the result with $\tilde{\scommitment}^\prime = \sum_i \gamma^i \cdot \scommitment_i$.
Unfortunately, a downside of this approach is that commitments to individual elements must be stored, resulting in a size that is linear in $\mathopen|\sInputs\mathclose|$.
This approach results in a \sPoCEval that is asymptotically more efficient than the hash-based approach.
In practice, hash-based approaches remain more concretely efficient for very small inputs.
However,  the Pedersen commitment approach becomes more efficient already for moderate input sizes.

\subsection{\oursystem PoC Protocol}
\label{sec:ourprotocol}
The key insight of our efficient PoC protocol is that
 we do not actually need to compute \sComCommit in order to verify that the \ssIdeal{\sInputs} in \sabb matches the input \sInputs of \sComCommit.
Instead, we propose a protocol that allows the prover \sprover to convince the verifiers of this fact with a polynomial identity test.
Towards this, we first define a polynomial \mbox{$\spolynomial(B) \defeq \sum_{i=1}^d x_i \cdot B^i$}, i.e., interpreting the elements of $\sInputs \in \sfield_p^\sNumSamples$ as the coefficients of the polynomial.
The prover commits to $\spolynomial$ using a (homomorphic) polynomial commitment scheme~\cite{Kate2010-px} to obtain a constant-size commitment $\scommitment$.
In \sPoCEval, the parties then first collaboratively sample a point $\beta \sample \sfield_p$ and then evaluate the polynomial at $\beta$ (using \sabb) by computing \mbox{$\rho \defeq  \sum_{i=1}^d \ssIdeal{x_i} \cdot \beta^i$} and opening $\rho$.
This is cheap in \gls{mpc} because it only involves addition and scaling operations on the secret shares \ssIdeal{\sInputs} which can be executed locally.
The prover, who originally committed to $\sInputs$ with $\scommitment$, can now do a polynomial commitment opening proof to show that $\spolynomial(\beta)$ equals $\rho$.
The other parties verify this evaluation proof, which, if valid, implies that the polynomial in $\sabb$ is (with high probability) equal to the one committed to with $\scommitment$.
One caveat with this approach is that $\rho$ reveals information about $\sInputs$.
We can overcome this by generating and committing to a random polynomial $\spolynomial_\sProtocolRandomness$ at the beginning of \sPoCEval and use it to (additively) mask $\spolynomial$.
We can then evaluate and open $\spolynomial(\beta) + \spolynomial_\sProtocolRandomness(\beta)$ which is now  indistinguishable from random.
The prover and verifiers proceed with the evaluation proof as before, but on the combined commitment of $\spolynomial$ and $\spolynomial_\sProtocolRandomness$, using the homomorphic property of the polynomial commitment scheme.
In the following, we provide a formal definition:

\begin{protocol}[Consistency Check]
    \label{protocol:cc}
    Let $\sabb$ be an instance of an ideal MPC functionality over a field $\sfield_p$,
    let $\sidealrand$ be an ideal functionality that returns a random element from $\sfield_p$ and
    let $\sinputSamples = (\sinputSample_1, \ldots, \sinputSample_{\sNumSamples}) \in \sfield_p^{\sNumSamples}$ by the input of prover $\sprover$.
    Let $\ssIdeal{\sinputSamples} = (\ssIdeal{\sinputSample_1}, \ldots, \ssIdeal{\sinputSample_{\sNumSamples}})$
    be the input of the prover $\sprover$ to $\sabb$.
    Let $\PCScheme$ be a polynomial commitment scheme as in~\rdef{pc}\citeextended that is also homomorphic as in~\rdef{homomorphic_commitment}\citeextended.
The protocol \sprotocolCC works as follows:
\begin{algos}
    \item $\sCCSetup(1^\lambda, \sNumSamples) \rightarrow \pp$: Run $\pp \leftarrow \PCSetup(\sNumSamples)$, where \sNumSamples is the number of elements in the input.
    \item $\sCCCommit(\pp, \sInputs, \sRandomness) \rightarrow \scommitment$:
    The prover computes a polynomial commitment $\scommitment \leftarrow \PCCommit(\texttt{pp}, \spolynomial, \sRandomness)$ where $\spolynomial$ is defined as \mbox{$\spolynomial(z) = \sum_{i=1}^d x_i \cdot z^i$}.
    The prover outputs $\scommitment$.
    \item $\sCCEval(\pp, \scommitment, \ssIdeal{\sInputs}; \sInputs, \sRandomness) \rightarrow \{ 0, 1\}$:
    The protocol proceeds as follows:
    \begin{enumerate}[align=left,itemsep=2pt,left=0pt]
        \item The prover samples a masking value $\sProtocolRandomness \sample \sfield_p$ and commitment randomness $\sRandomness_\sProtocolRandomness \sample \sfield_p$ and computes a polynomial commitment
            $\scommitment_\sProtocolRandomness \leftarrow \PCCommit(\texttt{pp}, \spolynomial_\sProtocolRandomness, \sRandomness_\sProtocolRandomness)$ to a degree-0 polynomial \mbox{$\spolynomial_\sProtocolRandomness(z) = \sProtocolRandomness$.}
        The prover sends $\scommitment_\sProtocolRandomness$ to all parties and inputs $\sProtocolRandomness$ to $\sabb$.
        \item The parties invoke $\sidealrand$ to obtain a random challenge $\beta \sample \sfield_p$.
        \item \lstep{cc:rho} The parties invoke $\sabb$ to compute \mbox{$\ssIdeal{\rho} \defeq \ssIdeal{\sProtocolRandomness} + \sum_{i=1}^{\sNumSamples} \ssIdeal{\sinputSample_i} \cdot \beta^i$}
        and subsequently open $\rho$.
        \item \lstep{cc:prover} The prover $\sprover$ generates a proof $\pi \leftarrow \PCProve(\pp, \scommitment + \scommitment_\sProtocolRandomness, \spolynomial + \spolynomial_\sProtocolRandomness, \sRandomness + \sRandomness_\sProtocolRandomness, \beta, \rho)$ and sends $\pi$ to
        each verifier $\sverifier$.
        \item \lstep{cc:verifiers} Each verifier runs $\PCCheck(\pp, \scommitment \cdot \scommitment_\sProtocolRandomness, \beta, \rho, \pi)$. If verification passes, they output $1$, otherwise $0$.
    \end{enumerate}
\end{algos}

\end{protocol}

\noindent Intuitively, security follows from the fact that if the committed polynomial is not equal to the polynomial evaluated on the secret shares, then the prover can only open the commitment to $\rho$ with negligible probability.
We provide a formal security proof in Appendix~\ref{sec:apx:proofs:consistency} (\rlem{cc}).
The protocol \sprotocolCC can be extended to provide \gls{idabort}, denoted as \sprotocolCCID, by using a broadcast channel for the prover in~\rstep{cc:prover} and using an MPC protocol that provides identifiable abort ($\sabbid$). %

\fakeparagraph{Cost Analysis.}
Many polynomial commitment schemes have a constant storage overhead independent in the input size,
resulting in each party having to store only a single, constant-sized commitment for each input vector.
Our protocol can be instantiated with any homomorphic polynomial commitment scheme
and inherits the efficiency profile of the underlying scheme.
If instantiated with KZG polynomial commitments~\cite{Kate2010-px}, we achieve a constant storage overhead independent of the input size and a constant verification time.
Although the public setup parameters of KZG are of size $O(\sNumSamples)$,
we can consider them as system parameters and reuse them for the input of each party~\cite{Kate2010-px}.
Hence, our protocol only requires a storage overhead linear in the size of the input and the number of parties, i.e., $O(\sNumPartiesGenericParty + \sNumSamples)$.
If used with an inner-product argument-based polynomial commitment, the commitment size could also be made constant (e.g., a single \acrfull{pvc}).
However, the verification time would be linear in the input size~\cite{Bootle2016-ff}.
\fakeparagraph{EC-MPC.}
Prior work has observed that most secret-sharing-based MPC protocols for finite field arithmetic generalize to arithmetic circuits involving elliptic curve points~\cite{Ozdemir2022-aw,Smart2019-mpcec}.
Using such protocols with additional support for computations over an elliptic curve group $\sgroup$ of order $p$ (which we denote as \sabbec) when instantiating \sPoCLabel offers a significant improvement to performance.
Specifically, we can accelerate the execution of our \sPoCCommit algorithm under secure computation (i.e., the creation of output commitments).

\fakeparspacingtop
\fakeparagraph{Batch verification}
Our protocol $\sprotocolCC$ allows the verifier to check the integrity of the prover's input by verifying one pairing equation.
However, the verifier still needs to perform this check for each input party.
We can optimize this further for KZG commitments by leveraging their homomorphic property~\cite{Kate2010-px,Gabizon2019-plonk}~(\rdef{homomorphic_commitment}\citeextended).
Let $\scommitment_1, \ldots, \scommitment_{\sNumPartiesGenericParty}$ be the set of commitments and $\rho_1, \ldots, \rho_{\sNumPartiesGenericParty}$ the set of
target evaluations for each prover $\sprover_1,\ldots,\sprover_{\sNumPartiesGenericParty}$ at a common random point $\beta \in \sfield_p$ from~\rstep{cc:rho} of the consistency check.
The verifier first computes a random linear combination of the commitments as $\tilde{\scommitment} \defeq \sum_i^{\sNumPartiesGenericParty} \gamma^i \scommitment_i$
for a randomly sampled $\gamma \in \sfield_p$, as well as the corresponding evaluation $\tilde{\rho} \eqdef \sum_i^{\sNumPartiesGenericParty} \gamma^i \rho_i$
and aggregate proof $\tilde{\pi} \eqdef \sum_i^{\sNumPartiesGenericParty} \gamma^i \pi_i$.
The verifier can then check this aggregated commitment using $\mathsf{PC.Check}(\mathsf{pp}, \tilde{\scommitment}, \beta, \tilde{\rho}, \tilde{\pi})$.
This allows the verifier to check, in the optimistic case, only one pairing equation instead of $\sNumPartiesGenericParty$  at the cost of a negligible statistical error.
Security follows from the fact that the aggregated polynomial commitment $\tilde{\scommitment}$ will only agree with the aggregated evaluation point $\tilde{\rho}$
at a random point $\beta$ with negligible probability due to the Demillo-Lipton-Schwartz-Zippel Lemma~(\cite{Demillo1978-SZDL}).
If verification passes, this implies that all commitments open to the correct evaluation point with overwhelming probability.
If verification fails, this must mean that at least one of the commitments is inconsistent with high probability.
In this case, the verifier can proceed to check the commitments and proofs individually.

\subsecspacingtop
\section{Auditing Functions}

\begin{table}
    \centering
    \setlength{\tabcolsep}{3pt}
    \begin{tabular}{lccc}
        \toprule
        & \textsf{T} & \textsf{M} & \textsf{I} \\
        \midrule
        Data Validation & & & \\
        \quad Input Checks~\cite{Chang2023-mk, Lycklama2023-cx} & \CIRCLE & \Circle & \Circle \\
        \quad Sample Attribution~\cite{Koh2017-by, Khanna2019-FisherKernel, Jia2019-ShapleyValue, Shan2021-ad, Hammoudeh2022-influence_auditing} & \CIRCLE & \CIRCLE & \CIRCLE \\
        \quad Party attribution~\cite{Lycklama2022-CamelMLSafety,utrace} & \CIRCLE & \CIRCLE & \CIRCLE \\

        \midrule

        Model Validation & & & \\
        \quad Validation Sets~\cite{Choi2023-zb} & \Circle & \CIRCLE & \Circle \\
        \quad Feature Attribution~\cite{Ribeiro2016-lime, Lundberg2017-unifiedmodelpredictions, Jetchev2023-xorshap} & \LEFTcircle & \CIRCLE & \CIRCLE \\
        \quad Certification~\cite{Jovanovic2022-tz, Kilbertus2018-zl, Segal2020-tu} & \Circle & \CIRCLE & \CIRCLE \\

        \midrule

        Process Validation & & & \\
        \quad Algorithm Verific.~\cite{Jia2021-PoL,Garg2023-bn,Sun2023-zkdl,Kang2022-zkml} & \CIRCLE & \CIRCLE & \Circle \\
        \quad Constraint Verific.~\cite{Shamsabadi2022-us} & \CIRCLE & \CIRCLE & \Circle \\

        \bottomrule
    \end{tabular}
    \caption{
    A priori and post hoc algorithms from the \gls{ml} interpretability and safety literature along with whether they require the training data (\textsf{T}), the model (\textsf{M}) and the inference (\textsf{I}) as input.
    }
    \vspace{-0.5em}
    \ltab{auditingfunctions}
\end{table}

\subsecspacingbot
\lsec{functions}
{
\definecolor{blond}{rgb}{0.98, 0.94, 0.75}
\newcommand{\secret}[1]{{\setlength{\fboxsep}{1.5pt}\colorbox{blond}{#1}}}
\let\x\spredictionX
\let\y\spredictionY
\let\smodelOld\smodel
\newcommand{\spartyinputI}{\secret{$\spartyinput_1$}}
\newcommand{\spartyinputII}{\secret{$\spartyinput_2$}}
\newcommand{\spartyinputN}{\secret{$\spartyinput_N$}}
\newcommand{\spartyinputs}{\spartyinputI$, \ldots $\spartyinputN}
\renewcommand{\vec}[1]{\mathbf{#1}}

Until now, our discussion has centered around the cryptographic protocol that enables robust and secure audits of private models.
We now focus on how we realize the audit functionality enabled by our framework.
The algorithmic side of auditing for ML is an active area, and alternative instantiations that enable different properties exist or are actively being developed (see \rtab{auditingfunctions} for an overview).
In our work, we focus on post-hoc audits, specifically  accountability (sample/party attribution), explainability (feature attribution), and robustness \& fairness (certification).
While \oursystem's design is highly extensible with additional auditing functionalities (similar to how we support a wide range of \gls{ppml} systems), here we focus our discussion on auditing functions currently implemented in our framework.
\ifdefined\isnotextended
In Appendix~\ref{sec:functions_extra} in the extended version of the paper~\cite{arc_full}, we provide a more detailed description of each function.
\else
In Appendix~\ref{sec:functions_extra}, we provide a more detailed description of each function.
\fi

\vspace{-0.5em}
\subsection{Robustness \& Fairness}
\lsec{functions:robustnessfairness}
The community has devised a range of techniques to show that a model is robust against adversarial examples~\cite{Gehr2018-certifiedrobustness, Xu2020-certifiedrobustness}, i.e., that the model is stable to small variations in the input.
While \emph{global} robustness guarantees are more naturally realized as a priori checks, \emph{local robustness}~\cite{Cohen2019-cn},
which certifies that a model consistently produces the same prediction $y$ for all inputs within a radius $R$ around the original input $x$, naturally suits the post-hoc auditing setting we consider.
In \oursystem, we adapt the algorithm proposed by Jovanovic et al.~\cite{Jovanovic2022-tz} for \gls{fhe} to the \gls{mpc} setting.
The algorithm samples $n$ perturbed inputs around the input $x$ by adding Gaussian noise and obtaining predictions for these samples.
Finally, a statistical check is conducted to assess whether the obtained prediction $y$ remains invariant to these perturbations with high probability.
The output of the auditing function is a boolean indicating whether the model is locally robust with confidence $1-\alpha$.
We can extend the same technique to achieve fairness guarantees,
as there is a well-established connection between robustness and individual fairness~\cite{Dwork2011-fairness, Yurochkin2019-fairness, Ruoss2020-certifiedfairness, Jovanovic2022-tz}.
It is sufficient to change the closeness metric of the sampling procedure to generate perturbed inputs that are close to $x$ in terms of fairness (i.e., inputs that should be treated similarly have a small distance).

\vspace{-0.5em}
\subsection{Accountability}
\lsec{functions:accountability}

We consider two flavors of accountability mechanisms that attribute responsibility for decision made by a model:
\emph{sample} attribution identifies the influence of individual data samples on a prediction, while  \emph{party} attribution merely provides the relative influence of each \gls{r:inputparty}'s dataset, providing auditability while revealing less information.
A variety of methods to identify the impact of individual data samples exist~\cite{Koh2017-by, Khanna2019-FisherKernel, Jia2019-ShapleyValue, Shan2021-ad, Hammoudeh2022-influence_auditing}, however, some (e.g., influence functions~\cite{Koh2017-by}) require substantial computational resources (e.g., inverting the Hessian matrix of the loss function) which makes them prohibitively expensive under secure computation.
In \oursystem, we leverage an approach using \gls{f:knnshapley} values~\cite{Jia2019-ShapleyValue}, i.e., Shapley values of a KNN classifier on the training data's latent space representation.
As there exists a closed-form formulation of the Shapley values for KNN~\cite{Jia2019-ShapleyValue}, this allows an efficient MPC realization.
For party attribution, \oursystem uses an efficient unlearning approach~\cite{Lycklama2022-CamelMLSafety}.
The key idea here is that if a suspicious prediction $(x, y)$ was (at least partially) the result of data provided by a \gls{r:inputparty}, then excluding that party's data will lead to the absence (or weakening) of the suspicious prediction.
Rather than recomputing leave-out models from scratch, we use an algorithm~\cite{Lycklama2022-CamelMLSafety} that uses unlearning of the data of a party from the original model~\cite{Shan2021-ad},
which requires only a small number of training epochs until sufficient differences are detectable.

\vspace{-0.5em}
\subsection{Explainability}
\lsec{functions:explainability}
A wide range of methods has been proposed to explain the predictions of complex models~\cite{Ribeiro2016-lime, Shrikumar2017-deeplift, Lundberg2017-unifiedmodelpredictions}. 
Of these, we consider additive feature attribution methods~\cite{Lundberg2017-unifiedmodelpredictions} as particularly suitable for privacy-preserving auditing as they provide an attractive trade-off between leakage and utility. 
These methods highlight which features of a prediction sample are most influential for the prediction, even for complex \gls{ml} models. 
They achieve this by approximating the target 
model's behavior locally (around the given prediction) with a simple and explainable (e.g., linear) model.
In \oursystem, 
we leverage KernelSHAP~\cite{Lundberg2017-unifiedmodelpredictions} to approximate the local decision boundary of the classifier,
through a linear regression on the features. 
We sample points around the prediction sample and weigh them based on their distance to the sample as measured by the Shapley kernel.
As a result, the regression coefficients directly correspond to the Shapley values of the features.

} %

\subsecspacingtop
\section{Evaluation}
\subsecspacingbot
\lsec{eval}

In this section, we evaluate the performance of \oursystem
in the training, inference and auditing phases for different workloads and auditing functions.
We evaluate the overhead of our protocol when instantiated with different approaches to the consistency layer \sPoCLabel.
For training and auditing, we focus on the \gls{mpc} versions of our protocol, as these are the most established forms of verifiable \gls{ml} computation.

\fakeparspacingtop
\fakeparagraph{Implementation.}
Our implementation is based on MP-SPDZ~\cite{Keller2020-mpspdz}, a popular framework for MPC computation that supports a variety of protocols.
We extend MP-SPDZ with protocols for share conversion and elliptic curve operations on the pairing-friendly BLS12-377 curve~\cite{Bowe2018-tw}
provided by the \texttt{libff} library~\cite{libff}.
We use ECDSA signatures on the secp256k1 curve~\cite{curve-secp256k1} for which a distributed signing protocol
was previously implemented in MP-SPDZ~\cite{Dalskov2020-mm}.
For the evaluation proofs of the polynomial commitments, we use the implementation of the KZG polynomial commitment scheme~\cite{Kate2010-px} provided by Arkworks' \texttt{poly-commit} library~\cite{arkworks}.
We perform share conversion (\cfref{Appendix~\ref{sec:share_conversion}}\citeextended) to convert between $\sring_{2^{64}}$ and the scalar field $\sfield_{\text{BLS12-377}}$.
The MPC computations for \gls{ml} training, inference, and auditing functions are expressed in MP-SPDZ's domain-specific language.
We rely on the higher-level \gls{ml} primitives that MP-SPDZ provides that use mixed-circuit computation.
Note that we perform exact truncation instead of probabilistic truncation for fixed-point multiplication because the latter has recently been shown to be insecure~\cite{Li2023-ay}.

To compare the performance of our consistency layer to other approaches (\cfref{\ref{sec:design:consistency:approaches}}),
we additionally implement a version of \sPoCLabel based on the SHA3-256 cryptographic hash function denoted by \gls{a:sha3}
that internally uses the Bristol-fashion circuit implementation of the Keccak-f sponge function~\cite{nist_sha3}.
We also implement a version \gls{a:ped} based on Cerebro~\cite{Zheng2021-tb} that uses Pedersen commitments by adapting the open-source implementation provided by the authors.
We make the implementation of our protocol and the various consistency layers that we evaluate available as open-source\footnote{\href{https://github.com/pps-lab/arc}{github.com/pps-lab/arc}}.

\fakeparspacingtop
\fakeparagraph{Experimental Setup.}
We run \oursystem on a set of AWS \texttt{c5.9xlarge} machines running Ubuntu 20.04, each equipped with 36 vCPUs of
an Intel Xeon 3.6~Ghz processor and 72~GB of RAM.
For the BERT transformer model, we use \texttt{c5.24xlarge}, featuring 96 vCPUs of an Intel Xeon 3.6~Ghz processor and 192~GB of RAM, as it provides sufficient memory to fit the significantly larger model.
The machines are connected over a \gls{lan} through a 12~Gbps network interface with an average \gls{rtt} of $0.5$~ms.
We additionally perform our experiments in a simulated \gls{wan} setting using the \texttt{tc} utility to introduce an \gls{rtt} of $80$~ms and limit the bandwidth the $2$~Gbps.
We report the total wall clock time and the total communication cost in terms of the data sent by each party.
This includes the time and bandwidth required for the online phase and the preprocessing phase that sets up the correlated randomness necessary for the \gls{mpc} protocol.
We also report the storage overhead for which we apply log scaling as the overhead varies significantly between different \sPoCLabel approaches and settings.
In experiments in the \gls{wan} setting and those involving maliciously secure protocols, we estimate the \gls{ml} training operations based on 5 and 50 batches of gradient descent, respectively.
For the BERT transformer model, we extrapolate further due to the significant size of the model and training data.
For the related work, which has overhead linear in the size of the input, we also extrapolate results for some of the larger instances.

We evaluate the computational phases in the \gls{3pc} setting with a maliciously secure-with-abort protocol that combines SPDZ-wise redundancy with replicated secret sharing over a 64-bit ring~\cite{Dalskov2021-rc}.
We also evaluate the performance of a semi-honest protocol based on replicated secret sharing.
These protocols are representative of the most efficient MPC protocols in the malicious and semi-honest settings for \gls{ml} workloads.
We apply an optimization for the auditing phase that uses the fact that all inputs in this phase are authenticated using commitments.
This allows us to optimistically use a security-with-abort protocol and, only, if the protocol aborts, restart the computation with a less efficient identifiable-abort protocol
with the guarantee that this execution uses the same inputs.
We choose the \gls{3pc} setting because it allows for the most efficient MPC protocols, favoring
\gls{a:ped} and \gls{a:sha3} whose \sPoCEvalNoLabel relies more heavily on MPC computation.
Other settings such as \gls{2pc} or non-optimistically executing the auditing would require more expensive MPC protocols, resulting in a higher overhead for the computation.
This would increase the relative overhead of the related approaches compared to ours.

\fakeparagraph{Scenarios.}
We evaluate four auditing functions from the previous section, (i) \gls{f:robustness}, (ii) \gls{f:fairness} (\cfref{\rsec{functions:robustnessfairness}}), (iii) \gls{f:knnshapley} (\cfref{\rsec{functions:accountability}}) and
(iv) \gls{f:shap} (\cfref{\rsec{functions:explainability}}),
on the following models and datasets.
\begin{itemize}[label={},leftmargin=0pt,align=left,labelsep=0pt,labelwidth=!]
    \item \gls{sc:adult}: A logistic regression model with 3k parameters trained on the Adult~\cite{Becker1996-adult} binary classification task for 10 epochs to predict whether a person's income exceeds \$50k per year.
    \item \gls{sc:mnist}: A LeNet model consisting of 431K parameters, referred to as `model C' in prior work~\cite{Keller2022-quantizedtraining, Wagh2019-SecureNN}
    trained on the MNIST image classification task~\cite{Lecun1998-wc} for 20 epochs.
    \item \gls{sc:cifar}: A variant of AlexNet~\cite{Krizhevsky2012-alexnet} as used in Falcon~\cite{Wagh2021-Falcon}, comprising
    3.9 million parameters trained on the CIFAR-10 image classification task~\cite{Krizhevsky2009-cifar10} for 40 epochs.
    \item \gls{sc:qnli}: A BERT transformer model~\cite{Devlin2019-bert} with sequence length 128, comprising 85 million parameters, finetuned for one epoch on 2500 samples of the Stanford Question Answering reading comprehension dataset~\cite{Wang2018-glue,qnli}.
\end{itemize}

{
\newcommand*{\TextVCenter}[1]{%
  \text{$\vcenter{\hbox{#1}}$}%
}
\newcommand{\mpc}{\TextVCenter{\footnotesize{ (mpc)}}}
\newcommand{\local}{\TextVCenter{\footnotesize{ (local)}}}

\begin{table}[t]
    \centering
    \setlength{\tabcolsep}{5pt}
     \begin{tabular}{lccc}
        \toprule
        & Storage & Prover Comp. & Verifier Comp. \\ \midrule
        \gls{a:sha3} & $N$ & - & $d \cdot N \mpc$ \\ \midrule
        \gls{a:ped} & $d \cdot N$ & - & $d \cdot N \local + N \mpc$ \\ \midrule
        \textbf{Ours} & $N$ & $d\local$ & $N \mpc$ \\ 
        \bottomrule
    \end{tabular}
    \caption{
        Asymptotic computational complexity (in big-$O$ notation) and storage requirements of the consistency approaches we evaluate for a computation with $N$ input parties each with
        $d$ input elements.
        We differentiate between local and \gls{mpc} computation. 
        For \gls{a:sha3}, we consider Keccak-f operations. 
        For \gls{a:ped} and ours, we consider group operations.
    }
    \ltab{complexity}
\end{table}

}

\subsection{Evaluation Results}
We evaluate the overhead that \oursystem imposes on the training, inference and auditing phases.
The main overhead of the consistency layer in training and inference consists of two parts:
Verifying the inputs of the computation using \sPoCEval and, afterwards,
computing the output commitments using \sPoCCommit.
Other components, such as those related to the signatures are negligible in comparison:
distributed signing takes at most 300ms for the \gls{wan} and verifying a signature is a local operation taking 1ms.
\Glspl{r:auditrequester} only have to store a single ECDSA signature of 64 bytes for the \gls{r:modelowner} and each \gls{r:inputparty}, and a joint signature of 64 bytes for the \glspl{r:icomputer} and the \glspl{r:tcomputer}.
We provide an analysis of the asymptotic complexity of the different approaches in~\rtab{complexity} and focus on the concrete performance numbers for the rest of this section.

\newlength{\infplotheight}

\begin{figure}[!t]
    \includegraphics[width=1.0\linewidth]{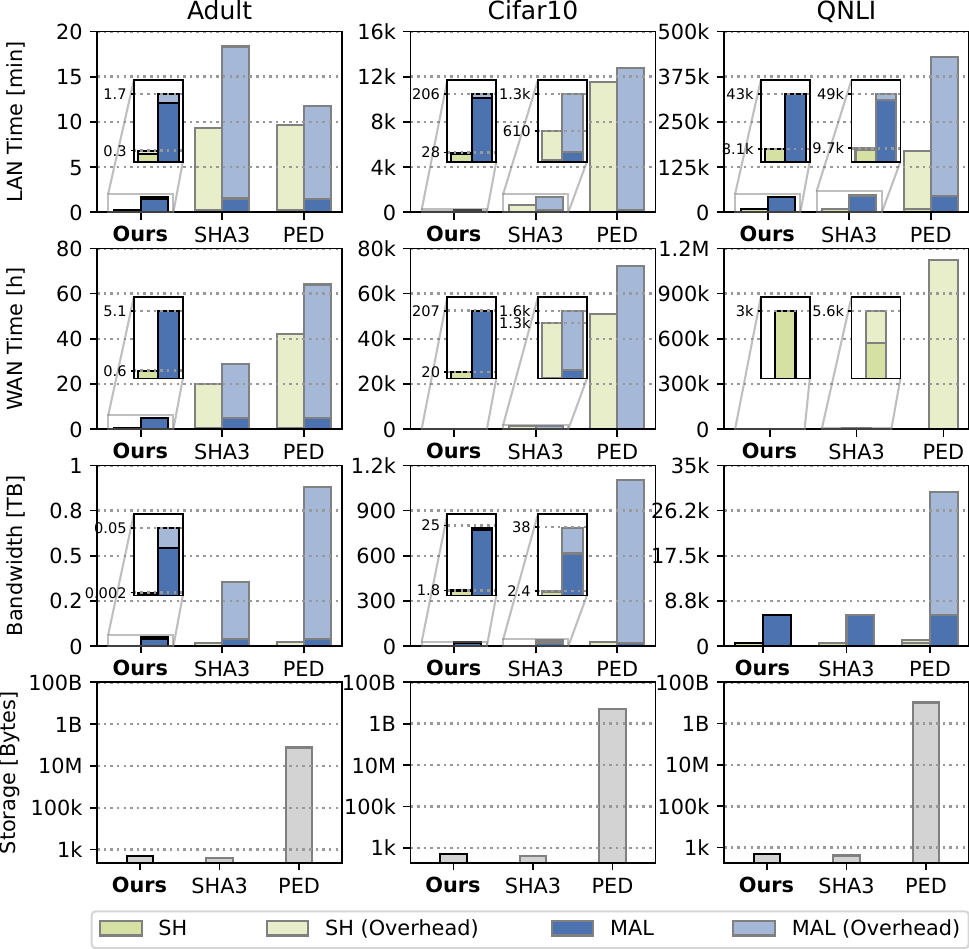} %
    \figcaptionvspace
    \caption{Evaluation of \oursystem comparing the approaches relative to a single epoch of \gls{ppml} training.}
    \figcaptionvspacebot
    \lfig{eval:training}
\end{figure}

\begin{figure}[t]
    \includegraphics[width=1.0\linewidth]{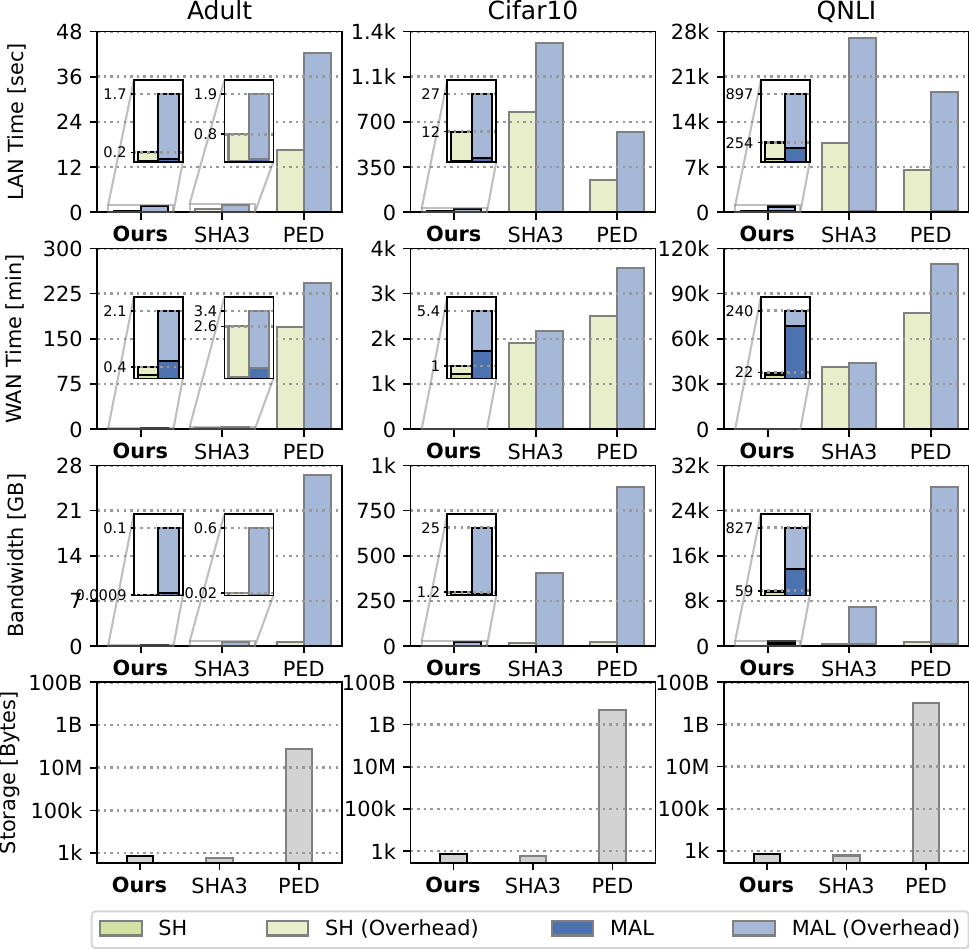} %
    \figcaptionvspace
    \caption{The overhead of our system's consistency protocol relative to a single \gls{ppml} inference for our three scenarios.}
    \figcaptionvspacebot
    \lfig{eval:inference}
\end{figure}

\begin{figure*}[!t]
    \includegraphics[width=1.0\textwidth]{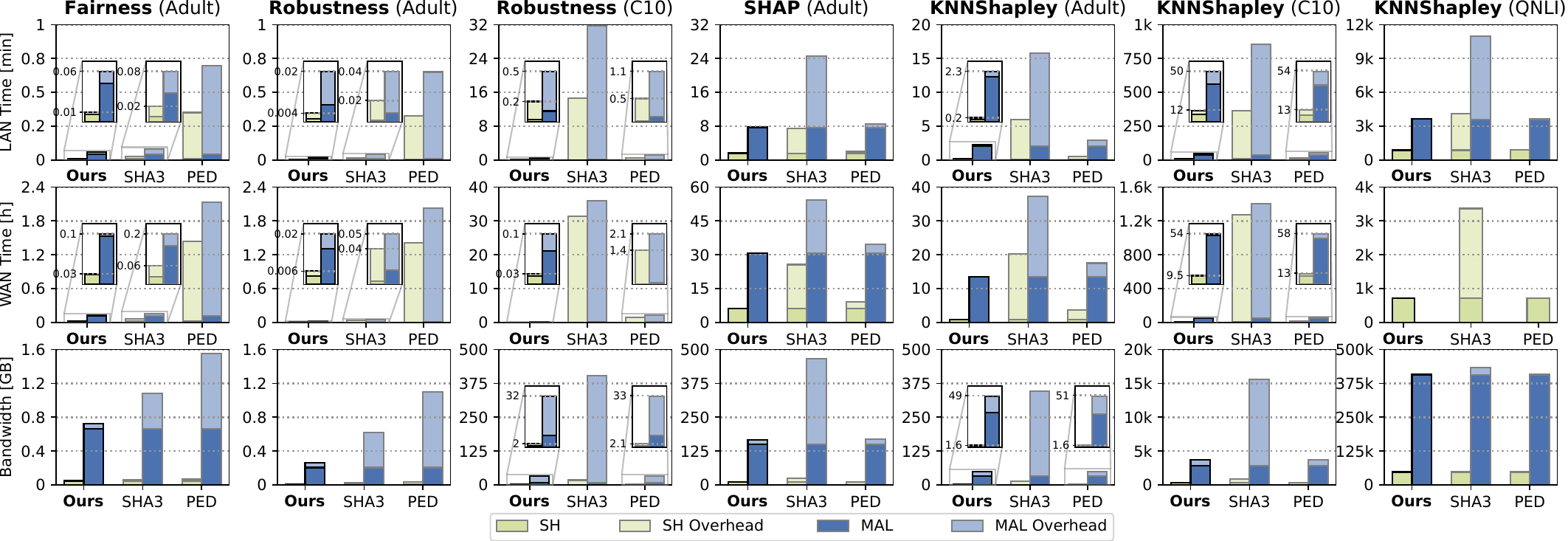} %
    \figcaptionvspace
    \caption{The overhead of \oursystem's consistency layer relative to the cost of the auditing function computation in \gls{mpc} for four different auditing functions across our three scenarios.}
    \figcaptionvspacebot
    \lfig{eval:auditing}
\end{figure*}

\fakeparspacingtop
\fakeparagraph{Training.}
We show the wall-clock time and bandwidth of a single training epoch in~\rfig{eval:training}, differentiating between the overhead induced by our framework and the cost of the underlying \gls{ppml} training.
In total, \gls{ppml} training \gls{sc:adult} takes \evalnum{119}~seconds, with consistency adding 5~seconds for ours, 545~seconds for \gls{a:sha3}, and 570~seconds for \gls{a:ped}.
For \gls{sc:mnist}, this is \evalnum{123} minutes (+~1.3~minutes, 2.9~hours, or 22~hours), for \gls{sc:cifar}, \evalnum{15} hours (+~4.1~minutes, 9.8~hours, or 8~days), and for \gls{sc:qnli}, this is \evalnum{5.4} days (+~11 minutes, 1.1 days, or 16 weeks), in the semi-honest LAN setting.
We refer to Appendix~\ref{sec:apx:evalextra}\citeextended for a full report of end-to-end training results and further results for \gls{sc:mnist} which have been omitted here due to space constraints.
Note that we do not evaluate  \gls{sc:qnli} training in the malicious WAN setting, as this is beyond the current state of the art for \gls{ppml}.
Since the bandwidth overhead is not significantly affected by network delays, we only present the bandwidth results once.
As storage is, additionally, also independent of the chosen MPC protocol, we do not differentiate between semi-honest and malicious settings for storage.

We observe that the overhead of the baseline approaches varies significantly for training.
The timing overhead induced by \gls{a:ped} is \evalnum{66-500x} compared to training but only \evalnum{0.1-26x} with \gls{a:sha3} in the \gls{lan} setting.
In the \gls{wan} setting the relative overhead further increases to \evalnum{at most three orders of magnitude} because of the large number of \gls{mpc} round-trips required to compute the operations related to hash functions and elliptic curve operations. 
In comparison, our consistency check protocol, which outperforms the related approaches across all configurations, introduces only \evalnum{0.001-1.35x} overhead in the \gls{lan} and less than \evalnum{1.02x} in the \gls{wan} setting.
This is because our protocol features more local computation than MPC computation and is therefore less impacted by the slowdown induced by network delays.
Overall, we conclude that the overhead of our approach, and also of \gls{a:sha3}, is effectively negligible in the context of PPML training.
In contrast, we observe that the \gls{a:ped} approach is infeasible for all but the simplest models. 
The primary cost of \gls{a:ped} is the time required to compute the individual Pedersen commitments to the model parameters; the overhead during verification is much smaller as this only involves computing a commitment for each of the three input parties. 
Finally, the bandwidth and storage required for \gls{a:ped} are significant as it needs to commit to each input element individually.
Meanwhile, the storage overhead of our approach is independent of the dataset and model sizes and similar to that of the hash-based approach with \evalnum{496} bytes compared to \evalnum{416} bytes for \gls{a:sha3}.

\fakeparspacingtop
\fakeparspacingtop
\fakeparagraph{Inference.}
Model inference is a significantly smaller operation than training, resulting in a larger relative consistency overhead  
As shown in~\rfig{eval:inference}, the consistency operations are at least an order of magnitude slower than the inference itself, even for our approach.
However, in absolute numbers, the overhead introduced by our approach is very small across all configurations, ranging from a few seconds to a few minutes for the large transformer model.
In contrast, the overhead for \gls{a:sha3} quickly becomes prohibitive for all but the smallest models, already requiring over thirty hours in the WAN setting for a single \gls{sc:cifar} inference.
\gls{a:ped} outperforms \gls{a:sha3} in some cases, it induces similarly prohibitive overheads across all configurations.

A significant fraction, \evalnum{35-66\%}, of the overhead in our approach is the result of the share conversion from the \gls{ppml} protocol's computation domain $\sring_{2^{64}}$ to the
scalar field domain $\sfield_{\text{BLS12-377}}$.
The conversion requires a bit decomposition and re-composition for each input parameter which is expensive and scales linearly in the input size.
Although this overhead is significant, it is small concretely, with \evalnum{$25$} seconds in the active security setting for \gls{sc:cifar}.
In the case that lower latency is required, the $\sring_{2^{64}}$ secret shares of the model can be cached on the inference servers after a single conversion to $\sfield_{\text{BLS12-377}}$ and verification with \sPoCEval.
When scaling to larger models, \gls{a:ped} and \gls{a:sha3} become prohibitively expensive concretely with a \evalnum{250-6000x} slowdown compared to a single inference.
As in training, we observe that the storage required for the \gls{r:auditrequester} receipts is comparable for our approach and 
\gls{a:sha3}, resulting in only \evalnum{720} bytes and \evalnum{608} bytes, respectively, per prediction.
Most importantly, these are independent of the model and training data size.
Meanwhile, \gls{a:ped} requires storage that is linear in the number of input elements, e.g., requiring \evalnum{5}GB for \gls{sc:cifar} and \evalnum{11}GB for \gls{sc:qnli}.
This results in \gls{a:ped} inducing prohibitive storage requirements for all but the simplest models.

\fakeparagraph{Auditing.}
We present the wall-clock time and bandwidth overhead for different auditing functions in~\rfig{eval:auditing}.
Note that we do not evaluate auditing on \gls{sc:qnli} for the malicious WAN setting, as the complexity of evaluating auditing functions on such a complex model in this setting is pushing the boundaries of what is possible with current PPML techniques.  
Across all settings, \oursystem significantly outperforms related approaches in terms of runtime, with a storage overhead comparable to the hash-based approach.
As we move to larger input sizes, for instance in the case of \gls{f:knnshapley} that considers the full training dataset, the main cost of our approach after share conversion is the \gls{msm} required to compute the opening proof of the polynomial commitment.
Each prover party must compute an \gls{msm} that is linear in the size of its input.
Due to the properties of KZG, the other parties only have to check one pairing equation per prover, which we can further reduce to a single pairing equation due to the batch verification (c.f. \rsec{ourprotocol}).
We also observe that, for larger models, \gls{a:ped} outperforms \gls{a:sha3} and approaches the performance of our approach.
This is primarily because there is no need for computing and storing commitments to the output during the auditing phase, sidestepping the major weaknesses of the \gls{a:ped} approach.
In fact, as model sizes increase, the high constant overhead of the Pedersen commitments becomes less noticeable.
However, our approach continues to outperform \gls{a:ped} both asymptotically (\cfref{\rtab{complexity}}) and also concretely across all configurations.

\noindent
In conclusion, we observe that our approach outperforms the related work both asymptotically and concretely across all configurations. 
While \gls{a:ped} and \gls{a:sha3} approach the performance of our solution in some phases (for certain configurations),
they remain prohibitively expensive from an end-to-end perspective.
Meanwhile, we demonstrate that \oursystem instantiated with our efficient \sPoCLabel protocol is highly practical across a wide range of settings, including scaling up to large and complex models that push the boundary of the current  state of the art in privacy preserving machine learning.

 \flushcolsend

\ifdefined\isnotanon
\vspace{-0.5em}
\section*{Acknowledgements}
We thank the anonymous reviewers for their insightful input and feedback.
We would also like to acknowledge our sponsors for their generous support, including Meta, Google, and SNSF through an Ambizione Grant No. PZ00P2\_186050.

\else
\fi

\newpage
\bibliographystyle{plain}
{\small

\bibliography{references, additional_references}

\flushcolsend
}

\appendix
\newpage
\appendix

\ifdefined\isnotextended
\vspace{1em}
\noindent
We refer to the extended version of the paper~\cite{arc_full} for the remaining appendices.
\else

\fi

\begin{myhideenv}

\section{Definitions}
\lsec{apx:proofs:definitions}

 \begin{definition}[Commitment Scheme]
     \ldef{commitments}
     A non-interactive commitment scheme consists of a message space $\scomMessageSpace$, randomness space $\scomRandomnessSpace$,
     a commitment space $\scomCommitmentSpace$
     and a tuple of polynomial-time algorithms $(\sComSetup,\sComCommit,\sComVerify)$
     defined as follows:
     \begin{algos}
         \item $\sComSetup(1^\lambda) \rightarrow \pp$: Given a security parameter $\lambda$, it outputs public parameters $\pp$.

         \item $\sComCommit(\pp, m, r) \rightarrow c$: Given public parameters $\pp$, a message $m \in \scomMessageSpace$ and randomness $r \in \scomRandomnessSpace$, it outputs a commitment $c$.

         \item $\sComVerify(\pp, c, r, m) \rightarrow \{0,1\}$: Given public parameters $\pp$, a commitment $c$, a decommitment $r$, and a message $m$, it outputs $1$ if the commitment is valid, otherwise $0$.
     \end{algos}
     A non-interactive commitment scheme has the following properties:
     \begin{algos}
         \item \fakeparagraph{Correctness.}
         For all security parameters $\lambda$, for all $m$ and for all $\pp$ output by $\sComSetup(1^\lambda)$, if $c = \sComCommit(\pp, m, r)$, then $\sComVerify(\pp, c, m, r) = 1$.

         \item \fakeparagraph{Binding.}
         For all polynomial-time adversaries $\mathcal{A}$, the probability
         \begin{equation*}
             \begin{split}
                 \Pr\bigl[\sComVerify(\pp, c, m_1, r_1) &= 1 \land \\
                 \sComVerify(\pp, c, m_2, r_2) &= 1 \land m_1 \neq m_2 : \\
                  \pp \leftarrow \sComSetup(1^\lambda), (c, r_1, r_2, & m_1, m_2) \leftarrow \mathcal{A}(\pp) \bigl]
             \end{split}
         \end{equation*}
         is negligible.

         \item \fakeparagraph{Hiding.}
         For all polynomial-time adversaries $\mathcal{A}$, the advantage
         \begin{equation*}
             \begin{split}
                 |\Pr[\mathcal{A}(\pp, c) = 1 : c &\leftarrow \sComCommit(\pp, m_1, r)] - \\
                 \Pr[\mathcal{A}(\pp, c) = 1 : c &\leftarrow \sComCommit(\pp, m_2, r)]|
             \end{split}
         \end{equation*}
         is negligible, for all messages $m_1, m_2$.
     \end{algos}
 \end{definition}

\begin{definition}[Polynomial Commitments~\cite{Kate2010-px}]
    \ldef{pc}
    Polynomial commitments enable a prover to commit to a polynomial in such a way that they can later reveal the polynomial's value at any particular point, with a proof that the revealed value is indeed correct.
    These commitments are notably used to construct succinct zero-knowledge proofs and verifiable computation protocols.
    A polynomial commitment scheme is a quadruple (\PCSetup, \PCCommit, \PCProve, \PCCheck) where
    \begin{algos}
        \item $\PCSetup(d) \rightarrow \texttt{pp}$: prepares the public parameters given the maximum supported degree of polynomials $d$ and outputting a common reference string $\texttt{pp}$.
        \item $\PCCommit(\texttt{pp}, \spolynomial, \sRandomness) \rightarrow c$: computes a commitment $c$ to a polynomial $\spolynomial$, using randomness $\sRandomness$.
        \item $\PCProve(\texttt{pp}, \scommitment, \spolynomial, \sRandomness, x, y) \rightarrow \pi$: The prover computes a proof $\pi$ using randomness $\sRandomness$ that $c$ commits to $\spolynomial$ such that $\spolynomial(x) = y$.
        \item $\PCCheck(\texttt{pp}, \scommitment, x, y, \pi) \rightarrow \{ 0, 1 \}$: The verifier checks that $c$ commits to $\spolynomial$ such that $\spolynomial(x) = y$.
    \end{algos}

    \noindent A polynomial commitment scheme is secure if it provides correctness, polynomial binding, evaluating binding, and hiding properties. 
    We refer to~\cite{Kate2010-px} for a formal definition of these properties.
\end{definition}
\begin{definition}[KZG Commitments~\cite{Kate2010-px}]
    \ldef{kzg}
    KZG commitments leverage bilinear pairings to create a commitment scheme for polynomials where the commitments have constant size.
    Let $\sgroup_1$, $\sgroup_2$ and $\sgroup_T$ be cyclic groups of prime order $p$ such with generators $\sgenerator_1 \in \sgroup_1$ and
    $\sgenerator_2 \in \sgroup_2$.
Let $e: \sgroup_1 \times \sgroup_2 \rightarrow \sgroup_T$ be a bilinear pairing, so that $e(\alpha \cdot \sgenerator_1, \beta \cdot \sgenerator_2) = \alpha \beta \cdot e(\sgenerator_1, \sgenerator_2)$.
The KZG polynomial commitment scheme for some polynomial $\spolynomial$ made up of coefficients $\spolynomial_i$ is defined by four algorithms:
    \begin{algos}
        \item $\PCSetup(d)$: Sample $\alpha \sample \sfield_p$ and output
        \begin{equation*}
            \texttt{pp} \leftarrow \left( \alpha \cdot \sgenerator_1, \ldots, \alpha^d \cdot \sgenerator_1, \alpha \cdot \sgenerator_2 \right)
        \end{equation*}
        \item $\PCCommit(\texttt{pp}, \spolynomial)$: Output $\scommitment = \spolynomial(\alpha) \cdot \sgenerator_1$, computed as
        \begin{equation*}
            \scommitment \leftarrow \sum_{i=0}^d \spolynomial_i \cdot (\alpha^i \cdot \sgenerator_1)
        \end{equation*}
        \item $\PCProve(\texttt{pp}, \scommitment, \spolynomial, x):$ Compute the remainder and quotient
        \begin{equation*}
            q(X),r(X) \leftarrow \left( \spolynomial(X) - \spolynomial(x) \right) / \left(X - x \right).
        \end{equation*}
        Check that the remainder $r(X)$ and, if true, output $\pi = q(\alpha) \cdot \sgenerator_1$, computed as $\sum_{i=0}^d \left( q_i \cdot (\alpha^i \cdot \sgenerator_1) \right)$.
        \item $\PCCheck(\texttt{pp}, \scommitment, x, y, \pi)$: Accept if the following pairing equation holds:
        \begin{equation*}
            e(\pi, \alpha \cdot \sgenerator_2 - x \cdot \sgenerator_2) = e(\scommitment - y \cdot \sgenerator_1, \sgenerator_2)
        \end{equation*}
    \end{algos}

    \noindent The security properties of KZG commitments fundamentally rely on the hardness of the polynomial division problem.
    The parameter $\alpha$ acts as a trapdoor and must be discarded after \texttt{PC.Setup} to ensure the binding property.
    Hence, we require a trusted setup to generate the public parameters and securely discard $\alpha$,
    which can be computed using MPC or,
    depending on the deployment, computed by the auditor acting as a trusted dealer.
    The hiding property relies on the discrete logarithm assumption, so if $\alpha$ is not discarded this breaks the binding property but not the hiding property.
    We refer to~\cite{Kate2010-px} for a detailed security analysis.

\end{definition}

\begin{definition}[Homomorphic Commitment Scheme~\cite{Bunz2018-mg}]
\ldef{homomorphic_commitment}
A homomorphic commitment scheme is a non-interactive commitment scheme such that
\scomMessageSpace, \scomRandomnessSpace and \scomCommitmentSpace are all abelian groups and
for all $m_1, m_2 \in \scomMessageSpace$ and $r_1, r_2 \in \scomRandomnessSpace$, we have
\begin{equation*}
    \begin{split}
        &\sComCommit(\pp, m_1 + m_2, r_1 + r_2) = \\
        &\sComCommit(\pp, m_1, r_1) + \sComCommit(\pp, m_2, r_2).
    \end{split}
\end{equation*}
\end{definition}

\noindent KZG commitments are homomomorphic, i.e.,
if $\scommitment_1$ and $\scommitment_2$ are commitments to polynomials $\spolynomial_1$ and $\spolynomial_2$,
then $\scommitment_1 + \scommitment_2$ is a commitment to polynomial $\spolynomial_1 + \spolynomial_2$.

 \begin{definition}[Digital Signature Scheme]
     \ldef{signatures}
     A digital signature scheme consists of a tuple of polynomial-time algorithms $(\sSigSetup, \sSigSign, \sSigVerify)$ defined as follows:
     \begin{algos}
         \item $\sSigSetup(1^\lambda) \rightarrow (\pk, \sk)$: Given a security parameter $\lambda$, it outputs a public key $\pk$ and a secret key $\sk$.

         \item $\sSigSign(\sk, m) \rightarrow \sigma$: Given a secret key $\sk$ and a message $m$, it outputs a signature $\sigma$.

         \item $\sSigVerify(\pk, m, \sigma) \rightarrow \{0,1\}$: Given a public key $\pk$, a message $m$, and a signature $\sigma$, it outputs $1$ if the signature is valid, otherwise $0$.

         \item $\sSigDistSign(\ssIdeal{\sk}, m) \rightarrow \ssIdeal{\sigma}$: Given a secret key share $\ssIdeal{\sk}$ and a message $m$, it outputs a signature share $\ssIdeal{\sigma}$. 
     \end{algos}

     \begin{algos}
         \item \fakeparagraph{Correctness.}
         For all $m$, and for $(\pk, \sk) \gets \sSigSetup(1^\lambda)$, if $\sigma = \sSigSign(\sk, m)$, then $\sSigVerify(\pk, m, \sigma) = 1$.
         \item \fakeparagraph{Unforgeability.}
         For all polynomial-time adversaries $\mathcal{A}$, 
         \begin{equation*}
                 \Pr\left[
                 \sSigVerify(\pk, m, \sigma) = 1 \,\middle|\, \begin{split}
                 (\pk, \sk) &\leftarrow \sSigSetup(1^\lambda) \\
                 (\sigma, m) &\leftarrow \mathcal{A}(\pk) \end{split}
                 \right]
         \end{equation*}
         is negligible, where $\mathcal{A}$ has not received $\sigma$ from a prior invocation of $\sSigSign(\sk, m)$.
     \end{algos}

 \end{definition}

 \begin{definition}[Proof-of-Training]
     \ldef{zkp:training}
     A valid Proof-of-Training is an interaction between a Prover protocol \sprover and Verifier protocol \sverifier.
     A public learning algorithm $\sAlgorithm$ (including hyperparameters), takes a training dataset $\strainset$
     and training randomness \strainrand as input and outputs a model $\smodel \leftarrow \sAlgorithm(\strainset, \strainrand)$.
     A Proof-of-Training is defined as a set of algorithms (\sPotSetup, \sPotProve, \sPotVerify) where:
     \begin{algos}
         \item $\sPotSetup(1^\lambda) \rightarrow \pppot$: A setup algorithm that outputs the public parameters $\texttt{pp}$ .
         \item $\sPotProve(\pppot, \scommitment, \strainset, \strainrand, \smodel, \scomRandomnessData, \scomRandomnessTrainrand, \scomRandomnessModel) \rightarrow \pi$:
         The prover generates a proof $\pi$ that $\smodel$ is computed as $\smodel \leftarrow \sAlgorithm(\strainset, \strainrand)$, and \scommitment{} is a commitment to $\strainset, \strainrand, \smodel$ under respective randomnesses $\scomRandomnessData, \scomRandomnessTrainrand, \scomRandomnessModel$. 
         \item $\sPotVerify(\pppot, \scommitment, \pi) \rightarrow \{ 0, 1 \}$:
         The verifier accepts if the proof $\pi$ is valid with respect to the commitment $\scommitment$ to the data,
         the training randomness and the model.
     \end{algos}

    \begin{algos}
         \item \fakeparagraph{Completeness.}
         For a security parameter $\lambda$, for all $\strainset, \strainrand, \smodel, \scomRandomnessData, \scomRandomnessTrainrand, \scomRandomnessModel$, $\pppot \gets \sPotSetup(1^\lambda)$, $\scommitment \gets \sPoCCommit(\pppot, \strainset, \strainrand, \smodel, \scomRandomnessData, \scomRandomnessTrainrand, \scomRandomnessModel)$, if $\pi \gets \sPotProve(\pppot, \scommitment, \strainset, \strainrand, \smodel, \scomRandomnessData, \scomRandomnessTrainrand, \scomRandomnessModel)$, then $\sPotVerify(\pppot, \scommitment, \pi) = 1$.
         \item \fakeparagraph{Soundness.} 
         The probability that any polynomial-time $\mathcal{A}$ outputs an accepting proof $\pi$ and either $\smodel$ was not generated by $\sAlgorithm(\strainset, \strainrand)$ or $c$ is not a valid commitment is negligible.
         \item \fakeparagraph{Zero-Knowledge.}
         For every verifier \sverifier there exists a \PPT simulator $\ssim$, which, when interacting with \sverifier and given inputs $\pppot$, its corresponding simulation trapdoor, $\pi$, and $c$, produces a computationally indistinguishable view from an interaction with \sprover. 
    \end{algos}
The definitions above straightforwardly extend to handle multiple training sets $\strainset_i$ with some minor syntactic changes. 
 \end{definition}

  \begin{definition}[Proof-of-Inference]
     \ldef{zkp:inference}
     A valid Proof-of-Inference is an interaction between a Prover protocol \sprover and Verifier protocol \sverifier.
     A Proof-of-Inference is defined as a set of algorithms (\sPoiSetup, \sPoiProve, \sPoiVerify) where:
     \begin{algos}
         \item $\sPoiSetup(1^\lambda) \rightarrow \pppoi$: A setup algorithm that outputs the public parameters $\texttt{pp}$ .
         \item $\sproofInference \leftarrow \sPoiProve(\pppoi, \scommitmentModel, \scommitmentX, \scommitmentY; \smodel, \sX, \sY, \scomRandomnessModel, \scomRandomnessX, \scomRandomnessY)$: 
         The prover generates a proof $\pi$ that $\sY$ is computed as $\smodel(\sX)$, and $\scommitmentModel, \scommitmentX, \scommitmentY$ are commitments to $\smodel, \sX, \sY$ under respective randomnesses $\scomRandomnessModel, \scomRandomnessX, \scomRandomnessY$. 
         \item $\sPoiVerify(\pppoi, \scommitmentModel, \scommitmentX, \scommitmentY, \pi) \rightarrow \{ 0, 1 \}$:
         The verifier accepts if the proof $\pi$ is valid with respect to the commitments to the model, the inference sample and the inference result.
     \end{algos}

\noindent A Proof-of-Inference satisfies \textbf{Completeness}, \textbf{Soundness}, \textbf{Zero-Knowledge}, which are defined as for Proof-of-Training. 
 \end{definition}

\end{myhideenv}

\begin{myhideenv}

\section{Security Proof for Auditing Protocol}
\lsec{apx:proofs:composability}
In the following, we provide a proof that $\sprotocolcamel$ securely instantiates $\sidealfun$.
Our proof follows the real/ideal-world paradigm~\cite{Canetti2000-plainmodel}, which considers the following two worlds:
\begin{algos}
    \item \fakeparagraph{In the real world,} the parties run the $\sprotocolcamel$ protocol to establish
    an auditing framework for \gls{ppml}.
    The adversary \sadversary can statically, actively, corrupt a subset of parties before the start of the protocol.

    \item \fakeparagraph{In the ideal world,} the honest parties send their inputs to the ideal functionality $\sidealfun$,
    which executes the behavior of a secure auditing framework.
    The ideal world defines the ideal behavior of the functionality that the protocol aims to emulate.
\end{algos}
A real-world protocol is secure if it manages to instantiate an ideal functionality in the ideal world.
To show that a protocol is secure, we must show that the adversary cannot distinguish between the real and the ideal world with high probability.
We can do this by defining a non-uniform probabilistic polynomial-time simulator $\ssim$ that interacts with the adversary \sadversary and the ideal functionality $\sidealfun$ in the ideal world
in such a way that \sadversary's view is indistinguishable when interacting with the protocol in the real world.

We model generic MPC operations in our framework as $\sabb$ to ensure our framework composes with any MPC protocol
and assume a secure randomness source $\sidealrand$.
We also model our assumption of a \acrlong{pki} with $\sidealpki$ that, upon initialization, sends each party its signing key \sk\xspace and the list of verification keys $\pk_i$ for the other parties in the system.
According to the sequential composition theorem, if a protocol securely computes a functionality in the $\sideal_\textsf{G}$-hybrid model for some functionality $\sideal_\textsf{G}$,
then it remains secure when composed with a protocol that securely computes $\sideal_\textsf{G}$~\cite{Canetti2000-plainmodel}.
We use this model to abstract the dependencies of our framework.
Additionally, we assume secure point-to-point communication channels and rely on an ideal broadcast channel $\sidealbroadcast$ to ensure all parties in the auditing phase receive the same commitments to achieve identifiable abort.
Finally, we allow the simulator to equivocate commitments it sent to the adversary with \sidealcrs
which outputs the trapdoor.

We are now ready to prove the security of our overall protocol (see~\rfig{protocol:camel}).
We recall our threat model, in which at least one computing party in each (non-plaintext) phase and at least one of the input-sending or output-receiving parties must be honest. 
Without loss of generality, we assume a one-to-one mapping between roles and parties.
For malicious parties, we already consider appropriate collusion, for honest parties it is straightforward to see that they learn only the union of their roles information and receive no other capabilities.

\begin{theorem}[Collaborative Auditing]
    \label{thm:camel:adv:collaborative}
    Given a set of $\sNumParties$ \glspl{r:inputparty} $\sPartyClient$,
    a set of $\sNumModelHolders$ \glspl{r:modelowner} $\sPartyModelHolder$,
    a set of $\sNumAuditors$ \glspl{r:auditrequester} $\sPartyAuditor$,
    a set of $\sNumServers$ \glspl{r:tcomputer} $\sPartyServer$,
    a set of $\sNumInferenceComputers$ \glspl{r:icomputer} $\sPartyInferenceComputer$,
    a set of $\sNumAuditComputers$ \glspl{r:acomputer} $\sPartyAuditComputer$,
    an adversary \sadversary who controls a set $\sPartyCorrupted = \{ \sPartyAll_i : i \in \Cset \}$
    where at least one of the input-sending or output-receiving parties is honest, i.e., $(\sPartyClient \cup \sPartyModelHolder \cup \sPartyAuditor) \setminus \sPartyCorrupted \neq \varnothing$,
    and (unless in the plaintext training setting) one $\sPartyServer$, 
    and (unless in the plaintext inference setting) one $\sPartyInferenceComputer$, 
    and one $\sPartyAuditComputer$ is honest.
    there exists a \PPT simulator $\ssim$ in the $(\sabb,\sidealrand,\sidealpki,\sidealcrs)$-hybrid model such that the distributions:
    \begin{equation*}
        \begin{gathered}
            \left\{ \ideal_{\sprotocolcamel, \ssim(\advAux),\sPartyCorrupted \cup \{ \sPartyAll_{\sPartyAuditor} \} } (\strainsetInputs, \spredictionInputs, \sauditInputs, \lambda) \right\}_{\strainsetInputs, \spredictionInputs, \sauditInputs, \advAux, \lambda} \\
            \approx \\
            \left\{ \real_{\sidealfun, \sadversary(\advAux), \sPartyCorrupted \cup \{ \sPartyAll_{\sPartyAuditor} \} } (\strainsetInputs, \spredictionInputs, \sauditInputs, \lambda) \right\}_{\strainsetInputs, \spredictionInputs, \sauditInputs, \advAux, \lambda}
        \end{gathered}
    \end{equation*}

    are computationally indistinguishable, where $\strainsetInputs$ is a list of training datasets for each \gls{r:inputparty},
    \spredictionInputs is a list of prediction feature vectors, \sauditInputs is a list of prediction feature vectors to audit,
    and $\advAux \in \{0,1 \}^*$ is an auxiliary input by the adversary to capture malicious strategy.

\end{theorem}

\begin{proof}
    We will define a simulator $\ssim$ through a series of subsequent modifications to the real
    execution, so that the views of $\sadversary$ in any two subsequent executions are computationally indistinguishable.
    Without loss of generality, we assume that if the simulator receives inconsistent values from some of the parties that should be the same according to the real protocol for training and inference, the simulator aborts.
    Similarly, should any signature verification fail, we assume the simulator aborts.
    During training and inference, the simulator simply forwards any aborts to the ideal functionality, during auditing, more care must be taken to achieve identifiable abort.
    We highlight arguments that are only relevant for the plaintext training setting in \cPlainTrain{olive}, and for the plaintext inference setting in \cPlainInf{blue}.

    \begin{hybrid}
        \item The view of $\sadversary$ in this hybrid is distributed exactly as the view of $\sadversary$ in $\real$.
        \item In this hybrid, the real execution is emulated by a simulator that knows the real inputs of the honest parties $\strainset_i$ for $i \notin \Cset$
        and runs a full execution of the protocol with $\sadversary$, which includes emulating the ideal interactions for training and inference and the auditing interactions.
        The view of the adversary in this hybrid is the same as the previous one.

        \item In this hybrid, we replace the commitments to the inputs in \rstep{train:inputparty} (training input phase) with dummy data.
        \cPlainTrain{unless we are in the case of plaintext training, where \ssim can simply forward the honest inputs sent by the ideal functionality to the \sadversary.}
        The simulator generates to all-zero training sets and associated commitments and decommitments for any honest $\sPartyClient_i$ and uses these as inputs for $\sabb$ and \sPoCEval.
        The view of the adversary in this hybrid is the same as the previous one because of the hiding property of the commitments and the zero-knowledge property of \sPoCEval.

        \item \label{proof:train:input} In this hybrid, the simulator replaces the \rstep{train:tcomputer} (model computation) with the result from the ideal functionality. 
        In the non-plaintext setting, \ssim can extract the training set $\strainset_i$ for any potential malicious $\sPartyClient_i$ via $\sabb$ and input these to \sidealfun to receive $\sidealmodelid$.    
        If \sadversary controls a \sPartyModelHolder, \ssim will also receive  $\smodel$ and can forward it to \sadversary.
        If the adversary is not involved in training at all, \ssim still receives \sidealmodelid from the ideal functionality.
        \ssim knows the commitment randomness from $\sabb$ \cPlainTrain{(or \sidealrand)} and can use this to either commit to the training randomness and either the actual model (if known) or an all-zero dummy model.           
        \ssim stores the tuple $(\sidealmodelid, \scommitment)$ internally,
        communicates \scommitment to the adversary (corresponding to the out-of-band communication in the real-world),
        and then emulates the last three steps of \rstep{train:tcomputer} with the above commitments and data.
        
        \cPlainTrain{In the plaintext setting,
        if there is no honest \gls{r:tcomputer}, \ssim needs to emulate the protocol and extract the model and (malicious) training data from \sproofTrain, which it can then use as inputs for the ideal functionality.
        Alternatively, if there is an honest \gls{r:tcomputer}, \ssim receives the malicious $(\strainset_i, \scomRandomnessDataI)$ directly from \sadversary and can input these to \sidealfun to receive $\sidealmodelid$.
        Note that, for any potential dishonest \glspl{r:tcomputer}, \ssim will receive the honest inputs from the ideal functionality, which it is free to forward in this setting.
        If there are any malicious \glspl{r:modelowner} or \glspl{r:tcomputer}, \ssim will also receive \smodel, generate a commitment to the model \scommitment and store the tuple $(\sidealmodelid, \scommitment)$ internally. 
        It communicates \scommitment to the adversary (corresponding to the out-of-band communication in the real-world).
        It can then emulate the remaining steps of \rstep{train:tcomputer}, unless the only malicious parties are the input parties (which are not involved in the remainder of \rstep{train:tcomputer}).}

        The replacement of \rstep{train:ip:sign} and \rstep{train:modelowner} follow trivially from the security of the underlying signature schemes and commitments.

        \item \lstep{proof:camel:inference} In this hybrid, we adapt Step~\ref{step:camel:inf:input:client} 
        (inference input) between the simulator and the adversary.
        As \sidealfun only allows audits of predictions that were actually made in the inference phase, \ssim has to ensure that the internal state of \sidealfun matches with that of the real protocol.
        However, as there might not always be an honest party present in each stage of the protocol, the adversary can, in some scenarios, generate valid predictions locally.
        We consider the four different scenarios that \ssim must handle:
        \begin{algos}
            \item \fakeparagraph{Honest \gls{r:auditrequester} and \gls{r:modelowner}:} 
            In this case, the (honest) request must be for a previously trained model, and we need to consider the case where at least one \gls{r:icomputer} is malicious. 
             In the non-plaintext setting, \ssim can either input and commit to an all-zero model or re-use a potential existing commitment to an all-zero model (if the adversary was involved in the training for this model id). 
            \cPlainInf{In the plaintext inference setting, the simulator always learns the actual model. Note that the adversary might already hold a commitment that should correspond to this model,  but instead is a commitment to an all-zero model.
            Using the trapdoor from $\sidealcrs$, \ssim can use equivocation to arrive at a decomitmment randomness that matches the commitment with the actual model.}          
            \item \fakeparagraph{Corrupt \gls{r:auditrequester} and honest \gls{r:modelowner}:}
            Either,  \ssim $\sX$ and \scommitmentModel from $\sadversary$, or (if $\sadversary$ controls all \gls{r:icomputer} and skips~\rstep{camel:inf:input:client}), \ssim receives a request for $\scommitmentModel$.
            The simulator finds the model information $(\sidealmodelid, \scommitment)$  from its internal storage corresponding to \scommitmentModel and aborts if it cannot find it. Note that, in the real protocol, an honest model holder also aborts if it does not have a corresponding model, and \ssim internal storage captures the view of an honest model holder.
            \ssim then forwards $(\sPartyModelHolder_k, \sidealmodelid, \sX)$ to \sidealfun to receive $y$ which it uses to simulate \sabb.
            \cPlainInf{In the plaintext setting: \sidealfun sends the model to \ssim and uses this for the \glspl{r:icomputer}.}
            The view is indistinguishable from the previous hybrid, because the simulator uses the prediction from the real model $\sY$ from $\sidealfun$, and because the \sidealfun will only abort if \sPartyAuditor sent a \scommitmentModel that does not exist which corresponds to the behavior in the real protocol.
            \item \fakeparagraph{Honest \gls{r:auditrequester} and corrupt \gls{r:modelowner}:}
            The simulator receives $\smodel, \scommitment$ from \sadversary.
            \ssim verifies that all training signatures $\ssigTrain$ and the \gls{r:tcomputer} signature \ssigTrainingComputer are valid signatures with respect to the commitments and otherwise aborts \sidealfun.
            \ssim also receives $\sidealmodelid$ from \sidealfun and finds the model information $(\sidealmodelid, \scommitment)$ from its internal storage corresponding to \scommitmentModel.
            The validity of \ssigTrainingComputer guarantees $(\sidealmodelid, \scommitment)$ exists in internal storage, because there is always at least one honest \gls{r:tcomputer}.
            \cPlainTrain{In the plaintext training setting, the \sadversary could have sent a different model \smodel that does not correspond to \sidealmodelid, because it has locally generated additional training runs.
            Therefore, \ssim now has to make sure that \sidealfun has the necessary internal state to send the right prediction to the honest client.
            \sidealfun will ask in~\rstep{camel:ideal:inf:askadv} for a model and training data from \ssim.
            \ssim responds with the training datasets and the model using the extractor guaranteed by knowledge-soundness of the proof of training \sproofTrain.
            }
            \cPlainInf{The adjustments for plaintext inference are the same as in previous cases.}

            \item \fakeparagraph{Corrupt \gls{r:auditrequester} and corrupt \gls{r:modelowner}:} In the non-plaintext setting, the simulator receives both the model and the sample from $\sadversary$ through $\sabb$ which is sufficient to emulate the protocol. \cPlainInf{Note that, in the case of plaintext inference, the \sadversary can generate predictions locally using \gls{r:icomputer}, which does not influence this Hybrid, but will become relevant later.}
            
        \end{algos}

        \item In this hybrid, we adapt Step~\ref{step:camel:inf:icomputer}, \ref{step:camel:inf:modelownerverify} and \ref{step:camel:inf:clientverify} (inference computation) between the simulator and the adversary.
        In the non-plaintext setting, as \ssim can simulate the operations through \sabb, we only discuss what happens when the \glspl{r:icomputer} inputs or opens values.
        If the \gls{r:auditrequester} is malicious, the simulator uses the $\sY$ received from \sidealfun to generate and open the commitments \scommitmentX and \scommitmentY.
        Otherwise, \ssim uses dummy inputs for \sX and \sY to generate commitments.
        The view in this scenario is indistinguishable from the view in the previous hybrid, because of the hiding property of the commitments.
        Note that \ssim can access the commitment randomness $\scomRandomnessX, \scomRandomnessY$ through simulating \sabbARand.
        If \sadversary controls the \gls{r:modelowner}, \ssim receives $\sidealmodelid$ from \sidealfun and sends $(\ssigReceipt, \scommitmentX^\prime, \scommitmentY^\prime)$ to \sadversary,
        where $\scommitmentX^\prime, \scommitmentY^\prime$ are two randomly sampled group elements.
        \cPlainInf{In the case of plaintext inference, $\smodel, \sX$ and $y$ are leaked to the \sadversary if it controls any \gls{r:icomputer} and, as a result, the simulator receives $\smodel, \sX$ and $y$ from $\sidealfun$, and can follow the protocol honestly on the inputs from the real protocol. If the adversary only controls the client or the model holder, the proof proceeds nearly identically to the non-plaintext case.}
        The replacement of \ref{step:camel:inf:modelownerverify} and \ref{step:camel:inf:clientverify} follow trivially from the security of the underlying signature schemes and commitments.

        \item \label{hybrid:camel:audit:verifyauditor} In this hybrid, we adapt the local checks of the audit verification phase in Step~\ref{step:camel:audit:clientinput} and~\ref{step:camel:audit:localchecks} to ensure \sidealfun receives the proper abort signal.
        A malicious \gls{r:auditrequester} sends $(\scommitment, \scommitmentX, \scommitmentY, \ssigReceipt, \ssigTrain, \ssigTrainingComputer \cPlainTrain{\text{ or }\sproofTrain}, \ssigInferenceComputer \cPlainInf{\text{ or } \sproofInference}, \pkMHK, \sauditfunction, \aux)$ to \ssim simulating \sidealbroadcast.
        The simulator checks that \pkMH is a valid identity from $\sidealpki$, that the signatures $\ssigTrain, \ssigReceipt$ are valid, that \ssigTrainingComputer \cPlainTrain{ or \sproofTrain} is valid, that \ssigInferenceComputer \cPlainInf{ or \sproofInference} is valid, and that $\sauditfunction \in \sauditfunctionSpace$.
        If the checks do not pass, $\ssim$ aborts $\sidealfun$ with $(\sidealabort, \sPartyAuditor)$.
        Otherwise, if a party $\sPartyAll$ controlled by $\sadversary$ aborts,
        \ssim forwards $(\sidealabort, \sPartyAll)$ to \sidealfun.
        This strategy for \ssim is valid because the broadcast channel allows the simulator to verify whether the corrupted \glspl{r:acomputer} have aborted rightfully or not.
        If the corrupted \gls{r:acomputer} aborts while the simulator's checks pass, the party must be malicious.
        Hence, the view is indistinguishable from the previous hybrid.

        \item In this hybrid, we adapt the \gls{zkpoc} of the \gls{r:auditrequester} (\rstep{camel:audit:verifyauditor}), \gls{r:modelowner} (\rstep{camel:audit:verifyinference}) and \glspl{r:inputparty} (\rstep{camel:audit:verifytraining}).
        Each corrupted party $\sPartyAll$ sends its inputs to $\sabbid$ and broadcasts a proof of consistency to \ssim through \sidealbroadcast as part of \sPoCEvalID.
        \ssim checks that the proof is valid with respect to the commitments \scommitment and, if not, aborts $\sidealfun$ with $(\sidealabort, \sPartyAll)$.
        The view is indistinguishable from the previous hybrid, because \sidealfun only sends identifiable aborts to \gls{r:acomputer}. There, this behavior is consistent with the real protocol.

        \item In this hybrid, we adapt the simulator to use dummy inputs in Steps~\ref{step:camel:audit:verifyauditor}, \ref{step:camel:audit:verifyinference} and~\ref{step:camel:audit:verifytraining}.
        \ssim uses all zeroes for all honest audit inputs as input to \sPoCEvalID~ in all steps.
        The view is indistinguishable from the previous hybrid, due to the zero-knowledge property of \sPoCLabel.
        Note that the simulator still uses the output provided by the real parties in~\rstep{camel:audit:function}, which we address in Hybrid~\ref{proof:audit:output}.

        \item \label{proof:audit:output} In this hybrid, we replace the output of the auditing function computation in~\rstep{camel:audit:function} with the output of \sidealfun.
        We begin discussing the setting where neither training nor inference were in the plaintext setting.
        In case \gls{r:auditrequester} is honest, the \sidealfun responds with the correct result and \sadversary receives no output.
        In case the \gls{r:auditrequester} is dishonest, \ssim has to ensure \sadversary receives the correct output.
        The simulator finds the model information $(\sidealmodelid, \scommitment)$ from its internal storage corresponding to \scommitmentModel.
        \ssim forwards $(\sidealaudittrigger, \sPartyModelHolder_j, \sidealmodelid, \sauditfunction, \spredictionX, \spredictionY, \aux)$ to \sidealfun to receive the real audit function output \sauditfunctionoutput.
        We now discuss why \sidealmodelid exists in \ssim's internal storage.
        Due to the validity of \ssigTrainingComputer, \ssim is guaranteed to find \sidealmodelid, because at least one of the \glspl{r:icomputer} is honest.
        In addition, \sidealfun will not abort because it has internally stored the $(\sPartyModelHolder_k, \sidealmodelid, \sX, \sY)$ because of the validity of \ssigInferenceComputer, which implies that the simulator has been involved in the prediction through an honest \gls{r:icomputer}.
        However, in the plaintext setting, two additional cases may occur:
        \begin{itemize}[wide=0pt]
            \item \fakeparagraph{In the case of \cPlainTrain{plaintext training} and \cPlainInf{plaintext inference}:}
            \ssim may not find a commitment while all proofs and signatures are valid, because \sadversary may have locally computed an extra training and inference.
            In this case, \ssim can extract all the auditing inputs from the proof of training \sproofTrain and the proof of inference \sproofInference.
            If either \sproofTrain or \sproofInference is invalid, the simulator can abort the malicious \gls{r:auditrequester} and \sidealfun, because
            \sidealfun only allows predictions from models that resulted from training on the \glspl{r:inputparty}' datasets.
            If both proofs are valid and \ssim can successfully extract the auditing inputs from the proofs, \ssim can then use these values to compute \sauditfunctionoutput locally (without \sidealfun), as no honest parties need to send inputs or receive output.
            In the case of an abort by the \glspl{r:acomputer}, \ssim simulates an abort at \sidealfun by sending an audit request with $\sauditfunction \notin \sauditfunctionSpace$.
            \item \fakeparagraph{In the case of \cPlainInf{plaintext inference}:}
            \sidealfun may not have the requested prediction stored in its internal state.
            In this case, \ssim can still forward the audit request from the malicious \gls{r:auditrequester} because
            the functionality does not require to have issued an inference in the case that the \sadversary controls the \gls{r:modelowner} and all \gls{r:icomputer} (\cfref{\rstep{camel:audit:exception_plaintextinf}}).
            If \sproofInference is invalid, the simulator can abort the malicious \gls{r:auditrequester} and \sidealfun, because
            \sidealfun only allows valid predictions from models.
        \end{itemize}
        The view is indistinguishable from the previous hybrid, because it uses the honest output from the ideal functionality.

        \item This hybrid is defined as the previous one, with the only difference being that the simulator now does not receive the inputs of the honest parties.
        Because the simulator no longer relied on receiving inputs from the honest parties, the view of the adversary is perfectly indistinguishable from the previous hybrid.

    \end{hybrid}
\end{proof}

\end{myhideenv}

\section{Consistency Check}
\lsec{apx:proofs:consistency}

In the following, we proof that our approach ($\sprotocolCC$) fulfills the properties for a \acrlong{poc} (\rdef{poc}), which are defined below.
We then provide formal definitions for the strawman constructions discussed in~\rsec{design:consistency:approaches}. 
In addition, we briefly discuss how to efficiently realize $\sabbec$ for Pedersen Vector Commitments.

A valid Proof-of-Consistency as in~\rdef{poc} satisfies the following properties:
    \begin{algos}
        \item \fakeparagraph{Correctness:}
        If \ssIdeal{\sInputs} is a valid input of \sInputs to $\sabb$ and \scommitment is a valid commitment to \sInputs computed as
        $\sPoCCommit(\pppoc, \sInputs, \sRandomness)$ for all public parameters $\pppoc \leftarrow \sPoCSetup(\smash{1^\lambda}, \sNumSamples)$ and randomness $\sRandomness \smash{\sample} \scomRandomnessSpace$,
        then $\sPoCEval(\pppoc, \scommitment, \ssIdeal{\sInputs}; \sInputs, \sRandomness) = 1$ with overwhelming probability.
        \item \fakeparagraph{Soundness:}
        If there exists no \sInputs such that $\sabb$ holds $\ssIdeal{\sInputs}$ and $\scommitment = \sPoCCommit(\pppoc, \sInputs, \sRandomness)$ for all public parameters \pppoc and for randomness $\sRandomness \smash{\sample} \scomRandomnessSpace$,
        then for all $\lambda \in \mathbb{N}$, and for all polynomial-time adversaries \sadversary on input $(\advAux, \ssIdeal{\sInputs}, \sRandomness, \scommitment))$, the probability that $\sPoCEval(\pppoc, \scommitment, \ssIdeal{\sInputs}; \sInputs, \sRandomness) = 1$ is negligible in $\lambda$.
        \item \fakeparagraph{Zero-knowledge:}
        For every probabilistic polynomial-time interactive machine $\sverifier^\prime$ that plays the role of the verifiers,
        there exists a probabilistic polynomial-time algorithm \ssim such that for any $\ssIdeal{\sInputs}$, randomness \sRandomness and $\scommitment = \sPoCCommit(\pppoc, \sInputs, \sRandomness)$
        the transcript of the protocol between $\sverifier^\prime$ and \sprover and the output of \ssim on input (\ssIdeal{\sInputs}, \sRandomness, \scommitment) are computationally indistinguishable.
    \end{algos}

\noindent
We now prove that our protocol for input consistency satisfies these properties.
\begin{lemma}
    \llem{cc}
    $\sprotocolCC$ is a \acrlong{poc} (\rdef{poc}).
\end{lemma}
\begin{proof}
    The protocol $\sprotocolCC$ in Protocol~\ref{protocol:cc} satisfies the properties of a \gls{poc}:
    \begin{algos}
        \item \fakeparagraph{Completeness:}
        From the correctness of the MPC protocol, it holds that $\rho = \sProtocolRandomness + \sum_{i=1}^{\sNumSamples} \sinputSample_i \cdot \beta^i$.
        Further, the opening proof of the polynomial commitment $(\scommitment \cdot \scommitment_\sProtocolRandomness)$ also evaluates
        to $\rho$ at $\beta$ due to the homomorphic property of the scheme.
        The verifiers accept because of the completeness of the polynomial commitment scheme.
        \item \fakeparagraph{Soundness:}
        Let \sProtocolRandomness be a random value, $\hat{\spolynomial} = \spolynomial + \spolynomial_\sProtocolRandomness$ be the polynomial defined as in the protocol as
        $\hat{\spolynomial}(z) = \sProtocolRandomness + \sum_{i=1}^{\sNumSamples} x_i \cdot z^i$,
        let $\ssGeneric{\sProtocolRandomness}$ be a secret-sharing of $\sProtocolRandomness$ and
        let $\ssGeneric{\sInputs}$ be a secret-sharing of \sInputs.
        If the verifiers do not hold a valid secret-sharing $\ssGeneric{\sProtocolRandomness}$ or $\ssGeneric{\sInputs}$,
        then the MPC protocol in~\rstep{cc:rho} aborts.
        Otherwise, the correctness of the MPC protocol guarantees that a valid secret-sharing of $\ssGeneric{\sProtocolRandomness}$ and $\ssGeneric{\sInputs}$ implies that
        $\rho$ equals $\sProtocolRandomness + \sum_{i=1}^{\sNumSamples} \sinputSample_i \cdot \beta^i$ and that $\hat{\spolynomial}(\beta) = \rho$ or the protocol aborts.
        Let $\scommitment^\prime$ be a polynomial commitment such that $\scommitment^\prime \neq \PCCommit(\texttt{pp}, \hat{\spolynomial}, \sRandomness + r_\sProtocolRandomness)$.
        Then, from the polynomial binding property of the polynomial commitment scheme, either $\scommitment^\prime$ is a commitment to a different polynomial $\spolynomial^\prime$ or
        the verifiers reject the proof in~\rstep{cc:verifiers} with overwhelming probability.
        In the case that $\scommitment^\prime$ is a commitment to a different polynomial $\spolynomial^\prime$, the verifiers only accept $\PCCheck(\pp, \scommitment^\prime, \beta, \rho, \pi)$
        if $\spolynomial^\prime$ agrees with $\hat{\spolynomial}$ at point $\beta$ because of the evaluation binding property of the polynomial commitment.
        Because $\beta$ was sampled uniformly at random, from the Demillo-Lipton-Schwartz-Zippel Lemma~\cite{Demillo1978-SZDL}, it holds that:
        \begin{equation*}
            \Pr\left[ \PCCheck(\pp, \scommitment^\prime, \beta, \rho, \pi) = 1 \right] \leq \frac{\sNumSamples}{p}\enspace,
        \end{equation*}
        which can be made negligible by choosing a suitably large $p$. Thus, the verifiers reject with overwhelming probability.

        \item \fakeparagraph{Zero-knowledge:}
        The simulator $\ssim$ works as follows: It samples $\sNumSamples$ random coefficients that define the polynomial $\spolynomial$
        and one random coefficient to define the polynomial $\spolynomial_\sProtocolRandomness$.
        Then, it samples a random point $\beta \smash{\sample} \sfield_p$ and runs the MPC simulator to produce the transcript for the computation of $\rho$.
        Finally, it samples $\sRandomness, \sRandomness_\sProtocolRandomness \smash{\sample} \sfield_p$ and
        computes $\scommitment = \PCCommit(\pp, \spolynomial, \sRandomness)$, $\scommitment_\sProtocolRandomness = \PCCommit(\pp, \spolynomial_\sProtocolRandomness, \sRandomness_\sProtocolRandomness)$
        and $\pi = \PCProve(\pp, \scommitment \cdot \scommitment_\sProtocolRandomness, \spolynomial + \spolynomial_\sProtocolRandomness, \sRandomness + \sRandomness_\sProtocolRandomness, \beta, \rho)$ and outputs $(\scommitment, \scommitment_\sProtocolRandomness, \pi, \rho, \beta)$.
        The indistinguishability with the real execution follows from the fact that $\rho$ is uniformly distributed in $\sfield_p$
        because of \sProtocolRandomness,
        the properties of the MPC protocol and the hiding property of the polynomial commitment scheme.
    \end{algos}

\end{proof}

\begin{myhideenv}
\section{Alternative Approaches}
\label{sec:alternativeapproaches}
\subsection{Direct Commitment}
\begin{protocol}[Strawman Consistency Check]
\label{protocol:cc:hash}
Let $\sabb$ be an instance of an ideal MPC functionality over a field $\sfield_p$,
let $\sidealrand$ be an ideal functionality that returns a random element from $\sfield_p$ and
let $\sinputSamples = (\sinputSample_1, \ldots, \sinputSample_{\sNumSamples}) \in \sfield_p^{\sNumSamples}$ by the input of prover $\sprover$.
Let $\ssIdeal{\sinputSamples} = (\ssIdeal{\sinputSample_1}, \ldots, \ssIdeal{\sinputSample_{\sNumSamples}})$
be the input of the prover $\sprover$ to $\sabb$.
Let $(\sComSetup,\sComCommit,\sComVerify)$ be a commitment scheme as in~\rdef{commitments}.
In order to improve the readability, we slightly abuse notation by allowing vector inputs to commitments, with the understanding that, where the commitment scheme does not accept vectors, we concatenate the elements in their natural order and, where necessary, apply a hash to match the input domain of the commitments.
The protocol \textsf{CC1} works as follows:
\begin{algos}
    \item $\textsf{CC1.Setup}(1^\lambda, \sNumSamples) \rightarrow \pp$: Run $\pp \leftarrow \sComSetup()$ where \sNumSamples is unused.
    \item $\textsf{CC1.Commit}(\pp, \sInputs, \sRandomness) \rightarrow \scommitment$:
    The prover computes a commitment $\scommitment \leftarrow \sComCommit(\texttt{pp}, \sInputs_1, \sRandomness_1)$ and outputs \scommitment.
    The prover outputs $\scommitment$.
    \item $\textsf{CC1.Check}(\pp, \scommitment, \ssIdeal{\sInputs}; \sInputs, \sRandomness) \rightarrow \{ 0, 1\}$:
    The protocol proceeds as follows:
    \begin{enumerate}
        \item \sprover inputs the randomness \sRandomness to \sabb.
        \item The parties invoke $\sabb$ to compute $\sComVerify(\pp, \ssIdeal{\scommitment^\prime}, \ssIdeal{\sInputs}, \ssIdeal{\sRandomness})$ and output the result.
    \end{enumerate}
\end{algos}

\end{protocol}

\subsection{Homomorphic Commitments}
The following is a formalization of Cerebro's security check.
In this protocol, we make explicit that the input randomness $\mathbf{\sRandomness}$ and the output $\mathbf{\scommitment}$ of \sComCommit are lists of randomness and commitments, respectively, by highlighting them in bold.

\begin{protocol}[Cerebro's Consistency Check~\cite{Zheng2021-tb}]
    \label{protocol:cc:cerebro}
    Let $\sabb$ be an instance of an ideal MPC functionality over a field $\sfield_p$,
    let $\sidealrand$ be an ideal functionality that returns a random element from $\sfield_p$ and
    let $\sinputSamples = (\sinputSample_1, \ldots, \sinputSample_{\sNumSamples}) \in \sfield_p^{\sNumSamples}$ by the input of prover $\sprover$.
    Let $\ssIdeal{\sinputSamples} = (\ssIdeal{\sinputSample_1}, \ldots, \ssIdeal{\sinputSample_{\sNumSamples}})$
    be the input of the prover $\sprover$ to $\sabb$.
    Let $(\sComSetup,\sComCommit,\sComVerify)$ be a commitment scheme as in~\rdef{commitments} that is also homomorphic~(\rdef{homomorphic_commitment}).
    The protocol \textsf{CC2} works as follows:
    \begin{algos}
        \item $\textsf{CC2.Setup}(1^\lambda, \sNumSamples) \rightarrow \pp$: Run $\pp \leftarrow \sComSetup()$ where \sNumSamples is unused.
        \item $\textsf{CC2.Commit}(\pp, \sInputs, \mathbf{\sRandomness}) \rightarrow \mathbf{\scommitment}$:
        The prover computes a list of $d$ commitments $\mathbf{\scommitment} \leftarrow \{ \sComCommit(\texttt{pp}, \sinputSample_1, \sRandomness_1), \ldots, \sComCommit(\texttt{pp}, \sinputSample_d, \sRandomness_d) \}$.
        The prover outputs $\scommitment$.
        \item $\textsf{CC2.Check}(\pp, \mathbf{\scommitment}, \ssIdeal{\sInputs}; \sInputs, \mathbf{\sRandomness}) \rightarrow \{ 0, 1\}$:
        The protocol proceeds as follows:
        \begin{enumerate}
            \item \sprover inputs the randomness $\mathbf{\sRandomness}$ to \sabb.
            \item The parties invoke $\sidealrand$ to obtain a random challenge $\beta \sample \sfield_p$.
            \item The parties invoke $\sabb$ to compute \mbox{$\ssIdeal{\tilde{x}} \defeq \sum_{i=1}^{\sNumSamples} \ssIdeal{\sinputSample_i} \cdot \beta^{i-1}$}
            and \mbox{$\ssIdeal{\tilde{r}} \defeq \sum_{i=1}^{\sNumSamples} \ssIdeal{\sRandomness_i} \cdot \beta^{i-1}$}
            \item The parties invoke $\sabb$ to compute $\ssIdeal{\scommitment^\prime} \leftarrow \sComCommit(\pp, \ssIdeal{\tilde{x}}, \ssIdeal{\tilde{r}})$
            and open $\scommitment^\prime$.
            \item Each verifier computes \mbox{$\tilde{\scommitment} = \sum_{i=1}^d \beta^{i-1} \cdot \scommitment_i$} and checks that $\tilde{\scommitment} = \scommitment^\prime$.
            If verification passes, they output $1$, otherwise $0$.
        \end{enumerate}
    \end{algos}

\end{protocol}

\end{myhideenv}

\begin{myhideenv}

\section{Share Conversion}
\lsec{share_conversion}
Our protocols require the computation domain of the MPC protocol to be the scalar field of the elliptic curve.
However, efficiently implementing these cryptographic operations requires using specific instantiations of MPC schemes.
Specifically, the consistency check protocol needs to be evaluated in $\sfield_p$ (with EC extensions, where appropriate),
where $p$ is a large prime. 
However, the same field-based MPC protocols are suboptimal for ML computations, which benefit significantly from ring-based arithmetic available in many MPC settings~\cite{Dalskov2021-rc,Koti2021-Swift,Barak2020-quantizedinference,Keller2022-quantizedtraining}.
Even where ring arithmetic is not available, \gls{mpc} computations can be realized with significantly smaller fields for the same security level (e.g., 128-bit modulus vs 256-bit modulus).
A naive approach that uses the same computation domain for both the consistency check and the \gls{ml} computation would introduce significant overhead to the underlying \gls{ml}.
While existing works have observed this issue, they have so far failed to address it.
Fortunately, we can convert a secret value $x$ in one arithmetic domain $\sring_{\smod}$ to another arithmetic domain $\sring_{\smod^\prime}$.
We do this by
decomposing $\ssA{x}$ into $\ell$ bits and recomposing the bits in $\sring_{\smod^\prime}$ where $\ell \ll \smod,\smod^\prime$. %
\begin{align*}
    (\ssBit{x_{\ell-1}},\ldots,\ssBit{x_0}) &\leftarrow \sdecomp_{\sring_{\smod}}(\ssA{x}) \\
    \ssB{x} &\leftarrow \scomp_{\sring_{\smod^\prime}}(\ssBit{x_{\ell-1}},\ldots,\ssBit{x_0})
\end{align*}
\noindent
Converting between arithmetic domains and binary domains is a common technique in advanced \gls{mpc} implementations to facilitate more efficient computation of non-linear functionality~\cite{Mohassel2018-gy,Rotaru2019-dabits,Escudero2020-hh,Catrina2010-kl}.
We can use these techniques to instantiate \sdecomp and \scomp.
If $\sring_{\smod}$ is a field $\sfield$ such that $2^{\ell + \kappa} < |\sfield|$ where $x \in [0, 2^\ell]$ and $\kappa$ is the security parameter.
Let $r$ be a random value such that $r = \sum_{i=0}^{\ell + \kappa} r_i \cdot 2^i \in [0, 2^{\ell+\kappa}]$ and let \ssField{r} be the sharing of $r$ in $\sfield$
and let $(r_{m-1},\ldots,r_0) \in \sring_2^\ell$ be the sharing of bits of $r$.
$r$ is typically generated in an offline phase, for instance using \glspl{edabit}~\cite{Escudero2020-hh}.
The parties can evaluate $\sdecomp(x)$ by locally computing $\ssA{\epsilon} \leftarrow \ssA{x} - \ssA{r}$ and opening $\epsilon$.
Note that $\epsilon$ statistically hides $x$ because the statistical distance between the distributions of $\epsilon$ and $r$ is negligible in $\kappa$.
The shares of the bits $(\ssBit{x_{\ell-1}},\ldots,\ssBit{x_0})$ can be computed by adding $(\ssBit{r_{\ell-1}},\ldots,\ssBit{r_0})$ to $(\epsilon_{\ell-1},\ldots,\epsilon_{0})$ using a binary adder.
The parties can compose the bits in $\sring_{\smod^\prime}$ using a second random value $r^\prime$ secret-shared in $\sring_{\smod^\prime}$ as \ssB{r^\prime} with secret-shared bits $(\ssBit{r^\prime_{\ell-1}},\ldots,\ssBit{r^\prime_0})$.
Parties compute the masked bits $\ssBit{\epsilon^\prime_i}$ by adding $\ssBit{\epsilon_i} + \ssBit{r^\prime_i}$ using a binary adder.
They then open the bits to construct $\epsilon^\prime = \sum_{i=0}^{\ell-1} \epsilon^\prime_i \cdot 2^i$ and get $\ssB{x}$ by locally computing $\epsilon^\prime - \ssB{r^\prime}$.

This process performs a logical conversion for a secret $x \in [0, 2^\ell]$,
    but we must take extra care when performing an arithmetic conversion that also supports negative numbers, i.e., $x \in [-2^{\ell-1}, 2^\ell]$.
These numbers may take more than $\ell$ bits to represent in binary, which requires additional steps to convert between the two domains.
Fortunately, we can solve this by shifting the number up by $2^{\ell-1}$ and shifting it back after conversion, because this ensures that the number is represented in $\ell$ bits.
The main overhead of the protocol comes from computing the binary addition circuit, which requires at most $(\ell + \log n) \cdot (n-1)$ AND gates for $n$ parties,
in addition to the cost of generating the \glspl{edabit}, which can be computed in a preprocessing phase.
The previous method works for general secret-sharing-based MPC protocols,
but some protocols allow a more efficient technique when $\sring_{\smod}$ is a ring~\cite{Mohassel2018-gy,Dalskov2021-rc}.
This technique is called share splitting and allows much more efficient computation of the bit decomposition by
utilizing the additional structure of the $\sring_{\smod}$ compared to prime fields.

\end{myhideenv}

\begin{myhideenv}

\ifdefined\isnotextended
\else
    \begin{table*}[h!]\centering
\begin{tabular}{ccccccc}
\toprule
\multirow{2}{*}{Dataset} & \multirow{2}{*}{Network} & \multirow{2}{*}{MPC} & \multirow{2}{*}{Training Time} & \multicolumn{3}{c}{Consistency Overhead} \\
\cmidrule{5-7}
& & & & \textbf{Ours} & \gls{a:sha3} & \gls{a:ped} \\
\midrule
\multirow{4}{*}{\gls{sc:adult}} & \multirow{2}{*}{LAN} & \multirow{1}{*}{SH} & 2m & 4.8s & 9.1m (114x) & 9.5m (119x)  \\

 &  & \multirow{1}{*}{MAL} & 14.7m & 14.4s & 16.8m (70x) & 10.3m (43x)  \\

 & \multirow{2}{*}{WAN} & \multirow{1}{*}{SH} & 5.9h & 24s & 19.6h (2932x) & 1.7d (6218x)  \\

 &  & \multirow{1}{*}{MAL} & 2.1d & 2.2m & 23.9h (646x) & 2.5d (1596x)  \\

\midrule\multirow{4}{*}{\gls{sc:mnist}} & \multirow{2}{*}{LAN} & \multirow{1}{*}{SH} & 2h & 1.3m & 2.9h (136x) & 22.5h (1074x)  \\

 &  & \multirow{1}{*}{MAL} & 14.6h & 3.9m & 5.6h (86x) & 22.6h (348x)  \\

 & \multirow{2}{*}{WAN} & \multirow{1}{*}{SH} & 4.3d & 2.8m & 2.3w (8293x) & 33.6w (121462x)  \\

 &  & \multirow{1}{*}{MAL} & 6.1w & 29.9m & 2.5w (849x) & 47.7w (16117x)  \\

\midrule\multirow{4}{*}{\gls{sc:cifar}} & \multirow{2}{*}{LAN} & \multirow{1}{*}{SH} & 15.2h & 4.2m & 9.8h (141x) & 1.1w (2777x)  \\

 &  & \multirow{1}{*}{MAL} & 5.3d & 13.3m & 19.1h (86x) & 1.3w (952x)  \\

 & \multirow{2}{*}{WAN} & \multirow{1}{*}{SH} & 4.6w & 8.7m & 7.5w (8758x) & 302.8w (351571x)  \\

 &  & \multirow{1}{*}{MAL} & 48.6w & 1.6h & 8.1w (842x) & 429.9w (44631x)  \\

\midrule\multirow{4}{*}{\gls{sc:qnli}} & \multirow{2}{*}{LAN} & \multirow{1}{*}{SH} & 5.4d & 11.1m & 1.1d (140x) & 16w (14589x)  \\

 &  & \multirow{1}{*}{MAL} & 4.2w & 28.9m & 2.9d (144x) & 38.2w (13291x)  \\

 & \multirow{1}{*}{WAN} & \multirow{1}{*}{SH} & 17.3w & 21.8m & 15.8w (7327x) & 6681w (3094988x)  \\

\bottomrule
\end{tabular}
\caption{Overhead of consistency approaches we evaluate relative to (extrapolated) end-to-end training. Multipliers in parentheses are slowdown over ours. Time is given in seconds (s), minutes (m), hours (h), days (d) and weeks (w).}
\ltab{e2e_training}
\end{table*}

\begin{figure}[t]
    \centering
    \begin{subfigure}[t]{0.45\columnwidth}
        \centering
        \includegraphics[width=1.2\textwidth]{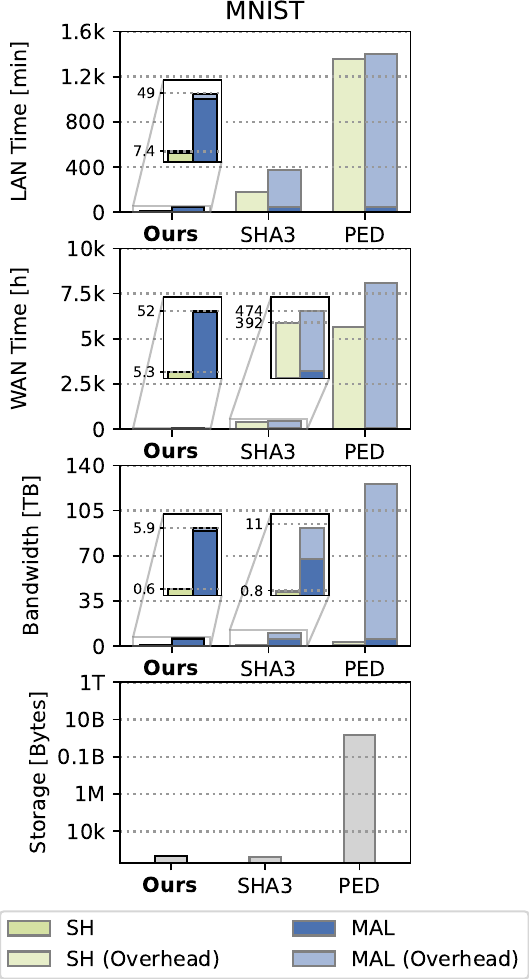}
        \caption{Training}
        \lfig{mnist:training}
    \end{subfigure}
    \hfill
    \begin{subfigure}[t]{0.45\columnwidth}
        \centering
        \raisebox{0.65cm}{\includegraphics[width=0.9\textwidth]{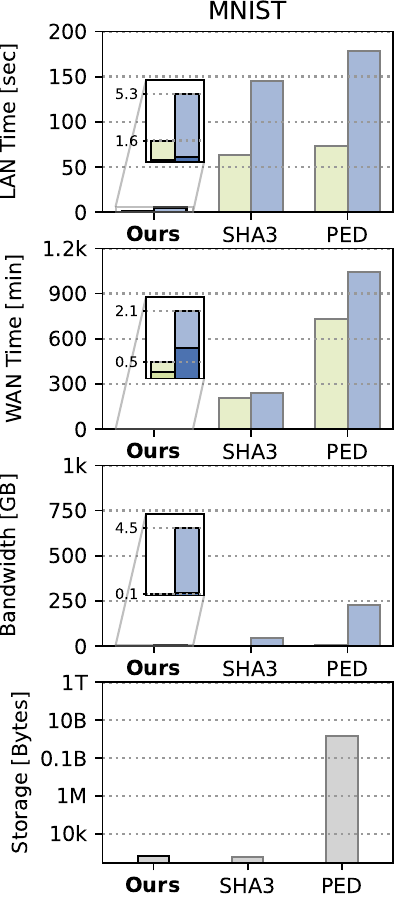}}
        \caption{Inference}
        \lfig{mnist:inference}
    \end{subfigure}
    \caption{The overhead of our system’s consistency protocol for \gls{sc:mnist}.}
    \lfig{mnist}
\end{figure}

\begin{figure}[t]
    \centering
    {\includegraphics[width=0.9\columnwidth]{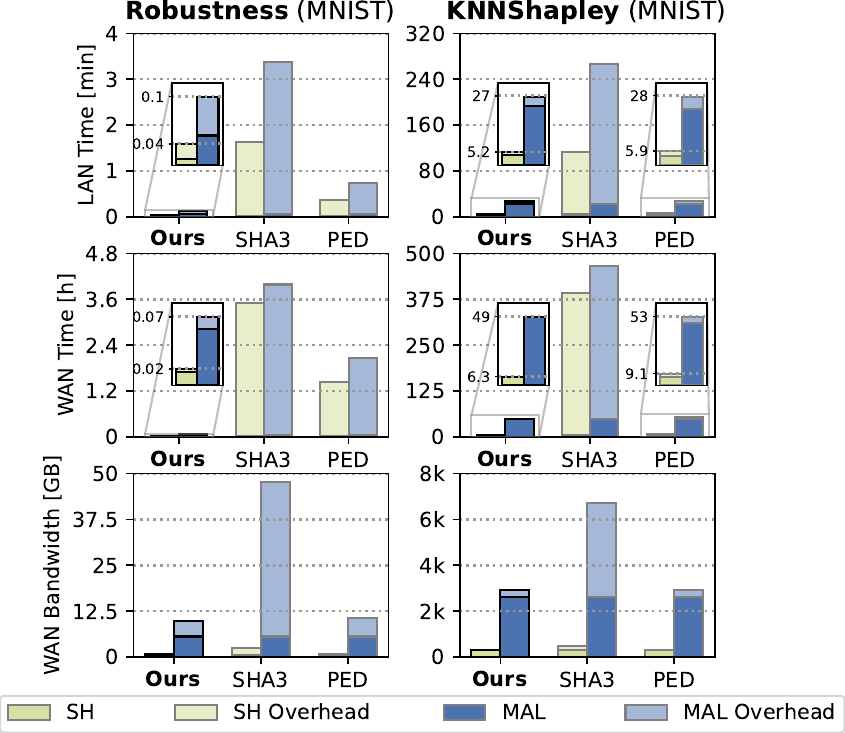}}
    \caption{The overhead of our system’s consistency protocol relative to the cost of the auditing function computation for \gls{sc:mnist}.}
    \lfig{mnist:auditing}
\end{figure}
\fi

\section{Additional Evaluation Results}
\lsec{apx:evalextra}
\raggedcolsend
We present evaluation results for end-to-end training and for \gls{sc:mnist} in this appendix.

\subsection{End-to-End Training}
We report the overhead of the consistency approaches compared to end-to-end training for \gls{sc:adult}, \gls{sc:mnist}, \gls{sc:cifar} and \gls{sc:qnli} in~\rtab{e2e_training}.
The overhead of our approach compared to training is at most \evalnum{0.04x} in the LAN setting and at most \evalnum{0.0011x} in the WAN setting.
This is because our approach largely consists of local computation as opposed to \gls{mpc} computation.

\subsection{MNIST Results}
The cost of a single epoch of \gls{ppml} training and inference for \gls{sc:mnist} compared to the overhead of the consistency approaches are shown in~\rfig{mnist}.
The overhead of our approach is less than \evalnum{1\%} in the \gls{wan} setting, which is in line with results for the other settings in~\rsec{eval}.
The storage overhead for the client for \gls{a:sha3} and our approach is 608 bytes and 720 bytes, respectively, which is the same as for the other scenarios.
The linear storage overhead of \gls{a:ped} results in a concrete storage requirement of \evalnum{1.5}GB.
\rFig{mnist:auditing} shows that our consistency layer has a small overhead compared to the execution of the auditing function.

\end{myhideenv}

\begin{myhideenv}

\section{Auditing Function Definitions}
\lsec{functions_extra}

\subsecspacingbot
\lsec{algorithms_functions}

\newcommand{\algorithmmargin}{}

{
\definecolor{blond}{rgb}{0.98, 0.94, 0.75}
\newcommand{\secret}[1]{{\setlength{\fboxsep}{1.5pt}\colorbox{blond}{#1}}}

\let\x\spredictionX
\renewcommand{\spredictionX}{\secret{$\x$}\xspace}

\let\y\spredictionY
\renewcommand{\spredictionY}{\secret{$\y$}\xspace}

\let\smodelOld\smodel
\renewcommand{\smodel}{\secret{$\smodelOld$}\xspace}

\newcommand{\spartyinputI}{\secret{$\spartyinput_1$}\xspace}
\newcommand{\spartyinputII}{\secret{$\spartyinput_2$}\xspace}
\newcommand{\spartyinputN}{\secret{$\spartyinput_N$}\xspace}

\newcommand{\spartyinputs}{\spartyinputI$, \ldots $\spartyinputN}

\renewcommand{\vec}[1]{\mathbf{#1}}

We provide a detailed description of the auditing functions presented in~\rsec{functions} in this appendix.
We discuss the key algorithmic components of each function and discuss the computational overheads associated with their secure evaluation in \gls{mpc}.
Each function receives, as secret inputs, a list of the sample \x, the prediction \y, the model $\smodelOld$ and the training datasets $\spartyinput_1,\ldots,\spartyinput_N$.
Additionally, each function takes a set of function-specific public auxiliary parameters.
Secret variables in the descriptions of the functions are highlighted in \secret{yellow},
    while unused secret inputs are shaded in light grey.

\subsection{Robustness \& Fairness}
\lsec{functions:apx:robustness_fairness}

The algorithm used to achieve post-hoc robustness and fairness guarantees is based on randomized smoothing,
which is a technique to certify robustness to adversarial perturbations that can be applied to any classifier~\cite{Cohen2019-cn}.
Jovanovic et al.~\cite{Jovanovic2022-tz} present an efficient \gls{fhe} variant of randomized smoothing to certify local robustness (and fairness) guarantees, which we adapt to the \gls{mpc} setting.
In particular, a model $\smodelOld$ is considered robust if it consistently outputs the same prediction $y$ within a radius $R$ around the input $x$,
where the distance is measured in the $p$-norm:
\begin{equation*}
   \label{eq:robustness}
   \forall \delta, \lVert \delta \rVert_p < R: \quad   \smodelOld(x + \delta) = y.
\end{equation*}

Randomized smoothing transforms a classifier into a \emph{smoothed} classifier that returns
the most probable prediction $y$ under noisy perturbations of $x$.
Since the prediction \spredictionY is already known at post-hoc auditing time,
our adaptation focuses solely on the second step of the algorithm, which certifies the statistical soundness of the prediction.
This process can be divided into two key steps:
In \textsc{SampleUnderNoise}, the algorithm samples $n$ perturbed inputs around $\spredictionX$ by adding noise sampled from a Gaussian distribution with
covariance matrix $\Sigma$,
and obtains predictions for each of these samples.
Afterwards, in \textsc{BinPValue}, a one-sided binomial test is conducted to determine whether the predicted label \spredictionY remains invariant across these perturbations with high probability.

\algorithmmargin
\begin{algorithm}
\begin{algorithmic}[1]
    \Function{CertifyRS}{\spredictionX, \spredictionY, \smodel, $\tau$, $\Sigma$, $n$, $\alpha$}
        \State \secret{$counts$} $\leftarrow$ \Call{SampleUnderNoise}{$\smodel$, $\spredictionX$, $n$, $\Sigma$}
        \State $pv$ $\leftarrow$ \Call{BinPValue}{\secret{$counts$}[\spredictionY], $n$, $\tau$}
        \State  \Return{$pv \leq \alpha$}
    \EndFunction
\end{algorithmic}
\end{algorithm}
\algorithmmargin
\noindent To certify that a classifier robustly predicts \spredictionY for an $\ell_2$ ball of radius $R$ around the input $\spredictionX$,
we sample noise from an isotropic Gaussian distribution, i.e., $\Sigma = \sigma \cdot I$.
In \textsc{BinPValue}, we then show that the probability for the classifier to output \spredictionY on the perturbed set is at least $\tau$ with a confidence level of $1-\alpha$,
where $\tau$ is the Gaussian CDF $\Phi$ of $R / \sigma$.

\algorithmmargin
\begin{algorithm}
\begin{algorithmic}[1]
    \Function{$f_\texttt{\gls{f:robustness}}$}{\spredictionX, \spredictionY, \smodel, {\color{gray}\spartyinputs}, $R$, $\sigma$, $n$, $\alpha$}
        \State $\tau$ $\leftarrow$ $\Phi(R / \sigma) \qquad \Sigma \leftarrow \sigma I$
        \State \Return \Call{CertifyRS}{$\spredictionX, \spredictionY, \smodel$, $\tau$, $\Sigma$, $n$, $\alpha$}
    \EndFunction
\end{algorithmic}
\end{algorithm}
\algorithmmargin
\noindent Evaluating the Gaussian-CDF is efficient because all of the parameters are public, allowing for precomputation.
The primary computational overhead arises from generating predictions for the perturbed samples, which requires $n$ model inferences.

A similar approach, based on randomized smoothing, can be applied to individual fairness,
leveraging the well-established link between robustness and individual fairness~\cite{Mukherjee2020-indivfairnessdata, Yurochkin2019-fairness, Ruoss2020-certifiedfairness}.
Specifically, we adopt the formalization of individual fairness by Dwork et al.~\cite{Dwork2011-fairness},
which defines fairness in terms of the Lipschitz continuity of a classifier $\smodelOld$ relative to a distance measure $\sfairDistance$ on the input space:
\begin{equation}
    \label{eq:ind-fairness}
    d_\sdomainY(\smodelOld(x), \smodelOld(x^\prime)) \leq L \cdot \sfairDistance(x, x^\prime) \quad \text{for all } x, x^\prime \in \sdomainX
\end{equation}
where $L > 0$ is a Lipschitz constant and $d_\sdomainY$ is a distance metric in the output space.
The fair metric encodes domain-specific intuition about which samples should be treated similarly by the \gls{ml} model.
Following previous work~\cite{Jovanovic2022-tz,Mukherjee2020-indivfairnessdata, Yurochkin2019-fairness}, we consider input space metrics of the form:
\begin{equation}
    \sfairDistance(x, x^\prime) = (x - x^\prime)^\intercal \Theta (x - x^\prime)
\end{equation}
where $\Theta$ is a symmetric positive definite matrix that encodes the domain-specific interpretation of similarity.
Each entry $\Theta_{i,j}$ in this matrix quantifies the difference between attributes $i$ and $j$ for which two individuals are considered similar.
The parameters $\Theta$ and $L$ should be set by a data regulator or domain expert based on the specific fairness requirements of the application~\cite{Dwork2011-fairness}.
Instead of sampling noise from an isotropic Gaussian, we now incorporate noise from a Gaussian with the covariance matrix $\Sigma = \Theta^{-1}$,
    while setting the prediction probability threshold $\tau = \Phi(\sqrt{1/L})$.
\algorithmmargin
\begin{algorithm}
\begin{algorithmic}[1]
    \Function{$f_\texttt{\gls{f:fairness}}$}{\spredictionX, \spredictionY, \smodel, {\color{gray}\spartyinputs}, $L$, $\Theta$, $n$, $\alpha$}
        \State $\tau$ $\leftarrow$ $\Phi(\sqrt{1 / L}) \qquad \Sigma \leftarrow \Theta^{-1}$
        \State \Return \Call{CertifyRS}{$\spredictionX, \spredictionY, \smodel$, $\tau$, $\Sigma$, $n$, $\alpha$}
    \EndFunction
\end{algorithmic}
\end{algorithm}
\algorithmmargin

\noindent The overhead of this function is similar to that of the robustness function, as the primary cost remains in generating predictions for the perturbed samples within \gls{mpc}.

\subsection{Accountability}
\lsec{functions:apx:accountability_party}
We consider two flavors of accountability: sample attribution~\cite{Koh2017-by} and party attribution~\cite{Lycklama2022-CamelMLSafety}.
The former identifies the influence of individual data samples on a prediction, while the latter attributes responsibility to a \gls{r:inputparty}.

\fakeparagraph{Sample-level Attribution.}
A variety of methods exists to identify the impact of individual data samples on a model, but not all of them are equally amendable to secure computation.
In \oursystem, we leverage a method based on \gls{f:knnshapley} values~\cite{Jia2019-ShapleyValue} which
computes the Shapley values of a KNN classifier on the training data's representation in the model's latent space.
These Shapley values have a closed-form solution that is relatively straightforward to compute, making them well-suited for \gls{mpc}.

{
\newcommand{\alphaI}{\secret{$\alpha_{1}$}}
\newcommand{\alphaII}{\secret{$\alpha_{2}$}}
\newcommand{\alphaN}{\secret{$\alpha_{|D|}$}}

\newcommand{\yalphaI}{\secret{$y_{\alpha_{1}}$}}
\newcommand{\yalphaII}{\secret{$y_{\alpha_{2}}$}}
\newcommand{\yalphaN}{\secret{$y_{\alpha_{|D|}}$}}

\newcommand{\salphaI}{\secret{$s_{\alpha_{1}}$}}
\newcommand{\salphaII}{\secret{$s_{\alpha_{2}}$}}
\newcommand{\salphaN}{\secret{$s_{\alpha_{|D|}}$}}

\newcommand{\lookup}{\secret{$Z$}}
\newcommand{\lookupi}[1]{\lookup[#1]}
The \gls{f:knnshapley} algorithm estimates the contribution of each training sample to a prediction by leveraging distances in latent space through a recursive process involving the $K$ nearest neighbors.
First, the algorithm computes the $L_2$ distance between each training data sample and the prediction $\spredictionX$ in
their latent representations, typically at the model's final layer before the output.
Next, the training samples are sorted in ascending order based on their proximity to $\spredictionX$,
resulting in a list of indices $(\alphaI, \ldots \alphaN)$.
With these indices, the Shapley values can be computed recursively in a single pass.
Each sample's Shapley value is computed based on the value of the previous sample and whether the sample's label matches the prediction $\spredictionY$.

    \begin{algorithm}
        \begin{algorithmic}[1]
            \Function{$f_\texttt{\gls{f:knnshapley}}$}{\spredictionX, \spredictionY, \smodel,\spartyinputs, $K$}
                \State \secret{$D$} $\gets \bigcup_{i=1\ldots N}$ \secret{$\spartyinput_i$}
                \State \secret{$dists$} $\gets$ \Call{LatentSpaceDistances}{\smodel, \spredictionX, \secret{$D$}}
                \State $(\alphaI, \ldots \alphaN) \gets$ \Call{SortedIndices}{\secret{$dists$}}

                \State \lookup $ \gets (\mathbbm{1}{\scriptstyle[\yalphaI = \spredictionY]}, \mathbbm{1}{\scriptstyle[\yalphaII = \spredictionY]}, \ldots \mathbbm{1}{\scriptstyle[\yalphaN = \spredictionY]})$

                \State $\salphaN \gets \frac{\lookupi{|D|}}{N}$

                \For{$ i \gets |D| - 1, 1$}
                    \State \secret{$s_{\alpha_i}$} $ \leftarrow $ \secret{$s_{\alpha_{i+1}}$}$ + (\lookupi{i} - \lookupi{i+1}) \cdot \frac{min(K, i)}{K \cdot i}$
                \EndFor

                \State \Return $[(\alphaI, \salphaI), \ldots (\alphaN, \salphaN)]$

            \EndFunction
        \end{algorithmic}
    \end{algorithm}
}

\noindent The primary overhead of $f_\texttt{\gls{f:knnshapley}}$ is computing the latent representations of the training data and the prediction sample;
however, the representations of the training data can be re-used across multiple audits.

\fakeparagraph{Party-level Attribution.}
Our system includes an algorithm designed to attribute responsibility to parties for model predictions without revealing information about individual data samples.
The function $f_\texttt{\gls{f:camel}}$ is based on the \gls{f:camel} algorithm~\cite{Lycklama2022-CamelMLSafety}.
The key idea of \gls{f:camel} is that if a suspicious prediction $(\spredictionX, \spredictionY)$
is influenced by data from a particular \gls{r:inputparty}, then removing that party’s data will weaken or eliminate the model's behavior for the suspicious prediction.
This approach approximates leave-out models by employing unlearning techniques to remove a party's data from the original model, $\smodel$.
To unlearn the data of party $i$, we use an efficient unlearning technique~\cite{Shan2021-ad}, where the labels of the party's data points are replaced with a uniform probability vector.
This represents the model's output when it is uncertain about its predictions~\cite{Vyas2018-jt, Lee2018-wr}.
After $E$ epochs of finetuning, the loss on the unlearned models, denoted as \secret{$\smodelOld_{-i}$}, becomes sufficient to detect outliers for the prediction $(\spredictionX, \spredictionY)$.
We then compute an influence score using the \gls{mad}, which is a robust dispersion measure.
Parties with a \gls{mad} score exceeding the threshold $\tau$ are flagged as potentially malicious.

\algorithmmargin
\begin{algorithm}[h!]
\begin{algorithmic}[1]
   \Function{$f_\texttt{\gls{f:camel}}$}{\spredictionX, \spredictionY, \smodel, \spartyinputs, $E, \tau$}

   \State \secret{$scores$} $\gets [\;]$
   \For{$i \gets 1, N$}
       \State \secret{$\spartyinput_{-i}$} $\gets \bigcup_{i \neq j}$ \secret{$\spartyinput_{j}$}
       \State \secret{$\smodelOld_{-i}$} $\gets$ \Call{UnlearnParty}{\smodel, \secret{$\spartyinput_i$}, \secret{$\spartyinput_{-i}$}, $E$}
       \State \secret{$S_i$} $\gets \ell(\smodel_{-i}(\spredictionX), \spredictionY)$ \Comment{compute loss}
       \State \secret{$scores$}$.append($\secret{$S_i$}$)$
   \EndFor
   \State \secret{$mads$} $\gets$ \Call{MADScores}{\secret{$scores$}}
   \State \secret{$outliers$} $\gets$ \Call{Outliers}{\secret{$mads$}, $\tau$} \Comment{data holder idx}
   \State \Return  \secret{$outliers$}
   \EndFunction
\end{algorithmic}
\end{algorithm}
\algorithmmargin

\noindent The primary overhead of $f_\texttt{\gls{f:camel}}$ is the unlearning,
    which requires $N \cdot E$ training epochs,
    where $N$ is the number of parties and $E$ represents the number of epochs needed to unlearn a party's dataset.
However, similar to the sample-level attribution function, this process is independent of the prediction sample \spredictionX and, therefore,
can be computed once and cached for multiple audits.
The majority of the remaining overhead of $f_\texttt{\gls{f:camel}}$ is
    $N$ forward passes of \spredictionX, which is significantly cheaper in comparison.

\subsection{Explainability}
\lsec{functions:apx:explainability}

Post-hoc model explanations explain a prediction \spredictionY, based on the features of the input \spredictionX.
In \oursystem, we adapt the \gls{f:shap}~\cite{Lundberg2017-unifiedmodelpredictions} method,
a post-hoc additive feature attribution method that offers local explanations for individual predictions
by computing Shapley values $\sfeature_0,\ldots,\sfeature_{n_x - 1}$ for the $n_x$ features.
\gls{f:shap} computes these by approximating the local decision boundary using a linear model
on a set of feature indicator vectors $\sproximity^\prime \in \{ 0, 1 \}^{n_x}$, i.e.,
\begin{equation*}
\label{eq:explainmodel}
\sexplainmodel(\sproximity^{\prime}) = \sfeature_0 + \sum_{i=1}^{n_x} \sfeature_i \sproximity_i^{\prime}
\end{equation*}
where $\sexplainmodel$ represents the explanation model.
The feature indicator vectors $\sproximity^\prime$ are used to sample variations of the input \spredictionX,
where a `1' in the vector indicates the presence of the corresponding feature from \spredictionX, and a `0' signifies its absence.

\algorithmmargin
    \begin{algorithm}[h]
        \begin{algorithmic}[1]
            \Function{$f_\texttt{\gls{f:shap}}$}{\spredictionX, \spredictionY, \smodel, \spartyinputs, $K$}

                \State \secret{$D$} $\gets \bigcup_{i=1\ldots N}$ \secret{$\spartyinput_i$} $\qquad \quad$ $n_{x} \gets$ \Call{NumFeatures}{\spredictionX}

                \State $\{\sproximity_1, \sproximity_2, \ldots \sproximity_K\} \gets \Call{SampleCoalitions}{K, n_{x}}$

                \State \secret{$\hat{\vec{y}}$} $\gets$ \Call{CreateMarginals}{\spredictionX, \spredictionY, \smodel, \secret{$D$}, $\vec{z}_1, \ldots \vec{z}_K$}

                \State $Z \gets [ \mathbf{1}^K  \; \vert \; \vec{z}_1^\intercal, \vec{z}_2^\intercal, \ldots \vec{z}_K^\intercal]$ \Comment{$Z \in \{0, 1 \}^{K \times n_x + 1}$}

                \State $W \gets diag(\pi(\vec{z}_1), \pi(\vec{z}_2), \ldots \pi(\vec{z}_K))$ \Comment{Weighting $\pi$}

                \State $[$\secret{$\phi_0$}, \ldots \secret{$\phi_{n_x}$}$]$ $\gets (Z^\intercal W  Z)^{-1} \; Z^\intercal W$ \secret{$\vec{\hat{y}}$} %

                \State \Return $[$\secret{$\phi_0$}, \ldots \secret{$\phi_{n_x}$}$]$

            \EndFunction

        \end{algorithmic}
    \end{algorithm}
    \algorithmmargin

The algorithm first generates the feature indicator vectors $\sproximity^\prime$ using the \textsc{SampleCoalitions} function.
For each sampled vector, the algorithm computes the model's prediction $\hat{y}$ on the corresponding sample, \secret{$x^\prime$}, represented by the indicator vector $\sproximity^\prime$.
This prediction step is carried out by the \textsc{CreateMarginals} function, which involves integrating over the marginal distributions of the features absent in \secret{$x^\prime$}.
The algorithm computes $\secret{$\hat{y}$}$ by averaging the predictions for the entire training dataset where each present feature, i.e., $\sproximity_i^\prime = 1$ is masked by the feature $\spredictionX_i$ from the prediction sample.
This results in a forward pass of the full training dataset for each sample $\sproximity^\prime$.
Finally, the algorithm fits the linear model $\sexplainmodel$ on the indicator vectors $\left\{ \sproximity^\prime \right\}$ using the Shapley weighting:

\begin{equation*}
    \pi(\vec{z}_k) = \frac{n_{x} - 1}{\binom{n_{x}}{|\vec{z}_k|} |\vec{z}_k| (n_{x}-|\vec{z}_k|)}.
\end{equation*}
By weighting the samples using this kernel, the parameters of the fitted linear model $\sexplainmodel$ correspond to the Shapley values of the features
at the decision boundary of the model around \spredictionX.

The function $f_\texttt{\gls{f:shap}}$ computes the linear model using its normal equation,
which is efficient under secure computation because only the prediction vector \secret{$\hat{\vec{y}}$} contains private data.
The primary computational overhead is computing the output of the model on the perturbed samples in \textsc{CreateMarginals}, which requires $K \cdot |~D~|$ inferences,
because of the need to integrate over the marginal distribution.

} %

\end{myhideenv}

\end{document}